\documentclass[11pt]{article}
\usepackage[utf8]{inputenc}
\usepackage{amsmath}
\usepackage{amssymb}
\usepackage{amsthm}
\usepackage{graphicx}
\usepackage{braket}
\usepackage{color}
\usepackage{hyperref}
\usepackage{qcircuit}
\usepackage{thmtools} 
\usepackage{thm-restate}
\usepackage{ytableau}

\usepackage[nameinlink]{cleveref}

\hypersetup{colorlinks={true},linkcolor={blue},citecolor=magenta}

\newtheorem{theorem}{Theorem}
\newtheorem{lemma}[theorem]{Lemma}
\newtheorem{proposition}[theorem]{Proposition}
\newtheorem{conjecture}{Conjecture}
\theoremstyle{definition}
\newtheorem{definition}{Definition}
\numberwithin{theorem}{section}
\numberwithin{lemma}{section}
\numberwithin{proposition}{section} %

\usepackage[margin=1in]{geometry}
\DeclareMathOperator{\Tr}{Tr}

\renewcommand{\cal}[1]{\mathcal{#1}}
\newcommand{\ketbra}[2]{|#1 \rangle \langle #2 |}

\newcommand{\BQP}{\mathsf{BQP}}

\newcommand{\QCMA}{\mathsf{QCMA}}
\newcommand{\QMA}{\mathsf{QMA}}
\newcommand{\QMAtwo}{\mathsf{QMA(2)}}
\DeclareMathOperator{\E}{\mathop{\mathbb{E}}}
\newcommand{\C}{\mathbb{C}}
\newcommand{\R}{\mathbb{R}}

\newcommand{\unitary}{\mathrm{U}}
\newcommand{\LU}{\mathrm{LU}}
\newcommand{\poly}{\mathrm{poly}}

\title{Unitary property testing lower bounds by polynomials}
\author{Adrian She\thanks{E-mail: \texttt{adrian.she@mail.utoronto.ca}} \\University of Toronto \and Henry Yuen\thanks{E-mail: \texttt{hyuen@cs.columbia.edu}}\\Columbia University}
\date{\today}

\begin{document}

\maketitle

\begin{abstract}
We study \emph{unitary property testing}, where a quantum algorithm is given query access to a black-box unitary and has to decide whether it satisfies some property. In addition to containing the standard quantum query complexity model (where the unitary encodes a binary string) as a special case, this model contains ``inherently quantum'' problems that have no classical analogue. Characterizing the query complexity of these problems requires new algorithmic techniques and lower bound methods.

Our main contribution is a generalized polynomial method for unitary property testing problems. By leveraging connections with invariant theory, we apply this method to obtain lower bounds on problems such as determining recurrence times of unitaries, approximating the dimension of a marked subspace, and approximating the entanglement entropy of a marked state. We also present a unitary property testing-based approach towards an oracle separation between $\mathsf{QMA}$ and $\mathsf{QMA(2)}$, a long standing question in quantum complexity theory. 
\end{abstract}

\newpage

\setcounter{tocdepth}{2}
\tableofcontents

\newpage

\section{Introduction}
\label{sec:intro}

The query model of quantum algorithms plays a central role in the theory of quantum computing. In this model, the algorithm queries (in superposition) bits of an unknown input string $X$, and after some number of queries decides whether $X$ satisfies a property $\cal{P}$ or not. We now have an extensive understanding of the query complexity of many problems; we refer the reader to Ambainis's survey~\cite{ambainis2018understanding} for an extensive list of examples. 

Although this query model involves quantum algorithms, the task being solved is \emph{classical property testing}, that is, deciding properties of classical strings. This has been very useful for comparing the performance of classical versus quantum algorithms for the same task. In contrast, \emph{quantum property testing} -- deciding properties of quantum objects such as states and unitaries -- has been been studied much less but has been receiving more attention in recent years~\cite{montanaro2013survey}.  

In this paper we focus on \emph{unitary property testing}, where the goal is to decide whether a unitary $U$ satisfies a property $\cal{P}$ by making as few queries to $U$ as possible. This systematic study of this topic was initiated by Wang~\cite{wang2011property}, and various aspects have been studied further in~\cite{montanaro2013survey,chen2022testing,aharonov2022quantum}. We continue explorations of this topic by developing a new lower bound technique, demonstrating its utility with several unitary property testing problems, and exploring intriguing connections between unitary property testing, invariant theory, and the complexity class $\QMAtwo$. Before presenting our findings in detail, we first explain the unitary property testing model.

\subsection{Unitary Property Testing}

The model of unitary property testing we consider is formally defined as follows. Fix a dimension $d$ and let $\cal{P}_{yes}, \cal{P}_{no}$ denote disjoint subsets (called \emph{yes} and \emph{no} instances respectively) of $d$-dimensional unitary operators. A \emph{tester} for deciding the problem $\cal{P} = (\cal{P}_{yes},\cal{P}_{no})$ is a quantum algorithm that, given query access to a unitary $U \in \cal{P}_{yes} \cup \cal{P}_{no}$ (called the \emph{problem instance}), accepts with high probability if $U \in \cal{P}_{yes}$ and otherwise accepts with low probability.\footnote{We note that the concept of property testing discussed in this paper is more general than typically presented in the literature (see, e.g.,~\cite{wang2011property,dall2021quantum}), where \emph{no} instances are defined to be $\epsilon$-far from the set of \emph{yes} instances for some distance measure. In this paper we allow for other ways of defining \emph{no} instances (as long as they are disjoint from \emph{yes} instances), which may be more natural in many contexts.}

This model includes the standard query model as a special case: a quantum query to a classical string $X \in \{0,1\}^d$ is defined to be a query to the unitary $U$ that maps $\ket{i}$ to $(-1)^{X_i} \ket{i}$ for all $i \in [d]$. In other words, the unitary is self-adjoint and diagonal in the standard basis.

We can go beyond self-adjoint, diagonal unitaries and study quantum analogues of classical property testing problems, such as:
\begin{itemize}
    \item Testing quantum juntas: if $U$ is an $n$-qubit unitary, determining whether there is a $k$-sized subset $S$ of qubits outside of which $U$ acts as the identity. This is analogous to determining whether the input $X$, interpreted as a function on $\{0,1\}^n$, only depends on a $k$-subset of coordinates. This was studied in~\cite{wang2011property,chen2022testing}.
    
    \item Approximate dimension: promised that $U$ applies a phase to all states $\ket{\psi}$ in a subspace $S$ of dimension either at $w$ or $2w$, determine the dimension of the subspace. This is analogous to the classical problem of approximating the Hamming weight of an input $X$. This was studied in~\cite{aaronson2020quantum}.
\end{itemize}
We can also study property testing problems that have no classical analogue at all, such as:
\begin{itemize}
    \item Unitary recurrence times: Determining whether $U^t = I$ or $\| U^t - I \| \geq \epsilon$ (promised that one is the case) where $t$ is a fixed integer.
    
    \item Hamiltonian properties: Promised that $U = e^{-iH}$ for some Hamiltonian $H$ with bounded spectral norm, determine properties of $H$, such as whether it is a sum of $k$-local terms, or the ground space is topologically ordered.
    
    \item Unitary subgroup testing: decide whether $U$ belongs to some fixed subgroup of the unitary group (such as the Clifford subgroup). This was studied in~\cite{brakerski2021unitary}.
    
    \item Entanglement entropy problem: Given  access to a unitary $U = I - 2 \ket{\psi} \bra{\psi}$ for some state $\ket{\psi} \in \C^d \otimes \C^d,$ decide if the entanglement entropy of the state $\ket{\psi}$ is low or high, promised one is the case.

\end{itemize}
These examples illustrate the rich variety of unitary property testing problems: some are motivated by well-studied classical problems in computer science (such as junta testing and approximate counting), whereas others are inspired by questions in quantum physics (e.g., identifying quantum chaos, topological order, or entanglement).

\subsection{A Generalized Polynomial Method and Its Applications}

Our main contribution is a lower bound technique for unitary property testing that generalizes the well-known polynomial method in complexity theory. First introduced by Nisan and Szegedy as~\cite{nisan1994degree} in classical complexity theory and then adapted for quantum algorithms by Beals, et al.~\cite{beals2001quantum}, the polynomial method is a powerful technique to lower bound the quantum query complexity of a variety of problems (see, e.g.,~\cite{aaronson2004quantum,bun2018polynomial}, and the references therein). 

The polynomial method is based on the fact that a quantum algorithm making $T$ queries to a boolean input $X = (x_1,\ldots,x_n)$ yields a real polynomial $p:\R^n \to \R$ of degree at most $2T$ such that $p(x_1,\ldots,x_n)$ is equal to the acceptance probability of the algorithm on input $X$. If the algorithm distinguishes between \emph{yes} and \emph{no} instances with some bias, so does the polynomial $p$. Thus, lower bounds on the degree of any such distinguishing polynomial directly translates into a lower bound on the bounded-error quantum query complexity for the same task.

We generalize this to arbitrary unitary properties. We note that we allow queries to both $U$ and $U^*$ as we are interested in lower bounds for the strongest possible query model. We leave it as an open question to see if there is a property testing problem that can be efficiently solved using oracle access to both $U$ and $U^*$ but not with access to only $U$. 

\begin{restatable}[Generalized polynomial method]{proposition}{genpoly}
\label{prop:basic-polynomial}
The acceptance probability of a quantum algorithm making $T$ queries to a $d \times d$ unitary $U$ and its inverse $U^*$ can be computed by a degree at most $2T$ self-adjoint\footnote{A self-adjoint polynomial is unchanged after complex conjugating every variable and every coefficient.} polynomial $p:\C^{2(d \times d)} \to \C$ evaluated at the matrix entries of $U$ and $U^*$. Thus, degree lower bounds on such polynomials yields a query lower bound on the algorithm.
\end{restatable}

Furthermore, we say that a unitary property $\mathcal{P} = (\mathcal{P}_{yes}, \mathcal{P}_{no})$ is closed under inversion if $U \in \mathcal{P}_{yes}$ iff $U^* \in \mathcal{P}_{yes},$ and $U \in \mathcal{P}_{no}$ iff $U^* \in \mathcal{P}_{no}.$ All properties we will study in this paper will be closed under inversion, and hence the polynomial $p$ satisfies a symmetry under this condition.

\begin{restatable}{proposition}{genpolyinversion}
\label{prop:basic-polynomial-inv}
Let $\mathcal{P}$ be an property closed under inversion and suppose there is a $T$-query quantum algorithm for testing property $\mathcal{P}$. Let $p$ be the polynomial from \Cref{prop:basic-polynomial} that computes the acceptance property of the algorithm. Then, we may assume that $p(U, U^*) = p(U^*, U).$
\end{restatable}

Hence, while establishing the existence of $p$ is straightforward, proving lower bounds on its degree is another matter. The standard approach in quantum query complexity is to \emph{symmetrize} $p$ to obtain a related polynomial $q$ whose degree is not too much larger than $p$, and acts on a much smaller number of variables (ideally a single variable). %
The choice of symmetrization method depends on the problem being analyzed.
For example, for (classical) properties that only depend on the Hamming weight of the boolean string (these are called \emph{symmetric} properties in the literature), the polynomial $p$ is averaged over all binary strings with Hamming weight $k$ in order to obtain a univariate polynomial $q(k)$ (this is known as \emph{Minsky-Papert symmetrization}~\cite{minsky2017perceptrons}). Lower bounds on the degree of $q$ can be then obtained by using Markov-Bernstein type inequalities from approximation theory.

\subsubsection{Lower Bounds for Unitarily Invariant Properties} 

To make \Cref{prop:basic-polynomial} useful, we develop symmetrization techniques for unitary properties that are invariant under certain symmetries. We first study \emph{unitarily invariant} properties\footnote{We acknowledge that the name ``unitarily invariant unitary property'' may seem redundant! The first ``unitarily'' refers to the symmetry; the second ``unitary'' refers to the type of property we are studying.}: these are properties $\cal{P} = (\cal{P}_{yes},\cal{P}_{no})$ such that conjugating an instance in $\cal{P}_{yes} \cup \cal{P}_{no}$ by any $g$ in the unitary group $\unitary(d)$ does not change whether it is a \emph{yes} or \emph{no} instance (in other words $g \cal{P}_{yes} g^{-1} \subseteq \cal{P}_{yes}$ and $g \cal{P}_{no} g^{-1} \subseteq \cal{P}_{no}$ for all $g \in \unitary(d)$). 
An example of such a property includes deciding whether a unitary $U$ is a reflection about a subspace of $\C^d$ of dimension at most $w$ (the \emph{no} instances) or dimension at least $2w$ (the \emph{yes} instances). %

It is easy to see that whether a unitary $U$ is a \emph{yes} or \emph{no} instance of a unitarily invariant property $\cal{P}$ only depends on the multiset of eigenvalues of $U$. In fact, we can say something stronger; the following establishes a symmetrization method for polynomials that decide unitarily invariant properties.

\begin{restatable}[Symmetrization for unitarily invariant properties]{theorem}{unitarilyinvariant}
Let $\cal{P} = (\cal{P}_{yes},\cal{P}_{no})$ denote a $d$-dimensional unitarily invariant property. Suppose there is a $T$-query quantum algorithm that accepts \emph{yes} instances with probability at least $a$ and \emph{no} instances with probability at most $b$. Then there exists a degree at most $2T$ symmetric\footnote{Here, symmetric means that for all permutations $\pi:[d] \to [d]$, permuting the variables $z_i \to z_{\pi(i)}$ and $z_i^* \to z_{\pi(i)}^*$ leaves the polynomial $q$ unchanged.} self-adjoint polynomial $q(z_1, \ldots, z_d, z_1^*, \ldots, z_d^*)$ satisfying
\begin{itemize}
    \item If $U \in \cal{P}_{yes}$ then $q(z_1, \ldots, z_d,z_1^*,\ldots,z_d^*) \geq a$
    
    \item If $U \in \cal{P}_{no}$ then $q(z_1, \ldots, z_d,z_1^*,\ldots,z_d^*) \leq b$
\end{itemize}
 where $(z_1, \dots, z_d)$ and $(z_1^*, \dots, z_d^*)$ are the eigenvalues of $U$ and their complex conjugates, respectively.
\label{lem:unitary_invariant_main_thm}
\end{restatable}

The symmetrization of \Cref{lem:unitary_invariant_main_thm} may at first appear quite modest. The symmetrized polynomial $q$ still acts on $2d$ variables, which is fewer than the $2d^2$ variables acted on by the original polynomial $p$ from \Cref{prop:basic-polynomial}, but is a far cry from a univariate polynomial which approximation theory is best suited to handle. We have made some progress, however: as mentioned, the property $\cal{P}$ only depends on the eigenvalue multiset of $U$, and the polynomial $q$ is directly a function of the eigenvalues (whereas the original polynomial $p$ is a function of the matrix entries of $U$, which are not obviously related to the property in a low-degree fashion). Furthermore, the polynomial $q$ can be symmetrized further to obtain a univariate polynomial $r$, which we analyze using approximation theory. We illustrate this with two applications of \Cref{lem:unitary_invariant_main_thm}. 

\paragraph{Unitarily Invariant Subspace Properties.} As a warmup, consider \emph{subspace properties}, which consist of reflections about a subspace, i.e., $U = I - 2\Pi$ where $\Pi$ is the projector onto some subspace $S \subseteq \C^d$. We say that $U$ \emph{encodes} the subspace $S$. An example of a unitarily invariant subspace property is the \emph{Approximate Dimension} problem, which we parametrize by an integer $w \in \{1,2,\ldots,d\}$. The \emph{yes} instances consist of (unitaries encoding) subspaces of dimension at least $2w$, and the \emph{no} instances consist of subspaces of dimension at most $w$. This is a quantum generalization of the \emph{Approximate Counting} problem, which is to determine whether the Hamming weight of an input string is at least $2w$ or at most $w$.

Since the eigenvalues of subspace unitaries are either $1$ or $-1$, by \Cref{lem:unitary_invariant_main_thm} unitarily invariant subspace properties yield polynomials that compute acceptance probabilities on (a subset of) $\{1,-1\}^d$. Note that, after mapping $\{1,-1\}$ to $\{0,1\}$, these are the same kind of polynomials that arise when analyzing classical properties! In fact, there is a one-to-one correspondence between symmetric classical properties $\cal{S}$ (properties that only depend on the Hamming weight of the input) and unitarily invariant subspace properties $\cal{P}$.
This implies that the polynomial $q$ given by \Cref{lem:unitary_invariant_main_thm} (associated to a unitarily invariant subspace property $\cal{P}$), also distinguishes between the \emph{yes} and \emph{no} instances of the associated classical symmetric property $\cal{S}$. This yields a lower bound method for the query complexity of $\cal{P}$:

\begin{restatable}{proposition}{symmetric}
\label{prop:symmetric}
Let $\cal{P}$ be a unitarily invariant subspace property and let $\cal{S}$ be the associated symmetric classical property. The query complexity of distinguishing between \emph{yes} and \emph{no} instances of $\cal{P}$ is at least the minimum degree of any polynomial that distinguishes between the \emph{yes} and \emph{no} instances of $\cal{S}$.  
\end{restatable}
Therefore, degree lower bounds on polynomials that decide a classical symmetric property $\cal{S}$, automatically yield query complexity lower bounds for the quantum property $\cal{P}$. Approximate degree lower bounds on symmetric boolean functions are well-studied in complexity theory~\cite{paturi92degree,de2008note}; these can be automatically ``lifted'' to the unitary property testing setting. 

We note that there is another way of seeing this reduction: any $T$-query tester for $\cal{P}$ is automatically a $T$-query tester for $\cal{S}$; thus lower bounds on $\cal{S}$ imply lower bounds on $\cal{P}$. 
Here our motivation is to present \Cref{prop:symmetric} as a simple application of \Cref{lem:unitary_invariant_main_thm}. 

For example, the Approximate Dimension problem contains the Approximate Counting problem as a special case. Any polynomial that decides with bounded error the Approximate Counting problem must have degree at least $\Omega(\sqrt{d/w})$, which implies the same lower bound for the query complexity of the Approximate Dimension problem.

\paragraph{Recurrence Time of Unitaries.}

Not all unitarily invariant properties reduce to classical lower bounds. For instance, we analyze a problem related to the recurrence times of unitaries. 

In general, the recurrence time of a dynamical system is the time that the system takes to return to a state that is close to its initial state (if it exists). The recurrence statistics of dynamical systems have been extensively studied in the physics literature, where they have been used as indicators of chaotic behaviour within a dynamical system \cite{saussol2009introduction}. As the time evolution of a quantum system is governed by a unitary operator, the recurrence times of unitary matrices is of particular interest. The Poincar\'e recurrence theorem guarantees that the recurrence time exists for certain quantum mechanical systems \cite{bocchieri1957quantum}. For example, the expected recurrence times of a Haar-random unitary were studied in \cite{https://doi.org/10.48550/arxiv.1412.3085}. We now define these concepts more formally.

\begin{restatable}[Recurrence Time Problem]{definition}{recurrencedef}
The $(t,\epsilon)$-Recurrence Time problem is to decide, given oracle access to a unitary $U$, whether $U^t = I$ (\emph{yes} case) or $\| U^t - I \| \geq \epsilon$ in the spectral norm (\emph{no} case), promised that one is the case.
\end{restatable}

Note that the instances of this problem are generally not self-adjoint; their eigenvalues can be any complex number on the unit circle. There is no obvious classical analogue of the unitary Recurrence Time problem, and thus it does not seem to naturally reduce to a classical lower bound. We instead employ \Cref{lem:unitary_invariant_main_thm} to prove the following lower bound on the Recurrence Time problem:

\begin{restatable}{theorem}{recurrence}
Let $\epsilon \leq \frac{1}{2 \pi}.$ Any quantum query algorithm solving the $(t,\epsilon)$-Recurrence Time problem for $d$-dimensional unitaries with error $\epsilon$ must use $\Omega(\max(\frac{t}{\epsilon}, \sqrt{d}))$ queries. 
\label{thm:recurrence}
\end{restatable}

We prove the lower bound by observing that the Recurrence Time problem is testing a unitarily invariant property, and hence \Cref{lem:unitary_invariant_main_thm} applies to give a polynomial $q$ representing the acceptance probability of any algorithm solving the problem in terms of the eigenvalues of the input unitary $U$. We then symmeterize the polynomial $q$ by constructing a distribution $D(p, z)$ on unitaries with exactly two eigenvalues $\{1, z = e^{i \theta}\}$ such that the expected acceptance  probability $r(p, z)$ of the algorithm over $D(p, z)$ remains a polynomial in $p$ and $z$. Afterwards, Markov-Bernstein type inequalities are used to prove a lower bound on the degree of $r$. 

We also establish the following upper bound:

\begin{restatable}{theorem}{recurrenceub}
The $(t,\epsilon)$-Recurrence Time problem can be solved using $O(t \sqrt{d}/\epsilon)$ queries.
\label{thm:recurrence-ub}
\end{restatable}

It is an interesting question to determine whether the upper bound or lower bound (or neither) is tight. 

\subsubsection{Beyond Unitarily Invariant Properties}
\label{sec:intro_lu_properties}

For unitarily invariant properties, there was a natural candidate for how to symmetrize the polynomials $p$ we get from  \Cref{prop:basic-polynomial} by using viewing $p$ as a polynomial in the eigenvalues of the matrix $U$ rather than the matrix entries of $U$. However, a symmetrization technique for other properties is less unclear. In this direction, we develop symmetrization techniques based on invariant theory.

Invariant theory studies the action of a group $G$ on a polynomial ring $\mathbb{C}[x_1, \dots, x_n]$. We denote the action of $g \in G$ on $f \in \mathbb{C}[x_1, \dots, x_n]$ by $g \cdot f.$ The ring of invariant polynomials $\mathbb{C}[x_1, \dots, x_n]^{G}$ is then the subring of $\mathbb{C}[x_1, \dots, x_n]$ consisting of polynomials satisfying $g \cdot f = f$ for all $g \in G$, that is $f \in \mathbb{C}[x_1, \dots, x_n]^G$ is left unchanged by the action of $g$ for all group elements.

There are many natural questions about the invariant ring $\mathbb{C}[x_1, \dots, x_n]^G$ one can ask, such as construction of a generating set for the invariant ring. For example, one classical example is the action of a permutation $\sigma \in S_n$ acting on a polynomial $p(x_1, \dots, x_n)$ by permuting the variables by $\sigma \cdot p = p(x_{\sigma(1)}, \dots, x_{\sigma(n)}).$ The invariant ring is known as the ring of symmetric polynomials, for which there are many well-known generating sets. One example of a generating set is the power sum symmetric polynomials given by $p_i = \sum_{j=1}^n x_j^i$, which generate the symmetric polynomial ring as an algebra. In other words, for any symmetric polynomial $f$, there exists a polynomial $g$ for which $f = g(p_1, \dots, p_n).$ There are similar characterizations of the invariant ring for numerous other group actions. 

To connect invariant theory with our \Cref{prop:basic-polynomial}, we prove the following result for testing $G$-invariant unitary properties. Since we are studying properties of general unitaries, not just boolean strings, we consider symmetries coming from subgroups of the unitary group $\unitary(d)$.  Let $G \subseteq U(d)$ be a compact subgroup equipped with a Haar measure $\mu$ (i.e., a measure over $G$ that is invariant under left-multiplication by elements of $G$). %

\begin{definition}[$G$-invariant property]
Let $G \subseteq U(d)$ be a compact group. A $d$-dimensional unitary property $\mathcal{P} = (\cal{P}_{yes},\cal{P}_{no})$ is $G$-invariant if for every $g \in G$ we have $g \cal{P}_{yes} g^{-1} \subseteq \cal{P}_{yes}$ and $g \cal{P}_{no} g^{-1} \subseteq \cal{P}_{no}$.
\label{defn:g-inv-property}
\end{definition}

\begin{definition}[Invariant rings]
Let $\C[X, Y]_d$ be the ring of complex polynomials in matrix variables $X = (x_{i,j})_{1 \leq i, j \leq d}$ and $Y = (y_{i,j})_{1 \leq i, j \leq d}.$ Observe that there is an action of $G$ on any $f(X, Y) \in \C[X, Y]$ by simultaneous conjugation $g \cdot f(X, Y) = f(g X g^{-1}, g Y g^{-1} ).$

The \emph{invariant ring} $\C[X, Y]_d^G$ is the subring of polynomials in $\C[X, Y]_d$ satisfying $g \cdot f = f$ for all $g \in G.$
\label{defn:invariants}
\end{definition}

The general theory of invariant theory guarantees the existence and finiteness of a generating set for the invariant ring $\mathbb{C}[X, Y]^G$ for all compact groups, which includes all finite groups, the unitary group, and products of unitary groups as special cases. Furthermore, the following proposition connects the invariant ring to property testers for $G$-invariant properties.

\begin{restatable}[Symmeterization for $G$-invariant properties]{proposition}{symmetrizedpolynomial}
Suppose $\mathcal{P} = (\cal{P}_{yes},\cal{P}_{no})$ is a $G$-invariant $d$-dimensional unitary property. If there is a $T$-query tester for $\cal{P}$ that accepts \emph{yes} instances with probability at least $a$ and \emph{no} instances with probability at most $b$, then there exists a self-adjoint degree-$2T$ polynomial $q$ in the invariant ring $\C[X, X^*]_d^G$ satisfying
\begin{itemize}
    \item If $U \in \mathcal{P}_{yes}$, then $q(U,U^*) \geq a$.
    \item If $U \in \mathcal{P}_{no}$, then $q(U,U^*) \leq b$.
\end{itemize}
\label{prop:basic-polynomial-symmetrized}
\end{restatable}

While \Cref{prop:basic-polynomial-symmetrized} at first may seem difficult to apply, the invariant ring has been characterized in numerous cases. Depending on the group, the associated invariant ring may have a much simpler description than the full polynomial ring, making it easier to prove degree lower bounds. For instance, in the case where $G$ is the full unitary group, the invariant polynomials are exactly symmetric polynomials in the eigenvalues of $U$ and the adjoint $U^*.$ 

We illustrate this connection to invariant theory by considering property testing questions related to \textit{entanglement} of quantum states, which is a central concept in quantum information theory. Recall that a state $\ket{\psi} \in \mathbb{C}^d \otimes \mathbb{C}^d$ is \textit{entangled} if it cannot be written as a tensor product of two states $\ket{\psi_1} \otimes \ket{\psi_2}$ where $\ket{\psi_1}, \ket{\psi_2} \in \mathbb{C}^d.$ 
The property of being entangled is invariant under the \textit{local unitary} group instead of the full unitary group.

\begin{restatable}[Local Unitary Group]{definition}{localunitary}
Let $d_1, d_2 \geq 2.$ The local unitary group $\LU(d_1, d_2)$ is the subgroup $\unitary(d_1) \times \unitary(d_2)$ of $\unitary(d_1 d_2)$ consisting of all unitaries of the form $g \otimes h$ where $g \in \unitary(d_1), h \in \unitary(d_2).$
\end{restatable}

Furthermore, the \textit{entanglement entropy} of a state $\ket{\psi} \in \mathbb{C}^d \otimes \mathbb{C}^d$ can be used as a measurement of how entangled the state is. Numerous definitions of entanglement entropy have been proposed in the physics and quantum information literature; we use the following definition of entanglement entropy in this work. 
\begin{restatable}[R\'{e}nyi $2$-entropy]{definition}{renyientropy}
Given a state $\ket{\psi} \in \mathbb{C}^d \otimes \mathbb{C}^d$ with reduced density matrix on the first register $\rho$, the R\'{e}nyi 2-entropy of $\ket{\psi}$ is defined as $H_2(\ket{\psi}) = - \log \Tr(\rho^2).$ 
\end{restatable}

We note that since $\ket{\psi} \in \C^d \otimes \C^d$ is a pure state, it does not matter whether or not the reduced density matrix is taken with respect to the first or second register, since both matrices will have the same set of eigenvalues. 

We now define the Entanglement Entropy problem as the task of distinguishing between high and low entropy states.

\begin{restatable}[Entanglement Entropy Problem]{definition}{entanglemententropyproblem}
Let $0 < a < b \leq \log d$. Given oracle access to a reflection oracle $U = I - 2 \ketbra{\psi}{\psi}$ where $\ket{\psi} \in \mathbb{C}^d \otimes \mathbb{C}^d,$ decide whether or not the state $\ket{\psi}$ satisfies one of the following two conditions, promised one of the following is the case:

\begin{itemize}
    \item Low entropy case: $H_2(\ket{\psi}) \leq a$
    \item High entropy case: $H_2(\ket{\psi}) \geq b$
\end{itemize}
\end{restatable}

Since entanglement entropy of a state is an $\LU$-invariant quantity (i.e. $H_2((g \otimes h) \ket{\psi}) = H_2(\ket{\psi})$ for all unitaries $g$ and $h$), the Entanglement Entropy problem corresponds to an $\LU$-invariant unitary property, opening the door to exploiting well-known results from invariant theory to prove query lower bounds. Indeed, leveraging a characterization of $\LU$-invariant polynomials by Procesi~\cite{PROCESI1976306} and Brauer~\cite{brauer1937algebras}, and specializing it to the Entanglement Entropy problem, we obtain the following lower bound using our generalized polynomial method and \Cref{prop:basic-polynomial-symmetrized}:
\begin{restatable}{theorem}{entanglemententropybound}
Assume $a \geq 5.$ Given parameters $a < b \leq \log d$, any tester must make $\Omega(\exp(a/4))$ queries to distinguish between the low and high entropy cases in the Entanglement Entropy problem.
\label{thm:entanglement_entropy_bound}
\end{restatable}
We hope that this connection between invariant theory and quantum query complexity can be used as a general framework to prove new lower bounds.

\subsection{Property Testing with Quantum Proofs}

We also study unitary property testing with \emph{quantum proofs}: in addition to getting query access to the instance $U$, the tester also receives an additional quantum state called a \emph{proof} that supposedly certifies that $U$ is a \emph{yes} instance. On one hand, having access to a quantum proof can significantly reduce the number of queries to $U$ needed. On the other hand, the proof state is not trusted and must be verified: if $U$ is a \emph{no} instance it must be rejected with high probability no matter what proof was provided. Since this definition is analogous to the definition of the complexity class $\QMA$, we call this the ``$\QMA$ property testing model''. Similarly we call the standard definition of unitary property testing (without quantum proofs) the ``$\BQP$ property testing model''. 

An illustration of $\QMA$ property testing is that of unstructured search, where the goal is to determine whether $U \ket{\psi} = -\ket{\psi}$ for some state $\ket{\psi}$ (and acts as the identity everywhere else) or whether $U = I$, promised that one is the case. If $U$ is $d$-dimensional, then the generalized polynomial method (see \Cref{sec:symmetric_properties}) implies the $\BQP$ query complexity of this problem is $\Theta(\sqrt{d})$, but on the other hand the $\QMA$ query complexity of this problem is $1$: given a proof state $\ket{\theta}$, the tester can verify using a single query whether $U$ applies a nontrivial phase to $\ket{\theta}$, in which case the tester would accept, and otherwise reject. Thus quantum proofs can dramatically reduce the query complexity of a problem. 

The $\QMA$ property testing model motivates the following questions: which problems admit query speedups when quantum proofs are provided? For which problems are quantum proofs useless? In addition to the $\BQP$ property testing lower bounds mentioned above, we prove $\QMA$ property testing lower bounds for the Approximate Dimension,  Recurrence Time, and Entanglement Entropy problems. 

We note that in the $\BQP$ setting, our lower bounds can also be obtained by other methods, such as the ``hybrid method'' of \cite{bennett1997strengths}. However, it is unclear how to apply this method in the $\QMA$ setting, and hence the polynomial method appears necessary to prove non-trivial $\QMA$ lower bounds. Furthermore, even in the $\BQP$ setting, we believe that the polynomial method provides a clean and simple method to prove  lower bounds compared to other methods.

The $\QMA$ lower bound for Approximate Dimension is obtained by observing that any $\QMA$ tester for Approximate Dimension is also a $\QMA$ tester for the classical Approximate Counting problem (i.e., counting the Hamming weight of an input string). Thus, the lower bound follows immedately from the $\QMA$ lower bound on Approximate Counting proved by Aaronson, et al.~\cite{aaronson2020quantum}:

\begin{restatable}[$\QMA$ lower bound for Approximate Dimension]{theorem}{appxdimqma}
Suppose there is a $T$-query algorithm that solves the Approximate Dimension problem (i.e. deciding whether a $d$-dimensional unitary encodes a subspace of dimension at least $2w$ or at most $w$) with the help of a $m$-qubit proof. Then either $m = \Omega(w)$, or $T \geq \Omega(\sqrt{\frac{d}{w}})$.
\label{thm:apx_dim_qma}
\end{restatable}

As with the $\BQP$ lower bound for the Recurrence Time problem, the $\QMA$ lower bounds for the Recurrence Time problem requires more work than leveraging lower bounds on a related classical problem. However, using a similar technique as the $\BQP$ lower bound, we obtain the following:

\begin{restatable}[$\QMA$ lower bound for the Recurrence Time problem]{theorem}{recurrenceqma}
Let $\epsilon \leq \frac{1}{2 \pi}.$ Suppose there is a $T$-query algorithm that solves the Recurrence Time problem for $d$-dimensional unitaries with the help of an $m$-qubit proof. Then either $m \geq \Omega(d)$, or $T \geq \Omega(\max(\sqrt{\frac{d}{m}}, \frac{t}{m}, \frac{1}{\epsilon}))$. 
\label{thm:recurrence_qma}
\end{restatable}

Finally, we also adapt the technique for the $\BQP$ lower bound for the Entanglement Entropy problem, to prove a $\QMA$ lower bound for the same problem.

\begin{restatable}[$\QMA$ lower bound for the Entanglement Entropy problem]{theorem}{entanglemententropyboundqma}
Assume $a \geq 5$ and $a < b \leq \log d$ Suppose there is a $T$-query algorithm that solves the entanglement entropy problem with the help of an $m$-qubit witness, then  $mT \geq \Omega(\exp(a/4)).$
\label{thm:entanglement_entropy_bound_qma}
\end{restatable}

We note that we are able to give an algorithm for the entanglement entropy problem using $O(\exp(a))$ queries and a certificate size with $O(\exp(a))$ qubits as long as the gap satisfies $b \geq 2a$. Furthermore, our best upper bound in the $\QMA$ setting for the Recurrence Time problem is identical to the $\BQP$ upper bound (\Cref{thm:recurrence-ub}), which uses $O(\frac{t\sqrt{d}}{\epsilon})$ queries. It is open whether or not the query complexity in the $\QMA$ setting can be improved by making use of the witness, or if the lower bounds can be tightened.

\subsection{$\QMA$ vs. $\QMAtwo$}

In addition to the $\QMA$ model of property testing with the help of a quantum proof, we can also study what happens if we place restrictions on the proof states allowed. For example, what if the proof is guaranteed to be a classical string (i.e., a $\QCMA$ proof) or is unentangled across a fixed bipartition of qubits (i.e., a $\QMAtwo$ proof)?

For example, consider the problem of testing whether or not a given unitary $U$ was the identity $I$ or a reflection $I - 2 \ketbra{\psi}{\psi}$ where $\ket{\psi}$ is an $n$-qubit state. As we observed in the previous section, this problem can be solved using one query given access to a proof state $\ket{\psi}$. However, Aaronson and Kuperberg \cite{aaronson2007quantum} showed that if the proof given was an $m$-bit \emph{classical} string, any quantum algorithm must use $\Omega(\sqrt{\frac{2^n}{m+1}})$ queries to distinguish between the two cases. In particular, the result was used  by Aaronson and Kuperberg in \cite{aaronson2007quantum} to give a quantum oracle separation between \textsf{QMA} and \textsf{QCMA}.

The complexity class \textsf{QMA(k)} is defined as the class of problems verifiable by a polynomial time quantum circuit with access to $k \geq 2$ \emph{un}entangled proofs. It was shown in \cite{harrow2013testing} that for any constant $k > 2$, we have \textsf{QMA(k)} = \textsf{QMA(2)}, as any \textsf{QMA(k)} verifier can be simulated by a \textsf{QMA(2)} verifier. %

As discussed in Aaronson's survey paper on quantum query complexity in \cite{aaronson2021open}, an oracle separation between \textsf{QMA(2)} and \textsf{QMA} is a notorious open problem in quantum complexity theory.  We note that $\QMA \subseteq \QMAtwo$ since a $\QMAtwo$ verifier could simulate a $\QMA$ verifier by using only one of the proofs provided. However, 
there is evidence that \textsf{QMA(2)} could be more powerful than \textsf{QMA}:

\begin{itemize}
    \item The existence of \textsf{QMA(2)} protocols for the verification of \textsf{NP}-complete problems (eg. graph 3-colouring) using a logarithmic number of qubits, as outlined in \cite{bliertapp2009}. If there were a  \textsf{QMA} protocol for these problems using a logarithmic number of qubits, then \textsf{NP} $\subseteq$ $\textsf{QMA}_{\textsf{log}}$ = \textsf{BQP}, which is considered unlikely \cite{marriott2005quantum}. Otherwise, if there exists a \textsf{QMA} protocol with a sublinear number of qubits for 3-SAT or 3-colouring, then the (classical) Exponential Time Hypothesis is false~\cite{aaronson2008power,chen2010short}.  \footnote{We can prove there is an oracle relative to which $\mathsf{NP}$ does not have sublinear-sized $\QMA$ proofs, using an argument of Aaronson similar to that in \cite{aaronson2011impossibility} that combines the ``guessing lemma'' and the polynomial method.} 
    \item Certain problems in quantum chemistry, specifically the pure state $N$-representability problem, known to have \textsf{QMA(2)} protocols but not \textsf{QMA} protocols \cite{liu2007}. 
    \item The non-existence of a product test using local quantum operations and classical communication (LOCC) as proven in \cite{harrow2013testing}. If an LOCC product test existed, then \textsf{QMA(2)} = \textsf{QMA}.  
\end{itemize} 
On the other hand, despite many years of study, the only complexity inclusions about $\QMAtwo$ known are $\textsf{QMA} \subseteq \textsf{QMA(2)} \subseteq \textsf{NEXP}$, a vast gap in the complexity-theoretic landscape. A first step towards showing that $\QMAtwo$ is indeed more powerful than $\QMA$ would be to identify an oracle relative to which $\QMAtwo$ is different than $\QMA$. This would already have very interesting consequences in quantum information theory, such as ruling out the existence of disentanglers~\cite{aaronson2008power}. 

In this paper we identify a unitary property testing problem, for which if we can prove a strong $\QMA$ lower bound would immediately imply an oracle separation between $\QMAtwo$ and $\QMA$. While we do not obtain a lower bound, we present some observations that may be helpful towards eventually obtaining the desired oracle separation. 

In order to define the problem, we first have to define the notion of an \emph{$\epsilon$-completely entangled subspace}. This is a subspace $S \subseteq \C^d \otimes \C^d$ such that all states $\ket{\theta} \in S$ are $\epsilon$-far in trace distance from any product state $\ket{\psi} \otimes \ket{\phi}$. It is known, via the probabilistic method, that there exist subspaces of dimension $\Omega(d^2)$ that are $\Omega(1)$-completely entangled~\cite{Hayden_2006}. We now introduce the Entangled Subspace problem:

\begin{restatable}[Entangled Subspace problem]{definition}{entangledsubspace}
Let $0 \leq a < b < 1$ be constants. The $(a,b)$-Entangled Subspace problem is to decide, given oracle access to a unitary $U = I -2\Pi$ where $\Pi$ is the projector onto a subspace $S \subseteq \C^d \otimes \C^d$, whether 
\begin{itemize}
    \item (\emph{yes} case) $S$ contains a state $\ket{\theta}$ that is $a$-close in trace distance to a product state $\ket{\psi} \otimes \ket{\phi}$. 
    \item (\emph{no} case) $S$ is $b$-completely entangled
\end{itemize}
promised that one is the case.
\end{restatable}

First, we observe that the Entangled Subspace property is $\LU$-invariant: applying local unitaries $g \otimes h$ to a subspace $S$ preserves whether it is a \emph{yes} instance or a \emph{no} instance of the problem. Thus one can hope to prove query lower bounds for the Entangled Subspace problem in both the $\BQP$ and $\QMA$ setting using our generalized polynomial method and tools from invariant theory. 

Next, we observe that there is in fact a $\QMAtwo$ upper bound for the Entangled Subspace problem, which is almost immediate from the definition of $\QMAtwo$:

\begin{restatable}[$\QMAtwo$ upper bound for the Entangled Subspace problem]{proposition}{entangledsubspaceqmatwo}
The Entangled Subspace problem
can be solved by a $\QMAtwo$ tester, meaning that the tester receives a proof state in the form $\ket{\psi} \otimes \ket{\varphi}$ of $\poly\log(d)$ qubits, makes one query to the unitary $U$, and can distinguish between \emph{yes} and \emph{no} cases with constant bias. %
\label{prop:entangledsubspaceqmatwo}
\end{restatable}

We also construct another $\QMAtwo$ verifier for the Entangled Subspace problem with the property that a valid proof state for the \emph{yes} instances is symmetric. We call this verifier the \textit{product test verifier}. This procedure is based on the \emph{product test}, which is a procedure for detecting whether a state $\ket{\theta}$ is close to a product state $\ket{\psi} \otimes \ket{\phi}$, given access to $k \geq 2$ copies $\ket{\theta}^{\otimes k}$. We use of a generalization of the  product test analysis of \cite{soleimanifar2022testing} that applies for all $k \geq 2$, which they showed for the case $k = 2$.  This additional property will be relevant when discussing \Cref{thm:one_dimensional_property_qma}. 

We conjecture the following $\QMA$ lower bound on the Entangled Subspace problem.

\begin{conjecture}
Any $\QMA$ tester for the Entangled Subspace problem that makes $T$ queries to the oracle and receives an $m$-qubit witness must have either $m$ or $T$ be superpolynomial in $\log d$. 
\label{conj:qmatwo}
\end{conjecture}

If this conjecture is true, then this would imply the existence of a quantum oracle that separates $\QMA$ from $\QMAtwo$: the oracle would encode, for each $\QMA$ tester, an instance of the Entangled Subspace problem that the tester decides incorrectly. 

We note that in the one-dimensional case, the Entanglement Entropy problem can be viewed as a special case of the Entanglement Subspace problem, since \Cref{lem:entanglement_schmidt} provides bounds on the entropy in terms of the trace distance to the closest product state. However, the $\QMA$ lower bound for the Entanglement Entropy problem in \Cref{thm:entanglement_entropy_bound_qma} is not sufficient to prove \Cref{conj:qmatwo}. Since we are assuming the gap between $a$ and $b$ is constant for a $\QMAtwo$ verifier to be able to distinguish between the two cases, the states given in the Entangled Subspace problem have entropy independent of the dimension $d$, whereas we only have we have a strong $\QMA$ lower bound for the Entanglement Entropy problem when the entropy of the hidden states grows with the dimension $d$.

Hence, we expect that new techniques are needed to resolve \Cref{conj:qmatwo}, such as techniques that would extend to the setting where the hidden subspace has a larger dimension. As far as we know, \Cref{conj:qmatwo} potentially could be proved by giving a $\QMA$ lower bound for the Entangled Subspace problem restricted to \emph{two-dimensional} subspaces. 

\paragraph{Constraints on the Conjecture.}
As a first step towards understanding whether \Cref{conj:qmatwo} is true, we identify constraints on the conjecture: we show that we cannot hope to prove a super-polynomial $\QMA$ lower bound on the Entangled Subspace problem when only considering one-dimensional subspaces. This is a consequence of the following general statement. 

\begin{restatable}{lemma}{onedimensionalqma}
Let $\cal{P}$ denote a property where the instances are unitaries encoding a one-dimensional subspace (i.e. a pure state): $U = I - 2\ketbra{\psi}{\psi}$ for some state $\ket{\psi}$. Suppose that there is a $T$-query $\QMAtwo$ tester that decides $\cal{P}$, with the condition that a valid proof state for \emph{yes} instances is $\ket{\psi}^{\otimes 2}$. Then there exists a $O(T)$-query $\QMA$ tester that also decides $\cal{P}$. 
\label{thm:one_dimensional_property_qma}
\end{restatable}

Let $\cal{P}$ denote the property testing problem where in the \emph{yes} case the unitary encodes a product state $\ket{\psi} = \ket{\varphi} \otimes \ket{\xi}$ and in the \emph{no} case it encodes an entangled state. As we can construct a $\QMAtwo$ verifier for $\mathcal{P}$ satisfying the conditions of \Cref{thm:one_dimensional_property_qma}, we immediately get that there is a $\QMA$ tester that decides $\mathcal{P}.$

Hence, this theorem is saying that property testing questions related to quantum \textit{states} are insufficient to give an oracle separation between $\QMA$ and $\QMAtwo$. On the other hand, in the higher dimensional setting we obtain a different result. We analyze the behavior of product test verifier in the $\QMA$ setting. This means that instead of getting a proof state that is promised to be unentangled across a bipartition, the product test verifier receives a single pure state that may be entangled everywhere. We show that the product test verifier loses its soundness in the $\QMA$ setting:

\begin{theorem}[Informal Version of \Cref{thm:counterexample}]
There exists a $6$-dimensional completely entangled subspace $S$ and an entangled proof state $\ket{\theta}$ that the product test verifier accepts with probability $1$ (even though it rejects all product proof states $\ket{\psi} \otimes \ket{\phi}$ with high probability).
\end{theorem}

The construction of the ``fooling'' subspace $S$ and the proof state $\ket{\theta}$ is completely explicit and combinatorial; we believe it can suggest how more general $\QMA$ verifiers (beyond the product test) can be fooled, which would give insight towards proving \Cref{conj:qmatwo}. Indeed in \Cref{sec:product_test_witnesses}, we characterize all states that pass the product test with probability one, which enables the construction of other ``fooling'' subspaces for the product test.

\paragraph{Average Case Problems.} Finally, we also propose two \emph{average case} variants of the Entangled Subspace problem, in which the task is to distinguish between two distributions over unitaries $U$. Let $U$ be a Haar-random matrix on $\C^d \otimes \C^d$. A Haar-random subspace of dimension $s$ is the image of $UP_S$, where $P_S$ is a projector onto any fixed subspace $S \subseteq \C^d \otimes \C^d$ of dimension $s$.

\begin{restatable}[Planted Product State Problem]{definition}{plantedproductstate}
Let $0 < s < d^2$ denote an integer parameter.
Consider the following two distributions over subspaces $S$ of $\C^d \otimes \C^d$:

\begin{itemize}
    \item {\bf No planted state}: $S$ is a Haar-random subspace of dimension $s$. 
    
    \item {\bf Has planted state}:  $S$ is an $(s+1)$-dimensional subspace chosen by taking the span of a Haar-random $s$-dimensional subspace with a product state $\ket{\psi}\otimes \ket{\phi}$ for Haar-random $\ket{\psi},\ket{\phi}$.
\end{itemize}

The Planted Product State problem is to distinguish, given oracle access to a unitary $U = I - 2\Pi$ encoding a subspace $S$, whether $S$ was sampled from the \textbf{No planted state} distribution (\emph{no} case) or the \textbf{Has planted state}  distribution (\emph{yes} case), promised that one is the case.
\end{restatable}

\begin{restatable}[Restricted Dimension Counting Problem]{definition}{restricteddimcount}
Let $0 < t \leq d$ and $0 < r \leq t^2$ denote integer parameters. Consider the following distribution, parameterized by $(t,r)$, over subspaces $S \subseteq \C^d \otimes \C^d$:
\begin{itemize}
    \item Sample Haar-random $t$-dimensional subspaces $R,Q \subseteq \mathbb{C}^d$.
    \item Sample a Haar-random $r$-dimensional subspace of $S \subseteq R \otimes Q$.
\end{itemize}
Let $0 < C_1 < C_2 < 1$ denote constants. The Restricted Dimension Counting problem is to decide, given query access to a unitary $U = I - 2\Pi$ encoding a subspace $S$, whether $S$ was sampled from either the $(t,C_1 t^2)$ distribution or $(t,C_2 t^2)$ distribution, promised that one is the case. 
\end{restatable}

The relationship between these two average case problems and the Entangled Subspace problem is captured by the following propositions.

\begin{restatable}{proposition}{plantedyesno}
If $S$ is sampled from the \textbf{Has planted state} distribution of the Planted Product State problem, then it is a \emph{yes} instance of the Entangled Subspace problem. If $S$ is sampled from the \textbf{No planted state} distribution with $s = Cd^2$ for some sufficiently small constant $C > 0$, then it is a \emph{no} instance with overwhelming probability.
\label{prop:plantedyesno}
\end{restatable}

\begin{restatable}{proposition}{restrictedyesno}
There exist constants $0 < C_1 < C_2 < 1$ such that if $S$ is sampled from the $(t,C_1 t^2)$ distribution from the Restricted Dimension Counting problem, it is a \emph{no} instance of the Entangled Subspace problem with overwhelming probability. If it is sampled from the $(t,C_2 t^2)$ distribution, then it is a \emph{yes} instance with overwhelming probability.
\label{prop:restrictedyesno}
\end{restatable}
These two propositions are proved using methods from random matrix theory. \Cref{prop:entangledsubspaceqmatwo} in turn implies that the Planted Product State and Restricted Dimension Counting problems can be solved by a $\QMAtwo$ tester with overwhelming probability. We further conjecture that in fact these two problems cannot be solved in polynomial time by a $\QMA$ tester. This conjecture would clearly also imply an oracle separation between $\QMA$ and $\QMAtwo$. 

\begin{conjecture}
Any $\QMA$ tester that solves the Planted Product State or Restricted Dimension Counting problems with constant probability either requires a proof state of super-polynomial size, or make super-polynomially many queries to the oracle.
\label{conj:qma_lower_bound_pps}
\end{conjecture}

As evidence towards \Cref{conj:qma_lower_bound_pps}, we show a lower bound against $\QCMA$ testers.

\begin{theorem}[Informal version of \Cref{thm:qcma_lower_bound}]
Any $T$-query quantum algorithm solving the Planted Product State problem with the help of an $m$-bit classical witness must have $m$ or $T$ superpolynomial in $\log d.$
\end{theorem}

We believe that these two average case problems are connected to several interesting topics. The first is the $\QMA$ versus $\QMAtwo$ problem, of course. Their formulation also suggests that tools from random matrix theory can be brought to bear to study them. Finally, the Restricted Dimension Counting problem is, as the name suggests, a more structured, special case of the general Approximate Dimension problem, for which we proved a strong $\QMA$ lower bound! Perhaps the techniques used to prove lower bounds on the Approximate Dimension problem can be extended to the Restricted Dimension Counting problem.

\subsection{Related Work}

In quantum query complexity, there have been two main paradigms for proving lower bounds for query complexity, namely the polynomial method \cite{beals2001quantum} and the adversary method \cite{vspalek2005all}. The methods are generally incomparable, as there are problems where one method is able to prove a tight lower bound but the other cannot. For instance, the collison problem \cite{kutin2003quantum} is a case where the polynomial method proves a tight lower bound but the adversary method provably fails to do so. On the other hand, Ambanis \cite{ambainis2006polynomial} constructed an example where the adversary method is provably better than the polynomial method. Furthermore, for evaluation of Boolean functions, the general adversary method characterizes the quantum query complexity up to constant factors \cite{reichardt2011reflections}.

However, in their original form, both methods assume that the quantum oracle encodes a Boolean function $f: \{0,1\}^n \rightarrow \{0,1\}$. A similar generalization of the adversary method to unitary property testing problems was introduced by Belovs \cite{belovs2015variations}. In particular, it was applied to the approximate counting problem \cite{belovs2020tight}, with a symmetrization technique based on the representation theory of the symmetric group.

Next, unitary property testing questions were also considered by Aharanov et al. in \cite{aharonov2022quantum}, who introduced a model called quantum algorithmic measurement and studied the query complexity of problems in this model. Their techniques used to prove lower bounds in their model primarily rely on the combinatorics of Weingarten functions, which we will also discuss in \Cref{subsection:weingarten}.  However, their model allows for interaction between a prover and a verifier, whereas we discuss query lower bounds in the non-interactive setting. Similar sample complexity lower bounds for quantum machine learning problems were proven using Weingarten calculus techniques in the works of \cite{chen2021hierarchy} and \cite{chen2022exponential}.

Next, Kretschmer \cite{Kretschmer2021quantumsupremacy} as well as Copeland and Pommersheim \cite{copeland2021quantum} also studied query complexity problems where the oracle does not necessarily encode a Boolean function. In particular, \cite{Kretschmer2021quantumsupremacy} studied the quantum query complexity of the heavy output generation problem and is motivated by recently quantum supremacy experiments, and \cite{copeland2021quantum} studied problems where the set of possible oracles form a representation of a group. However, there is a substantial difference between lower bound techniques in our works, which relies on linear programming duality in \cite{Kretschmer2021quantumsupremacy} and group character theory in \cite{copeland2021quantum}. 

Finally, we also note that Gur et al. \cite{gur2021sublinear} consider the query complexity of estimating the von Neumann entropy to a multiplicative factor of $\alpha > 1$, promised that the entropy of the given quantum state is not too small. They established a $\Omega(d^{\frac{1}{3 \alpha^2}})$ lower bound on the query complexity by reduction to the collision problem lower bound of \cite{kutin2003quantum}. On the other hand, the strongest lower bound we are able to prove is a $\Omega(d^{\frac{1}{4 \alpha}})$ lower bound by setting parameters $a = \frac{1}{\alpha} \log d$ and $b = \log d$ in the definition of our problem. However, our results are incomparable with their results since the oracle access models we consider are different.

\paragraph{Acknowledgments.} We thank Joshua Grochow for helpful conversations about invariant theory. We thank Shivam Nadimpalli for suggesting the problem about testing $k$-locality of a Hamiltonian. We thank Kunal Marwaha and Gregory Rosenthal for helpful discussions. We thank the reviewers for their feedback. A.S. is supported by an NSERC Canada Graduate Scholarship.  H.Y. is supported by AFOSR award FA9550-21-1-0040, NSF CAREER award CCF-2144219, and the Sloan Foundation.

\section{Preliminaries}
\label{sec:prelims}

\subsection{Testers for Unitary Properties}

Let $\cal{P} = (\cal{P}_{yes},\cal{P}_{no})$ be a $d$-dimensional unitary property. A quantum algorithm is a \emph{tester} for $\cal{P}$ if the following holds:
\begin{enumerate}
    \item If $U \in \cal{P}_{yes}$, then the algorithm makes queries to $U$ and accepts with probability $2/3$.
    \item If $U \in \cal{P}_{no}$, the algorithm makes queries to $U$ and accepts with probability at most $1/3$.
\end{enumerate}
A quantum algorithm is a \emph{$\QMA$ tester} for $\cal{P}$ if the following holds:
\begin{enumerate}
    \item If $U \in \cal{P}_{yes}$, then there exists a \emph{quantum proof state} $\ket{\psi}$ such that the algorithm on input $\ket{\psi}$, makes queries to $U$, and accepts with probability at least $2/3$.
    \item If $U \in \cal{P}_{no}$, then the algorithm given any proof state as input, making queries to $U$, accepts with probability at most $1/3$.
\end{enumerate}
As usual in complexity theory, the probabilities $2/3$ and $1/3$ can be set to any constants $a$ and $b$ as long as $a > b$ without changing any of the arguments that follow.

\subsection{The Guessing Lemma}

A simple but useful way to prove query lower bounds on $\QMA$ testers is to remove the proof state via the ``Guessing Lemma'', used in~\cite[Lemma 5]{aaronson2011impossibility} and ~\cite[Lemma 17]{aaronson2020quantum}:

\begin{lemma}
Suppose there is a $\QMA$ tester for a property $\cal{P}$ that makes $T$ queries and receives an $m$-qubit proof. Then there is a (standard) tester for the $\cal{P}$ that makes $O(mT)$ queries, receives no proof state, and satisfies
\begin{itemize}
    \item For all $U \in \mathcal{P}_{yes}$ the tester accepts with probability at least $2^{-m}$.
    \item For all $U \in \mathcal{P}_{no}$, the tester accepts with probability at most $2^{-10m}$.
\end{itemize}
\label{lem:guessing_lemma}
\end{lemma}

\subsection{Approximation Theory and Laurent Polynomials}

When we say that a polynomial is degree-$T$, we mean that it has degree \emph{at most} $T$. We say that a polynomial $p(z_1,\ldots,z_d,z_1^*,\ldots,z_d^*)$ with complex coefficients is \emph{self-adjoint} if $p(z_1, \ldots, z_d, z_1^*, \ldots, z_d^*) = p(z_1, \ldots, z_d, z_1^*, \ldots, z_d^*)^*.$ We say that $p$ is \emph{symmetric} if applying a permutation $\pi: [d] \to [d]$ to the variables $z_i$ and $z_i^*$ leaves $p$ unchanged. 

We record here some useful facts about polynomials, univariate polynomial approximation, and Laurent polynomials. %

\begin{lemma}[Fundamental Theorem of Algebra]
A real-valued degree-$T$ univariate polynomial has at most $T$ zeros.
\label{lem:fta}
\end{lemma}

\begin{lemma}[Markov brother's inequality \cite{markov1889question}]
Let $p(x)$ be a degree-$n$ real-valued polynomial. If $|p(x)| \leq H$ on the interval $[a,b],$ then for all $x \in [a,b]$, 
$$|p'(x)| \leq \frac{2H n^2}{b-a}.$$
\label{lem:markov_inequality}
\end{lemma}

\begin{lemma}[Paturi's bound \cite{paturi92degree}]
If $p$ is a degree-$d$ real polynomial satisfying $|p(x)| \leq 1$ for all $|x| \leq 1$, then for all $|x| \leq 1 + \mu$ we have 

$$|p(x)| \leq \exp(2d \sqrt{2 \mu + \mu^2}).$$
\label{lem:paturi}
\end{lemma}

A Laurent polynomial $p$ in variables $x_1,\ldots,x_k$ is a polynomial in the variables $x_1,\ldots,x_k$ and $x_1^{-1},\ldots,x_k^{-1}$. We say that a univariate Laurent polynomial $p(z)$ is \emph{symmetric} if $p(z) = p(z^{-1})$. The following fact about Laurent polynomials was stated as \cite[Lemma 14]{aaronson2020quantum}.

\begin{lemma}
Suppose $p(z)$ is a symmetric Laurent polynomial. Then there exists a univariate polynomial $q$ such that $p(z) = q(z + \frac{1}{z})$.
\label{lem:symmetric_laurent}
\end{lemma}

We require a slight modification of Lemma \ref{lem:symmetric_laurent}. %

\begin{lemma}
If $p(x, x^*)$ is a self-adjoint degree-$d$ polynomial satisfying $p(x, x^*) = p(x^*, x)$, there exists a real-valued univariate polynomial $q$ of degree $d$ with the property that $q(x + x^*) = p(x, x^*)$ when $x$ is restricted to the unit circle. 
\label{lem:symmetric_hermitian}
\end{lemma}

\begin{proof}
If $x$ is restricted to the unit circle, then $x x^* = 1$. Let $q(x) = p(x, \frac{1}{x})$. We claim that $q$ is a symmetric Laurent polynomial. Since $p$ is self-adjoint and $p(x,x^*) = p(x^*, x)$, $$q(x) = p \left (x, \frac{1}{x} \right ) = p(x, x^*) = p(x^*, x) = p \left (\frac{1}{x}, x \right) = q\left (\frac{1}{x} \right ),$$

and $q$ has real coefficients. Hence, by Lemma \ref{lem:symmetric_laurent}, there exists a polynomial $r$ such that $p(x, x^*) = q(x) = r(x + \frac{1}{x}) = r(x + x^*)$.
\end{proof}

We will also use the following bounds on the cosine function, which follow from elementary calculus. %

\begin{lemma}
For all $|x| \leq 2,$ $\frac{x^2}{3} \leq 1 - \cos x \leq \frac{x^2}{2}.$
\label{lem:cosine_lemma}
\end{lemma}

\begin{proof}
By the Taylor series expansion, $1 - \frac{x^2}{2} \leq \cos x \leq 1 - \frac{x^2}{2} + \frac{x^4}{24}.$ Therefore, $1 - \cos x \leq \frac{x^2}{2}.$ 

For the lower bound, observe that if $|x| \leq 2,$ then $x^4 \leq 4x^2.$ Therefore, $\cos x \leq 1 - \frac{x^2}{2} + \frac{x^2}{6} = 1 - \frac{x^2}{3}.$ Hence, $\frac{x^2}{3} \leq 1 - \cos x$ in this range. 
\end{proof}

\subsection{Distance between Quantum States}

Let $\rho$ and $\sigma$ be $d \times d$ density matrices.

\begin{definition}
    The trace distance between $\rho$ and $\sigma$ is defined as $T(\sigma, \rho) = \frac{1}{2} || \rho - \sigma ||_1.$ The trace norm for a $d \times d$ Hermitian matrix $M$ is defined as $||M||_1 = \sum_{i=1}^d |\lambda_i|$ where $\lambda_i$ are the eigenvalues of $M$. 
\end{definition}

For pure states, which is when $\rho = \ket{\psi_1} \bra{\psi_1}$ and $\sigma = \ket{\psi_2} \bra{\psi_2}$ are rank one matrices, the trace distance and fidelity satisfy a well-known relation.

\begin{lemma}
    Given two pure states $\rho = \ket{\psi_1} \bra{\psi_1}$ and $\sigma = \ket{\psi_2} \bra{\psi_2}$, the trace distance satisfies $T(\rho, \sigma) = \sqrt{1 - |\braket{\psi_1 \vert \psi_2}|^2}.$ The quantity $|\braket{\psi_1 \vert \psi_2}|^2$ is also known as the fidelity between the states $\ket{\psi_1}$ and $\ket{\psi_2}.$
\end{lemma}

\section{The Generalized Polynomial Method}
\label{sec:genpoly}

The following Proposition, restated from the introduction, is the foundation for our generalized polynomial method:

{\renewcommand{\footnote}[1]{}\genpoly*}

\begin{proof}
A tester $A$ that queries an oracle $U$ can be written as a product of fixed unitary maps  $A_0,A_1,\ldots,A_T$ that don't depend on the oracle $U$, interleaved with controlled applications of $U$ and its inverse $U^*$ (denoted by $\mathrm{c}U$ and $\mathrm{c}U^*$ respectively). Let $\ket{\psi_t}$ denote the state of the circuit before the $t^{th}$ query, so that either $\ket{\psi_{t+1}} = A_{t+1} \mathrm{c}U \ket{\psi_t}$ or $\ket{\psi_{t+1}} = A_{t+1} \mathrm{c}U^* \ket{\psi_t}$. Write 
\[
\ket{\psi_t} = \sum_{b = 0}^1 \sum_{i=1}^d \sum_{k=1}^r p^t_{b,i,k}(U,U^*) \ket{b,i, k}
\]
where $b$ denotes the control qubit, $i$ denotes the query register, $k$ denotes the ancilla register, and $p^t_{b,i,k}(U,U^*)$ is some function of the matrix entries of $U$ and $U^*$.  

We claim that the amplitudes $p^t_{b,i,k}(U,U^*)$ are polynomials in the matrix entries of $U$ and $U^*$ of degree at most $t-1$. To show this, we proceed by induction on the number of queries $t$. This is clearly true for $t = 1$ since $\ket{\psi_1}$ is some fixed state independent of $U$, so $p^1_{b,i,k}(U)$ are constants for all $b,i,k$. Now assume the claim to be true for $\ket{\psi_{t}}.$

Suppose the $t^{th}$ oracle call is to $\mathrm{c}U$. Then
\begin{align*}
\mathrm{c}U \ket{\psi_{t}} &= \sum_{i=1}^d \sum_{k=1}^r p^t_{0,i,k}(U,U^*) \ket{0} \otimes \ket{i} \otimes \ket{k} + \sum_{i=1}^d \sum_{k=1}^r p^t_{1,i,k}(U,U^*) \ket{1} \otimes U\ket{i} \otimes \ket{k} \\
&= \sum_{j=1}^d \sum_{k=1}^r p^t_{0,j,k}(U,U^*) \ket{0,j,k} + \sum_{j=1}^d \sum_{k=1}^r \left (\sum_{i=1}^d U_{j,i} \,\, p^t_{1,i,k}(U,U^*)\right ) \ket{1,j,k}
\end{align*}
Hence, since $\ket{\psi_{t+1}} = A_{t+1} \mathrm{c}U \ket{\psi_t}$ for some linear map $A_{t+1},$ then $p^{t+1}_{i,k}(U,U^*)$ are linear combinations of the  polynomials $p^t_{0,i,k}(U,U^*)$ and $\sum_{i=1}^d U_{j,i} \,\,  p^t_{1,i,k}(U,U^*)$. Hence if $p^t_{b,i,k}(U,U^*)$ have degree at most $t - 1$, $p^{t+1}_{b,i, k}(X)$ have degree at most $t$. The same argument holds if the oracle call was to $\mathrm{c}U^*$.

Thus the amplitudes of the final state of the algorithm can be expressed as polynomials of degree at most $T$. The proposition follows since the acceptance probability of the circuit is given by a measurement of the ancilla qubits and seeing if the string given lies in some set $S$, which is a sum $p(U,U^*) = \sum_{(b,i,k) \in S} |p^{T+1}_{b,i,k}(U,U^*)|^2.$ This is a degree-$2T$ self-adjoint polynomial since each $p^{T+1}_{i,k}(U,U^*)$ has degree at most $T$. %
\end{proof}

We also now show the following, restated from the introduction.

\genpolyinversion*

\begin{proof}
Let $p$ be the polynomial from \Cref{prop:basic-polynomial} for the property $\mathcal{P}$. Define $$q(U, U^*) = \frac{p(U, U^*) + p(U^*, U)}{2}.$$ Clearly, $q(U, U^*) = q(U^*, U)$ and the degree of $q$ is no more than the degree of $p$. Furthermore, since the property $\mathcal{P}$ is closed under inversion, $q(U, U^*) \geq a$ if $p(U, U^*) \geq a$ and $q(U, U^*) \leq b$ if $p(U, U^*) \leq b.$ 
\end{proof}

Then, when applying \Cref{prop:basic-polynomial} to a $T$-query tester for a property $\cal{P} = (\cal{P}_{yes},\cal{P}_{no})$, we have that there exists a degree-$2T$ polynomial $p$ such that if $U \in \cal{P}_{yes}$, then $p(U,U^*) \geq 2/3$, and if $U \in \cal{P}_{no}$, then $p(U,U^*) \leq 1/3$. Furthermore, by \Cref{prop:basic-polynomial-inv}, we can further assume that $p(U, U^*) = p(U^*, U)$ for properties that are closed under inversion. However, proving degree lower bounds on $p$ directly is difficult for general properties $\cal{P}$. As mentioned in the introduction, we focus on properties that obey certain symmetries in order to further simplify the polynomial $p$. For example, the acceptance probability $p$ corresponding to symmetric classical properties of boolean strings can be averaged to a univariate polynomial $q$ using Minsky-Papert symmetrization~\cite{beals2001quantum}.

Recall the definition of $G$-invariant properties and the invariant ring as stated in \Cref{defn:g-inv-property} and \Cref{defn:invariants} in the introduction. The following observation states that testers for $G$-invariant properties give rise to low-degree polynomials in the invariant ring $\C[X,X^*]$ that decide the property:

\symmetrizedpolynomial*

\begin{proof}
Let $p(U,U^*)$ be the polynomial from \Cref{prop:basic-polynomial} corresponding to the tester. Define the function
\[
q(U,U^*) = \E_{g \sim \mu} p(gUg^{-1}, gU^* g^{-1})~.
\]
It is clear that $q(U,U^*)$ is a self-adjoint polynomial with degree at most $2T$. Furthermore by construction it belongs to the invariant ring $\C[X,X^*]^G_d$ because for all $h \in G$, 
\[
    q(hUh^{-1},hU^*h^{-1}) = \E_{g \sim \mu} p(hgUg^{-1}h^{-1}, hgU^* g^{-1}h^{-1}) = \E_{g \sim \mu} p(gUg^{-1}, gU^* g^{-1})=q(U,U^*)
\] 
where the second equality follows from the fact that the Haar measure $\mu$ is invariant under left multiplication. 

Finally, the stated bounds on the values of $q(U,U^*)$ hold because for all $U \in \cal{P}_{yes}$, the unitary $gUg^{-1}$ is also in $\cal{P}_{yes}$, and similarly for the \emph{no} instances.

\end{proof}

In this paper we focus on two subgroups of the unitary group $\unitary(d)$. The first is the full unitary group itself. The invariant ring in this case has an extremely simple description. The following result is a special case of a more general theorem due to Procesi~\cite{PROCESI1976306}, who computed the invariant rings of $n$-tuples of matrices under simultaneous conjugation by the classical groups. As we only need the case of 2 matrices (the unitary $U$ and its adjoint $U^*$) in this work, we specialize the original result. Firstly, given a permutation $\sigma \in S_n$, define $\Tr_\sigma(A_1, \dots, A_n) = \prod_{C \in C(\sigma)} \Tr(\prod_{j \in C} A_j)$ where $C(\sigma)$ is the set of disjoint cycles of $\sigma$.

\begin{theorem}[{{\cite[Section 11]{PROCESI1976306},  \cite[Chapter 4]{kraft1996classical}}}]
Let $\C[X, X^*]^G_d$ be the invariant ring corresponding to the group $G = \unitary(d)$. Then all homogenous degree-$r$ polynomials $f \in \C[X, X^*]^G_d$ can be written as a linear combination of invariants of the form $\Tr_\sigma(A_1, \dots, A_r),$ where each $A_i = X$ or $X^*$ and $\sigma$ is a permutation in $S_r$. All invariants of degree $\leq r$ are linear combinations of homogenous invariants of degree $\leq r.$  %
\label{thm:unitarily_invariants}
\end{theorem}

We therefore get the following result (restated from the introduction) from combining \Cref{prop:basic-polynomial-symmetrized} and \Cref{thm:unitarily_invariants} for testing a unitarily invariant property.

{\renewcommand\footnote[1]{}\unitarilyinvariant*}

\begin{proof}
By \Cref{prop:basic-polynomial-symmetrized}, there exists a polynomial $p$ of degree at most $\leq 2T$ in the invariant ring $\C[U, U^*]^G$ with the property that $p$ distinguishes between \textit{yes} and \textit{no} instances. Furthermore, by \Cref{thm:unitarily_invariants}, $p$ is a linear combination of polynomials of the form $\Tr_\sigma(A_1, \dots, A_r)$ where $\sigma$ is a permutation on $r \leq 2T$ elements and each $A_i$ is $U$ or $U^*.$ Since $U U^* = I$, then each $\Tr_\sigma(A_1, \dots, A_r)$ is a product of terms of the form $\Tr(U^p)$ or $\Tr((U^*)^q)$ for $p, q \leq r$. Observe that each generator $\Tr(U^p) = \sum_{i=1}^d z_i^p$ is a power sum symmetric polynomial in the eigenvalues $z_i$ of $U$. Hence, since each term $\Tr_\sigma$ is a polynomial in the eigenvalues $(z_1, \dots, z_d)$ of $U$ and their conjugates of degree $\leq 2T$, that is symmetric under permutations of $(z_1, \dots, z_n)$ or $(z_1,^*, \dots, z_n^*)$, $p$ satisfies the same property since $p$ is a linear combination of the polynomials $\Tr_\sigma.$ 
\end{proof}

A natural subclass of unitarily invariant properties we consider are \emph{unitarily invariant subspace properties}, where the oracles $U$ are reflections about a subspace, i.e., $U = I - 2\Pi$ where $\Pi$ is the projection onto a subspace $S \subseteq \C^d$. 
Here, we get an analogue of the Minsky-Papert symmetrization technique originally applied to quantum query complexity in \cite[Lemma 3.2]{beals2001quantum}, since the acceptance probability of the quantum algorithm is reduced from a multivariable polynomial to a univariate polynomial.

\begin{lemma}
Suppose there is a $T$-query tester for a unitarily invariant subspace property $\mathcal{P}$. Then there exists a degree at most $2T$ \emph{univariate} polynomial $p(k)$ that for integer $k \in \{0,1,\ldots,d\}$ is equal to the acceptance probability of the tester when querying a unitary encoding a $k$-dimensional subspace.
\label{lem:unitarily_invariant_2}
\end{lemma}

\begin{proof}
The oracles corresponding to subspace properties satisfy $U = U^*$ and $U^2 = I$. Since for all integer $j$, the traces $\Tr(U^j)$ and $\Tr((U^*)^j)$ are either equal to $\Tr(I) = d$ or $\Tr(U) = \Tr(I - 2\Pi) = d - 2\dim(\Pi)$, \Cref{thm:unitarily_invariants} implies that the acceptance probability can be expressed as a degree-$2T$ polynomial in $d$ and $d - 2\dim(\Pi)$; since $d$ is constant, we can perform a change of variables to obtain a degree-$2T$ univariate polynomial in $\dim(\Pi)$ only.

\end{proof}

\subsection{An Alternative Proof of \Cref{lem:unitary_invariant_main_thm}}
\label{subsection:weingarten}

Here we give an alternative proof of \Cref{lem:unitary_invariant_main_thm} using random matrix theory. When the property is invariant under all unitaries, the polynomial $q$ constructed in the proof of \Cref{prop:basic-polynomial-symmetrized} is an average of a polynomial over the Haar measure over the full unitary group $\unitary(d)$. A consequence of Schur-Weyl duality is that there are relatively explicit \emph{Weingarten formulas} for averaging over the unitary group. For example,
if $g$ is sampled from the Haar distribution over $\unitary(d)$ , and $B$ was a matrix of the appropriate dimension, then we have the formula \cite[Section 7.5]{Brandao2021}:
$$\mathbb{E}_{g \sim \unitary(d)} \,\, (g^*)^{\otimes 2k} B g^{\otimes 2k} = \sum_{\sigma, \tau \in S_{2k}} \mathrm{Wg}(\sigma^{-1} \tau, d) \Tr(B P_\tau) P_\sigma,$$
where $\sigma$ and $\tau$ are permutations on $2k$ elements, $P_\sigma \ket{i_1, \dots, i_{2k}} = \ket{i_{\sigma^{-1}(1)}, \dots, i_{\sigma^{-1}(2k)}}$ is a permutation operator, and $\mathrm{Wg}$ is the Weingarten function that we define shortly. 

Let $p(U,U^*)$ denote the acceptance probability of a $T$-tester that makes queries to $U$. By construction (see \Cref{prop:basic-polynomial}),  $p$ is a self-adjoint degree-$2T$ polynomial where every monomial has at most degree $T$ in the $U_{i,j}$ variables and at most degree $T$ in the $U_{i,j}^*$ variables, so there exist a matrix $B$ such that 
\[
    p(U,U^*) = \Tr\left (B \Big( U^{\otimes T} \otimes (U^*)^{\otimes T} \Big) \right)~.
\]
To see this, note that by linearity it suffices to argue that a monomial of $p$ in the matrix entries of $U, U^*$ can be expressed in this form. However, this is simply a specific matrix element of the tensor product $U^{\otimes T} \otimes (U^*)^{\otimes T}$. Thus, leveraging the Weingarten formula, the polynomial $q$ can be written as
\begin{equation*} 
\begin{split}
q(U,U^*) &= \E_{g \sim \unitary(d)} p(gUg^*, gU^*g^*) \\ 
&= \E_{g \sim \unitary(d)} \Tr\left (B \Big( (gUg^*)^{\otimes T} \otimes (gU^*g^*)^{\otimes T} \Big) \right) \\
&= \E_{g \sim \unitary(d)} \Tr\left ( \Big((g^*)^{\otimes 2T} B g^{\otimes 2T} \Big)  \Big( U^{\otimes T} \otimes (U^*)^{\otimes T} \Big) \right) \\
&= \sum_{\sigma, \tau \in S_{2T}} \mathrm{Wg}(\sigma^{-1} \tau, d) \Tr(B P_\tau) \Tr(P_\sigma U^{\otimes T} \otimes (U^*)^{\otimes T}) \\
&= \sum_{\sigma, \tau \in S_{2T}} \mathrm{Wg}(\sigma^{-1} \tau, d) \Tr(B P_\tau) \Tr_{\sigma^{-1}}(U, \dots, U, U^*, \dots, U^*),
\end{split}
\end{equation*}
where we define $\Tr_\sigma(A_1, \dots, A_n) = \prod_{C \in C(\sigma)} \Tr(\prod_{j \in C} A_j)$ where $C(\sigma)$ is the set of disjoint cycles of $\sigma$. In particular, since $UU^* = I$, every expression of the form $\Tr_{\sigma^{-1}}(U, \dots, U, U^*, \dots, U^*)$ is a product of traces of the form $\Tr(U^p)$ or $\Tr((U^*)^q)$ for some integers $p$ and $q$, and hence we once again obtain the result of \Cref{lem:unitary_invariant_main_thm}. 

The Weingarten functions are defined as follows: If $\sigma$ is a permutation in $S_n,$ then $\mathrm{Wg}(\sigma, d)$ is defined as 
$$
    \mathrm{Wg}(\sigma, d) = \frac{1}{n!^2} \sum_{\lambda} \frac{\chi^{\lambda}(1)^2 \chi^{\lambda}(\sigma)}{s_{\lambda}(1^d)}
$$
summed over all partitions $\lambda$ of $n$, where $\chi^\lambda$ is the character of the $S_n$ representation associated to the partition $\lambda$, and $s_\lambda(1^d)$ is the Schur function for $\lambda$ evaluated at $x_1 = 1, \dots, x_d = 1$ (and all other variables zero). Notably, the value of the Weingarten function depends only on the cycle type of $\sigma$. The asymptotics of Weingarten functions for large $d$ according to the cycle type of $\sigma$ are well-studied (\cite{gu2013moments}, \cite[Theorem 3.6]{collins2016random}, \cite[Lemma 16]{montanaro2013weak}).

Although we don't use the Weingarten formulas further in this paper, we believe this alternative derivation of \Cref{lem:unitary_invariant_main_thm} via the use of Weingarten averaging formulas will be useful for future applications of the generalized polynomial method. For example, one could potentially prove bounds on the acceptance probabilities of algorithms by using asymptotic bounds on the Weingarten functions. Furthermore, this calculation can be easily extended to averages over different groups, such as product groups $\unitary(d) \times \unitary(d)$ when considering local unitary symmetries (see \Cref{sec:local_unitary}). We leave investigation of this for future work.

\section{Unitarily Invariant Property Testing}
\label{sec:applications}

We now illustrate some applications of the general theory developed in the previous section. The applications will make crucial use of  \Cref{lem:unitary_invariant_main_thm} and \Cref{lem:unitarily_invariant_2} for testing unitarily invariant properties.

\subsection{Testing Unitarily Invariant Subspace Properties}
\label{sec:symmetric_properties}
We first start with the class of unitarily invariant subspace properties. Recall that a subspace property $\cal{P}$ is one where all instances are reflections about some subspace, i.e., $U = I - 2\Pi$ where $\Pi$ is the orthogonal projector onto a subspace $S \subseteq \C^d$ (we say that $U$ encodes the subspace $S$). Such unitaries have eigenvalues $1$ or $-1$.

We will show that lower bounds for testing $\mathcal{P}$ follow immediately from lower bounds for testing symmetric properties $\cal{S}$ of classical strings, which means that the instances of $\cal{S}$ are $d$-bit strings, and the \emph{yes} instances are invariant under permutation of the coordinates (and similarly with the \emph{no} instances).

There is a one-to-one correspondence between unitarily invariant subspace properties and symmetric classical properties:
\begin{itemize}
    \item 
Given a unitarily invariant subspace property $\cal{P}$, we define $\cal{S}_{yes/no} = \{ \mathrm{spec}(U) : U \in \cal{P}_{yes/no} \}$, where $\mathrm{spec}(U)$ denotes the multiset of eigenvalues of $U$, interpreted as a $d$-bit string (with $+1$ mapped to $0$ and $-1$ mapped to $1$). The resulting classical property $\cal{S}$ is symmetric. 
\item Given a classical symmetric property $\cal{S}$, we define $\cal{P}_{yes/no} = \{ V^* D_x V : x \in \cal{S}_{yes/no}, V \in \unitary(d) \}$ where $D_x$ is a diagonal matrix with $(-1)^{x_i}$ on the $i$'th diagonal entry. The resulting unitary property $\cal{P}$ is a subspace property and is unitarily invariant. 
\end{itemize}
It is straightforward to see that this correspondence is a bijection. 

We now establish the following simple relation between the query complexity of the classical property $\cal{S}$ to that of the quantum proprety $\cal{P}$:

\symmetric*
\begin{proof}

From \Cref{lem:unitarily_invariant_2}, a $T$-query tester for a unitarily invariant subspace property $\cal{P}$ yields a degree-$2T$ polynomial $q(k)$ representing the success probability of the algorithm on unitaries $U = I - 2 \Pi$ encoding subspaces $\Pi$ of dimension $0 \leq k \leq d$. Thus the polynomial $q$ also decides the associated classical symmetric property $\cal{S}$, by considering $k$ as the Hamming weight of the associated string $\text{spec}(U)$. %
\end{proof}

We note that one can also prove \Cref{prop:symmetric} by observing that a $T$-query tester for a unitarily invariant subspace property $\cal{P}$ is also a $T$-query tester for the associated classical symmetric property $\cal{S}$.

We now mention some easy applications of \Cref{prop:symmetric}.

\begin{theorem}[Unstructured Search Lower Bound]
Any tester that decides whether an oracle $U$ is a reflection about some quantum state $\ket{\psi}$, i.e., $U = I - 2 \ketbra{\psi}{\psi}$, or is the identity $U = I$, must make $\Omega(\sqrt{d})$ queries to $U$.
\end{theorem}

\begin{proof}
Define $\cal{P}_{yes} = \{ I - 2\ketbra{\psi}{\psi} : \ket{\psi} \in \C^d \}$ and $\cal{P}_{no} = \{I \}$. Clearly $\cal{P} = (\cal{P}_{yes},\cal{P}_{no})$ is a unitarily invariant subspace property. The associated classical property $\cal{S}$ consists of \emph{yes} instances that are binary strings of Hamming weight $1$ (because the \emph{yes} instances of $\cal{P}$ have exactly one $-1$ eigenvalue) and the \emph{no} instances is the all zeroes string. This is essentially the Grover search problem, it is well-known via the standard polynomial method~\cite{beals2001quantum} that any polynomial deciding $\cal{S}$ requires $\Omega(\sqrt{d})$ queries, which implies $\Omega(\sqrt{d})$ query lower bound for property $\cal{P}$ by \Cref{prop:symmetric}. 

\end{proof}

We also consider the \emph{Approximate Dimension} problem, which for some integer parameter $ 0 \leq w \leq d$, distinguish between whether the subspace encoded by the oracle $U$ has dimension at least $2w$ (\emph{yes} instances) or at most $w$. This is a unitarily invariant subspace property testing problem, as conjugating a reflection $U = I - 2\Pi$ by any unitary $V$ leaves the dimension of the encoded subspace unchanged. This generalizes the classical Approximate Counting problem, which is to determine whether the Hamming weight of an input string is at most $w$ or at least $2w$. Again leveraging the standard polynomial method we obtain the following lower bound:

\begin{theorem}[Approximate Dimension Lower Bound]
Any tester that decides between whether a unitary encodes a subspace of dimension at least $2w$ or at most $w$ requires $\Omega(\sqrt{\frac{d}{w}})$ queries.
\end{theorem}

\begin{proof}
The associated classical property, Approximate Counting, is where the \emph{yes} instances correspond to strings with Hamming weight at least $2w$ and the \emph{no} instances have Hamming weight at most $w$. By reduction to the Grover search problem, we get that the degree of any polynomial that decides Approximate Counting is at least $\Omega(\sqrt{d/w})$, which by \Cref{prop:symmetric} is also a lower bound on the number of queries needed to decide the Approximate Dimension problem. 

\end{proof}

We note that by using an appropriate modification of the quantum counting algorithm of Brassard et al. \cite{brassard1998quantum}, we obtain a matching upper bound.

\begin{proposition}
There exists a tester that using $O(\sqrt{\frac{d}{w}})$ queries and certifies whether or not a unitary $U = I - 2P$ encodes a subspace $S$ of dimension at least $2w$ or at most $w$. 
\label{prop:dim_counting_algorithm}
\end{proposition}

\begin{proof}
Prepare the maximally entangled state $\ket{\Phi}$ in $\C^d \otimes \C^d$ and observe that $\ket{\Phi}$ can be written as $\ket{\Phi} = \frac{1}{\sqrt{n}} \sum_{i=1}^n \ket{v_i} \ket{\overline{v_i}}$ for any basis $B = \{\ket{v_i}\}$ of $\C^d$. Hence, we can assume that $B = B_1 \cup B_2$ where $B_1$ is a basis for $S$ and $B_2$ is a basis for the orthogonal complement $S^{\perp}.$ 

Let $s = \dim S$, and $\ket{\Phi_S} = \frac{1}{\sqrt{s}} \sum_{\ket{v_i} \in B_1} \ket{v_i} \ket{\overline{v_i}}$ be maximally entangled over $S$ and $\ket{\Phi_{S^{\perp}}} = \frac{1}{\sqrt{d-s}} \sum_{\ket{v_i} \in B_2} \ket{v_i} \ket{\overline{v_i}}$ be maximally entangled over $S^{\perp}.$ Let $R = 2 \ket{\Phi} \bra{\Phi} - 1$ be the reflection around the maximally entangled state. Observe that $\ket{\Phi} \in \text{span}\{\ket{\Phi_S}, \ket{\Phi_{S^{\perp}}}\}$ and furthermore $(U \otimes I) \ket{\Phi_S} = -\ket{\Phi_S},$ and $(U \otimes I) \ket{\Phi_{S^{\perp}}} = \ket{\Phi_{S^{\perp}}}.$ Hence, by the analysis of Grover search, the operator $G = R (U \otimes I)$ is a rotation in the plane spanned by $\ket{\Phi_S}$ and $\ket{\Phi_{S^{\perp}}}$ by an angle of $2 \theta$, where $\sin^2 \theta = \frac{s}{d}.$ Hence applying phase estimation with the operator $G$ and the state $\ket{\Phi}$ as input, produces an estimate of the angle $\theta$ and hence the dimension $s$ since the eigenvalues of $G$ are $e^{\pm 2 i \theta}.$ By the analysis of phase estimation, at most $O( \sqrt{\frac{d}{w}})$ oracle calls to $G$ can be used to get an estimate $\tilde{s}$ of $s$ satisfying $0.9 s \leq \tilde{s} \leq 1.1s$, and hence we can distinguish whether or not $s \geq 2w$ or $s \leq w$ with access to this estimate $\tilde{s}.$
\end{proof}

Aaronson, et al.~\cite{aaronson2020quantum} showed that having access to a quantum proof does not help reduce the query complexity of the classical Approximate Counting problem, unless the proof state is very large (at least $w$ qubits). Since a $\QMA$ tester for Approximate Dimension is automatically a $\QMA$ tester for the Approximate Counting problem, the lower bound proved by~\cite{aaronson2020quantum} directly gives the following:

\appxdimqma*

We note that \Cref{prop:dim_counting_algorithm} shows that \Cref{thm:apx_dim_qma} is tight in the regime where the quantum proof satisfies $m = o(w)$. Otherwise, we conjecture that providing $O(w)$ copies of the mixed state $\rho$, where $\rho$ is maximally mixed over the the hidden subspace $S$ and performing swap tests to estimate the purity of $\rho$, suffices to solve the approximate dimension problem. However, this algorithm seems to require an unentanglement guarantee on the witness, which does not immediately show that it is a $\QMA$ tester. We leave this investigation to further work.

\subsection{Recurrence Times of Unitaries}
\label{sec:recurrence}

We now turn to analyzing the problem of testing recurrence times of unitaries. This corresponds to analyzing unitarily invariant properties that are not subspace properties. As mentioned in the introduction, one cannot directly use lower bounds on a related classical property testing problem; instead we have to make full use of the generalized polynomial method.

Recall the Recurrence Time problem defined in the introduction:

\recurrencedef*

\paragraph{Upper Bound.} We first present an upper bound on the query complexity of the Recurrence Time Problem.

\recurrenceub*

\begin{proof}

Fix an integer $t$. The goal is to determine whether there is an eigenvector $\ket{\psi}$ of $U^t$ such that the phase $e^{2\pi i \varphi}$ associated with $\ket{\psi}$ is more than $\epsilon$ far from $1$, or, equivalently, whether the phase $e^{2 \pi i \theta}$ of $\ket{\psi}$ with respect to $U$ satisfies $2 \pi i \theta t$ being more than $\epsilon$ away from an integer multiple of $2 \pi i$. If there is no such eigenvector, then $t$ is an $\epsilon$-recurrence time for $U$. To find such an eigenvector, we first prepare the $d$-dimensional maximally entangled state $\ket{\Phi} = \frac{1}{\sqrt{d}} \sum_j \ket{j} \ket{j}$. Let $\{ \ket{\psi_j} \}_j$ denote an eigenbasis for $U$ with associated eigenvalues $\{ e^{2\pi i \theta_j} \}_j$; then we have that $\ket{\Phi}$ can be equivalently expressed as 
\[
    \ket{\Phi} = \frac{1}{\sqrt{d}} \sum_j \ket{\psi_j} \ket{\overline{\psi_j}}
\]
where $\ket{\overline{\psi_j}}$ denotes the complex conjugate of $\ket{\psi_j}$ with respect to the standard basis. We perform phase estimation on the first register of $\ket{\Phi}$ with respect to $U$ to estimate the phases $\theta_j$ up to $\pm \epsilon/8t$ additive error, with success probability at least, say, $99\%$. The analysis of~\cite[Section 5.2.1]{nielsen2010quantum} shows that this requires $O(t/\epsilon)$ calls to the unitary $U$. 

The state then has the form
\[
   \ket{\Phi'} = \frac{1}{\sqrt{d}} \sum_{j,k} \alpha_{j,k} \ket{\widetilde{\theta}_j^{(k)}}\ket{\psi_j} \ket{\overline{\psi_j}}
\]
where $\widetilde{\theta}_j^{(k)}$ are the estimates of $\ket{\theta_j}$ from phase estimation, and $\alpha_{j,k}$ are the amplitudes of each of the estimates. As mentioned, the sum of squares of amplitudes $\alpha_{j,k}$ such that the estimate $\widetilde{\theta}_j^{(k)}$ differs from $\theta_j$ by more than $\epsilon/8t$ (we call such an estimate $\widetilde{\theta}_j^{(k)}$ \emph{bad}, otherwise it is \emph{good}) is at most $1\%$. 

We now perform amplitude amplification in order to identify whether there is an estimate $\ket{\widetilde{\theta}^{(k)}_j}$ such that 
\begin{equation}
    \label{eq:phase-epsilon}
    \Big |e^{2\pi i t \widetilde{\theta}^{(k)}_j} - 1 \Big | \geq \epsilon/2~.
\end{equation}
The amplitude amplification procedure will alternate between applying a phase on the $\ket{\widetilde{\theta}^{(k)}_j}$ states satisfying~\eqref{eq:phase-epsilon}, and reflecting about the state $\ket{\Phi'}$. Let $P$ be the projector onto estimates satisfying ~\eqref{eq:phase-epsilon} in the first register. %

In the \textit{no} case, either phase estimation fails, which occurs with occurs at most $1\%$ probability, or there is an estimate for a phase $\ket{\theta_j^k}$ that is $\epsilon$-far away from 1.  We claim that amplitude amplification finds a phase satisfying the condition ~\eqref{eq:phase-epsilon} with constant probability. If there is a phase $\theta_j$ such that $|e^{2\pi i t \theta_j} - 1| \geq \epsilon$, then the good estimates $\widetilde{\theta}_j^{(k)}$ of $\theta_j$ satisfy
    \begin{align*}
    \Big |e^{2\pi i t \widetilde{\theta}^{(k)}_j} - 1 \Big | &\geq  \Big |e^{2\pi i t \theta_j} - 1 \Big | -  \Big |e^{2\pi i t \widetilde{\theta}^{(k)}_j} - e^{2\pi i t \theta_j} \Big |  & \text{(triangle inequality)} \\
    &\geq \epsilon - 4 t \Big |\theta_j - \widetilde{\theta}_j^{(k)} \Big| & (\text{calculus}) \\
    &\geq \epsilon - \epsilon/2 = \epsilon/2.
    \end{align*}

Thus when phase estimation succeeds, the initial state satisfies $|P \ket{\Phi'}| \geq \frac{1}{\sqrt{d}}$ and hence when the marked phase is unique, $O(\sqrt{d})$ iterations suffice to boost the probability on the marked phase to constant probability. In the case where there are multiple phases satisfying the condition \eqref{eq:phase-epsilon}, then we run the amplitude amplification algorithm for $\sqrt{d}, \sqrt{\frac{d}{2}}, \sqrt{\frac{d}{4}}, \dots,$ iterations and so on. Since a marked item can be found with $O(\sqrt{\frac{d}{k}})$ iterations if there are $k$ marked items, the binary search procedure  terminates in $O(\sqrt{d})$ iterations and finds a phase that is $\epsilon$-far from 1 with constant probability.

Otherwise, in the \textit{yes} case where $U^t = I,$ the analysis of \cite[Section 5.2.1]{nielsen2010quantum} shows that the phase estimation algorithm produces exact values for the phase register $\ket{\theta_j^{k}}$ since the phases are integer multiples of $2 \pi$. Hence, the initial state in the amplitude amplification algorithm has no overlap with the subspace satisfying ~\eqref{eq:phase-epsilon} and  hence the final state after amplification remains the same as the initial state up to a global phase. Thus, the algorithm never finds a phase $\epsilon$ far from 1.

Hence in $O(\frac{t \sqrt{d}}{\epsilon})$ iterations we are able to distinguish between the \textit{yes} and \textit{no} cases with constant bias.

\end{proof}

\paragraph{Lower Bound.} Using the generalized polynomial method for unitarily invariant properties, we prove the following query lower bound for the Recurrence Time problem. 

\recurrence*

Before doing this, we introduce a useful symmetrization lemma we use to reduce the number of variables of the polynomial.

\begin{lemma}
Let $q(z_1, \dots, z_d)$ be the polynomial obtained from \Cref{lem:unitary_invariant_main_thm} for the acceptance probability of a $T$-query algorithm on the Recurrence Time problem. 

Let $D(p, z)$ be a distribution on $d$-dimensional diagonal unitaries where each diagonal entry is chosen to be equal to $z = e^{i \theta}$ with probability $p$ and otherwise equal to 1 with probability $1 - p.$ Then the expected value $$r(p, z) = \mathbb{E}_{(z_1, \dots, z_d) \sim D(p, z)}[q(z_1, \dots, z_d)]$$

is a self-adjoint polynomial of degree at most $2 \deg q.$
\label{lem:recurrence_time_lemma}
\end{lemma}

\begin{proof}
Recurrence Time is a unitarily invariant property as $U^t = I$ iff $(VUV^*)^t = I$ and also $\|U^t - I\| \geq \epsilon$ iff $\|(VUV^*)^t - I\| \geq \epsilon$ for any unitaries $U$ and $V$. Therefore, \Cref{lem:unitary_invariant_main_thm} guarantees that the acceptance probability of a $T$-query algorithm for the problem can be written as a degree $\leq 2T$ self-adjoint polynomial in the eigenvalues of $U$. 

Since $q$ is a self-adjoint polynomial defined on the unit circle, $p$ can be expanded in a basis of binomials $z_I z_J^* + z_J z_I^*$ where $I \subseteq [n], J \subseteq [n]$, $I \cap J = \emptyset$, $z_I = \prod_{i \in I} z_i$ and $z_J = \prod_{j \in J} z_j^*$. The expected value of each binomial under when the eigenvalues are chosen according to $D(p, z)$ is then

$$\sum_{k_1 = 0}^{|I|} \sum_{k_2 = 0}^{|J|} \binom{|I|}{k_1} \binom{|J|}{k_2} p^{k_1 + k_2} (1-p)^{|I| + |J| - k_1 - k_2} [z^{k_1 - k_2} + (z^*)^{k_1 - k_2}].$$

Hence the expected value  $r(p, z) = \mathbb{E}_{(z_1, \dots, z_d) \sim D(p, z)}[q(z_1, \dots, z_d)]$ is a polynomial of degree at most $2 \deg q$ with the property that $r(p, z) = r(p, z^*).$ 
\end{proof}

We are now ready to prove \Cref{thm:recurrence}.

\begin{proof}
By \Cref{lem:recurrence_time_lemma}, if there was a $T$-query algorithm for the Recurrence Time problem, $r(p, z)$ is a polynomial of degree at most $4T$ that represents the expected probability the algorithm accepts on the distribution $D(p, z).$ We now lower bound the degree of $q$ by lower bounding the degrees of $p$ and $z$ separately.

Firstly we lower bound the degree of $p$ by fixing $z' =  \exp(\frac{4 \pi i \epsilon}{t}).$ For this value of $z'$, $r_1(p) = r(p, z')$ is a real-valued univariate polynomial with the property that $r_1(0) \geq \frac{2}{3}$ (since if $p = 0$ we are given the identity unitary as input).

Otherwise if $p = \frac{2}{d}$, the number of eigenvalues equal to $z'$ is a binomial random variable with $d$ trials and success probability $p = \frac{2}{d},$ which for sufficiently large $d$ is approximately Poisson distributed with mean equal to $2.$ Hence, for sufficiently large $d$, the probability that the input is the identity is at most $e^{-2}.$  If not, the input is a \textit{no} instance, since in this case $U^t$ would have an eigenvalue equal to $(z')^t = \exp(4 \pi i {\epsilon})$, and therefore

$$\|U^t - I\| = \sqrt{2 - 2 \cos(4 \pi \epsilon)} \geq \sqrt{\frac{2}{3} (4 \pi \epsilon)^2} \geq 10 \epsilon,$$ 

by assumption that $|4 \pi \epsilon| \leq 2$ and \Cref{lem:cosine_lemma}.

Hence we have that
$$r_1 \Big (\frac{2}{d} \Big ) \leq e^{-2} + \frac{1}{3} (1 - e^{-2}) \leq \frac{1}{2}.$$

Therefore, $r_1$ satisfies the properties that $0 \leq r_1(p) \leq 1$ for all $0 \leq p \leq 1$, $r_1(0) \geq \frac{2}{3}$, and $r_1(\frac{2}{d}) \leq \frac{1}{2}.$ By Markov's inequality (\Cref{lem:markov_inequality}), the inequality $$\frac{d}{12} \leq 2 (\deg r_1)^2,$$ must be satisfied, so $\deg r_1 \geq \Omega(\sqrt{d}).$

Now we lower bound by degree of $z$ by fixing $p = \frac{2}{d}$ and consider the polynomial $r_2(z) = r(p, z).$ Observe that $r_2$ has the property that $r_2(z^*) = r_2(z)$. Hence  \Cref{lem:symmetric_hermitian} applies and we can assume $r_2(z) = s_2(z + z^*)$ for some real-valued polynomial $s_2$ of the same degree. Observe that the $s_2$ is bounded by one and defined on the interval $[-2, 2].$ Furthermore, for $z_1 = 1,$ we have $s_2(z_1 + z_1^*) = r(p, z_1) \geq \frac{2}{3},$ and otherwise for $z_2 = \exp(\frac{4 \pi i \epsilon}{t}),$ we have from the previous calculation that $s_2(z_2 + z_2^*) =  r(p, z_2) \leq \frac{1}{2}.$ Since by \Cref{lem:cosine_lemma} and the assumption that $\epsilon \leq \frac{1}{2\pi},$

$$|(z_1 + z_1^*) - (z_2 + z_2^*)| = |2 - 2 \cos (\frac{4 \pi \epsilon}{t})| \leq \frac{16 \pi^2 \epsilon^2}{t^2},$$ we conclude that the derivative of $s_2$ must satisfy $|s_2'(x)| \geq \frac{t^2}{96 \pi^2 \epsilon^2}.$
for some point $x \in [2 - 2 \cos (\frac{4 \pi \epsilon}{t}), 2]$. Hence, by Markov's inequality, we have that the degree of $s_2$ must satisfy:

$$\frac{t^2}{96 \pi^2 \epsilon^2} \leq \frac{2 (\deg s_2)^2}{4},$$

Hence, $\deg r_2 = \deg s_2 \geq \Omega(\frac{t}{\epsilon}).$

Therefore, combining the two lower bounds implies that there must be monomials in $r(p, z)$ with $p$-degree at least $\Omega(\sqrt{d})$ and $z$-degree at least $\Omega(\frac{t}{\epsilon}).$ Hence, since $\deg r \leq 2 \deg q \leq 4 T$ where $T$ is the query complexity of the algorithm, we obtain $T \geq \Omega(\max(\frac{t}{\epsilon}, \sqrt{d})).$
\end{proof}

We note that a similar lower bound can also be obtained using hybrid method of \cite{bennett1997strengths}. However, it is unclear whether the hybrid method can be used to obtain lower bounds in the $\QMA$ setting. We now modify the previous arguments to show that the Recurrence Time problem remains hard even when the tester receives a quantum proof that is not too large. %

\recurrenceqma*

\begin{proof}
If $m \geq \Omega(d)$ we are done, so we assume that $m \leq o(d)$.

Let $r(p, z)$ be obtained from \Cref{lem:recurrence_time_lemma} and \Cref{lem:guessing_lemma}. Again we will lower bound the degree of $q$ by considering the degree of $p$ and $z$ separately. 

We first lower bound the degree of $p$ by fixing $z' = \exp(\frac{4 \pi i \epsilon}{t}).$ For this value of $z'$, $r_1(p) = r(p, z')$ is a real-valued univariate polynomial with the property that $r(0) \geq 2^{-m}$. Otherwise for all $p \geq \frac{10 m}{d},$ we have the probability that there is there is no non-trivial eigenvalue is bounded by $(1 - p)^d \leq \exp(-dp) \leq \exp(-10m),$ and hence for all $\frac{10m}{d} \leq p \leq 1,$ we have

$$r_1(p) \leq \exp(-10m) + 2^{-10 m} (1 - \exp(- 10m)) \leq 2^{-9m}$$

since this value of $z'$ corresponds to a no instance. 

Therefore, if $y_0 = \frac{10m}{d}, y_1 = 1$ the polynomial

$$s_1(x) = 2^{9m} r_1 \left(\frac{y_0 - y_1}{2} (x - 1) + y_0 \right),$$

implies that the polynomial $s_1$ satisfies $\deg s_1 \leq \deg r_1,$ $|s_1(x)| \leq 1$ for all $|x| \leq 1$, and $s(1 +  \frac{2y_0}{y_1 - y_0}) \geq 2^{8m}.$ Observe that $\frac{2y_0} {y_1 - y_0} = \frac{20m}{d-10m} \leq \frac{40m}{d}$ by our assumption on $m$. Hence, applying Paturi's Lemma (Lemma \ref{lem:paturi}) with these conditions and $\mu = \frac{40m}{d}$, we obtain the inequality:

$$2^{8m} \leq \exp(2 (\deg s_1) \sqrt{\mu^2 + 2 \mu}) \leq \exp(4 (\deg s_1) \sqrt{\mu})$$

since by assumption $\mu \leq 2.$ Therefore, solving for $\deg s_1$ implies that $\deg r_1 \geq \deg s_1 \geq \Omega(\sqrt{md}).$

Now we lower bound the degree of $z$ by fixing $p = \frac{10m}{d}$ and consider the polynomial $r_2(z) = r(p, z).$ Observe that $r_2$ has the property that $r_2(z^*) = r_2(z)$. Hence Lemma \ref{lem:symmetric_hermitian} applies and we can assume $r_2(z) = s_2(z + z^*)$ for some real-valued polynomial $s_2$ of the same degree. Observe that since for any integer $j$, since any unitary whose only eigenvalue is $z = \exp(\frac{2 \pi i j}{t})$ is a \textit{yes} instance, we have $$r_2\left(\exp(\frac{2 \pi i j}{t}) \right) = s_2\left(2 \cos \frac{2 \pi j}{t} \right) \geq 2^{-m}.$$

Otherwise, there is at least one point in the interval $x \in (2 \cos \frac{2 \pi (j+1)}{t}, 2 \cos \frac{2 \pi j}{t}),$ where $s_2(x) \leq 2^{-9m},$ since all unitaries whose only eigenvalue is equal to $z = \exp(\frac{2 \pi i (j + 1/2)}{t})$ corresponds to a \textit{no} instance, which corresponds to the point $z + z^* = x = 2 \cos \frac{2 \pi (j + 1/2)}{t}.$ . 

Hence, $s_2(x) - (\frac{2^{-m} + 2^{-9m}}{2})$ has at least $\frac{t}{2}$ roots in the interval $[-2, 2]$ since $s_2 - (\frac{2^{-m} + 2^{-9m}}{2})$ changes sign at least $\frac{t}{2}$ times, which implies that the degree of $s_2$ at least $\frac{t}{2}$ by the Fundamental Theorem of Algebra (\Cref{lem:fta}). Therefore, since $\deg r_2 = \deg s_2,$ the degree of $r_2(z)$ is at least $\Omega(t).$

Finally, we consider the dependence on the error $\epsilon$. Furthermore, since $s_2(2) \geq 2^{-m}$ and $s_2(x) \leq 2^{-10m}$ for all $y_0 = 2 \cos \frac{2 \pi (1 - \epsilon)}{t}  \leq x \leq y_1 = 2 \cos \frac{4 \pi \epsilon}{t},$ we have that

$$s_3(x) = 2^{10m} s_2 \left(\frac{y_1 - y_0}{2} (x - 1) + y_1 \right),$$

satisfies $|s_3(x)| \leq 1$ for all $|x| \leq 1$, and that when $x = 1 + \frac{2 (2 - y_1)}{y_1 - y_0}$ we have $s_3(x) \geq 2^{9m}.$ By \Cref{lem:cosine_lemma}, $y_1 - y_0 \geq \frac{8 \pi^2}{3 t^2} + O(\epsilon)$ and $2 - y_1 \leq \frac{16 \pi^2 \epsilon^2}{t^2}.$ Hence, we may take $\mu \leq O(\epsilon^2)$ in Paturi's Lemma to conclude that $\deg s_3$ satisfies

$$2^{9m} \leq \exp(4 (\deg s_3) \sqrt{\mu}) = \exp(2 (\deg s_3) O(\epsilon)),$$

and hence $\deg r_2 = \deg s_2 \geq \deg s_3 \geq \Omega(\frac{m}{\epsilon}).$

Putting these bounds together, we conclude that either $m \geq \Omega(d),$ or otherwise, since $\deg r_1 \leq \deg r$ and $\deg r_2 \leq \deg r$, we have $$\max(\Omega(\sqrt{md}), \Omega(t), \Omega(\frac{m}{\epsilon})) \leq \deg r \leq 2 \deg q \leq O(mT),$$

which was the claimed bound.
\end{proof}

Observe that there is a weaker dependence on $\epsilon$ in our $\QMA$ lower bound, compared to the $\BQP$ lower bound for the Recurrence Time problem. We leave improving this dependence to further work.

We end with some brief observations about the \textsf{coQMA} query complexity of the problem where we provide a certificate for \textit{non}-recurrence. We note that the query complexity of the problem changes significantly in the \textsf{coQMA} setting compared to the \textsf{QMA} setting. Here, a valid certificate is an eigenvector $\ket{\psi}$ of $U^t$ with eigenvalue not  equal to $e^{i \theta}$ where $\theta \in [-\epsilon, \epsilon]$, and a quantum phase estimation can be used with $O(\frac{t}{\epsilon})$ queries to compute the the corresponding eigenvalue of $\ket{\psi}$ to $O(\epsilon)$ precision. In particular, there is no dependence on the dimension $d$. Therefore, in the setting where the recurrence time $t$ is constant, the unitary recurrence problem provides an exponential query complexity separation between \textsf{QMA} and \textsf{coQMA}.

\section{Invariance under Local Unitaries}

\label{sec:local_unitary}

\subsection{Local Unitary Invariants}

Recall from the introduction the definition of the local unitary group.

\localunitary*

As discussed in the introduction, $\LU$-invariance naturally captures the symmetry associated with entanglement properties of states and operators. \Cref{prop:basic-polynomial-symmetrized} implies that a $T$-query tester for an $\LU$-invariant property $\cal{P}$ gives rise to a degree-$2T$ polynomial $q$ belonging to the invariant ring $\C[X,X^*]^{\LU(d_1,d_2)}$ that decides $\cal{P}$, where $X,X^*$ represent the variables and their conjugates of matrices acting on $\C^{d_1} \otimes \C^{d_2}$.

As with the full unitary group case, it is possible to characterize the polynomial functions on matrices that are invariant under the local unitary group. The next theorem, due to Procesi \cite{PROCESI1976306} and Brauer \cite{brauer1937algebras}, presents such a characterization.

\begin{theorem}[Generators for $\LU$-invariant polynomials]
Let $\sigma, \tau$ be permutations on $k$ elements and let $R_\sigma, R_\tau$ be the corresponding permutation operators on $(\mathbb{C}^d)^{\otimes k}$. Then the homogenous degree $k$ part of the invariant ring $\C[X]^{\LU(d,d)}$, where $X$ represents the variables of a matrix acting on $\C^d \otimes \C^d$, is in the linear span of the polynomials $\Tr((R_\sigma \otimes R_\tau) X^{\otimes k})$,\footnote{The way the operators should be multiplied is as follows: if the $i$'th copy of $X$ acts on registers $A_i B_i$, then $R_\sigma$ permutes the $A_i$ registers and $R_\tau$ permutes the $B_i$ registers.} ranging over all permutations $\sigma, \tau \in S_k.$
\label{thm:local_invariant_polys}
\end{theorem}

As an aside, we note that \cite{biamonte2013tensor} has provided an interpretation of these invariants in terms of \textit{tensor networks}, which are a visual tool for representing high dimensional tensors. We also note that these invariants have been studied extensively in the pure mathematics and physics literature \cite{qiao2020local, grassl1998computing, turner2017complete}.

Ultimately, we would like to use \Cref{thm:local_invariant_polys} to prove query complexity lower bounds on $\LU$-invariant properties. However, this characterization of the $\LU$-invariant ring, while explicit, appears less simple to use than \Cref{thm:unitarily_invariants}. In the full unitarily invariant case, the generators $\Tr(R_\sigma \, X^{\otimes k})$ are symmetric polynomials depending only on the cycle structure of $\sigma$ and the eigenvalues of $X$. This information is sufficient for us to leverage tools from approximation theory to lower bound the degree the invariant polynomial. 

In contrast, it is not so clear how to make use of the quantities $\Tr((R_\sigma \otimes R_\tau) U^{\otimes k})$ for general unitaries $U$; for example we do not know if the traces can be expressed as a polynomial of some natural quantities (like how the eigenvalues of $U$ are natural linear-algebraic quantities) that capture some entanglement properties of $U$. 
However, there is a special case for which we can give a good characterization of the invariant polynomials, which is when $X$ is a projector onto a one-dimensional subspace:

\begin{theorem}[{{\cite[Theorem 22]{biamonte2013tensor}}}]
Let $\Pi = \ketbra{\psi}{\psi}$ be the projector onto a bipartite state $\ket{\psi} \in \C^d \otimes \C^d.$ Let $\rho = \sum_i \lambda_i \ketbra{v_i}{v_i}$ denote the reduced density matrix of $\rho$ on the first subsystem. Let $\sigma, \tau \in S_k$. Then  $\Tr((R_\sigma \otimes R_\tau) \Pi^{\otimes k})$ is a symmetric degree-$k$ polynomial in the eigenvalues $\lambda_i$ of $\rho$.
\label{thm:one_dim_local_unitaries}
\end{theorem}
\noindent The tuple of eigenvalues $(\lambda_1,\ldots,\lambda_d)$ is called the \emph{entanglement spectrum} of $\ket{\psi}$. This has the following consequence for $\LU$-invariant one-dimensional subspace properties:

\begin{lemma}
Let $\cal{P} = (\cal{P}_{yes},\cal{P}_{no})$ denote a $\LU$-invariant subspace property where the instances consist of reflections $U = I - 2\ketbra{\psi}{\psi}$ for some pure state $\ket{\psi} \in \C^d \otimes \C^d$. Suppose there is a $T$-query tester for $\cal{P}$ that accepts \emph{yes} instances with probability at least $a$ and \emph{no} instances with probability at most $b$. Then there exists a degree-$2T$ symmetric polynomial $p(\lambda_1,\ldots,\lambda_d)$ such that 
\begin{itemize}
    \item If $U \in \cal{P}_{yes}$, then $p(\lambda_1,\ldots,\lambda_d) \geq a$.
    \item If $U \in \cal{P}_{no}$, then 
    $p(\lambda_1,\ldots,\lambda_d) \leq b$.
\end{itemize}
Here, $p$ is evaluated at the entanglement spectrum $(\lambda_1,\ldots,\lambda_d)$ of the pure state $\ket{\psi}$ corresponding to $U$.
\label{lem:lu-one-dim}
\end{lemma}
\begin{proof}
By \Cref{prop:basic-polynomial-symmetrized} there exists a degree-$2T$ polynomial $q(U,U^*)$ belonging to the invariant ring $\C[X,X^*]^{\LU(d,d)}$ that decides $\cal{P}$ with the acceptance probabilities at least $a$ and at most $b$ for \emph{yes} and \emph{no} instances respectively. However since $U$ is self-adjoint this means that $q$ in fact belongs to the invariant ring $\C[X]^{\LU(d,d)}$. Since $U = I - 2\ketbra{\psi}{\psi}$, the polynomial $q$ can be equivalently expressed as a degree-$2T$ function of the projector $\ketbra{\psi}{\psi}$. 
By \Cref{thm:local_invariant_polys}, the polynomial $q$ can be expressed as a linear combination of $\Tr((R_\sigma \otimes R_\tau) \ketbra{\psi}{\psi}^{\otimes k})$ over all permutations $\sigma,\tau$ of $k$ elements with  $1 \leq k \leq 2T$. By \Cref{thm:one_dim_local_unitaries}, these traces are degree-$k$ symmetric polynomials in the entanglement spectrum of $\ket{\psi}$. Put together, this yields the desired polynomial $p$. %
\end{proof}

\subsection{The Entanglement Entropy Problem}

We use \Cref{lem:lu-one-dim} to prove lower bounds on an entanglement testing problem. Recall the Entanglement Entropy problem as defined in the introduction:

\renyientropy*

\entanglemententropyproblem*

Recall that the entanglement entropy is invariant under local unitaries, and that the entanglement entropy can be computed by the formula $H_2(\ket{\psi}) =  - \log (\sum_{i=1}^d \lambda_i^2)$ where $\lambda_i$ are the eigenvalues of the reduced density $\rho$. In particular, if the reduced density $\rho$ of $\ket{\psi}$ was maximally mixed on a subspace of dimension $r$, then $H_2(\ket{\psi}) = \log r.$ With a fairly straightforward application of the polynomial, we can prove the following query lower bound for the entanglement entropy problem.

\entanglemententropybound*

\begin{proof}
By properties of the entanglement entropy, we observe that if oracle $O = I - 2 \ket{\psi} \bra{\psi}$ satisfies the low or high entropy condition, so does the oracle $(U \otimes V) O (U^* \otimes V^*)$ for any unitaries $U$ and $V$. Hence applying \Cref{lem:lu-one-dim}
implies that if there was a $T$ query algorithm to solve the problem, there exists a degree $\leq 2T$ symmetric polynomial $p(\lambda_1, \dots, \lambda_n)$ which represents the success probability of the algorithm where $(\lambda_1, \dots, \lambda_n)$ are the eigenvalues of the reduced density $\rho$ of $\ket{\psi}.$

Let $r \geq 1$ be an integer and consider the success probability of the algorithm on instances where $\rho$ has $r$ eigenvalues equal to $\frac{1}{r}$ (i.e. $\rho$ is maximally mixed on a subspace of dimension $r$). For this set of eigenvalues $\Lambda_{r}$ we have $\lambda_1^i + \dots + \lambda_n^i = \frac{1}{r^{i-1}}.$ Hence, the substitution $q_1(r) = p(\Lambda_r)$ produces a Laurent polynomial with non-positive exponents only. This means that $q_2(r) = q_1(\frac{1}{r})$ is a polynomial satisfying the properties that:

\begin{itemize}
    \item $q_2(\exp(-a))) \leq \frac{1}{3}$ to satisfy the low entropy case.
    \item $q_2(\exp(-b)) \geq \frac{2}{3}$ to satisfy the high entropy case
    \item $0 \leq q_2(i) \leq 1$ at all points $i = \frac{1}{n}$.
\end{itemize}

Now we bound the degree of $q_2$. First assume that that $q_2(x)$ is bounded by 2 for all $x$ in the range $x \in [\frac{1}{d}, \exp(-a/2)]$. Then there is a point $y$ where the derivative of $q_2$ satisfies

$$|q_2'(y)| \geq \frac{\frac{2}{3} - \frac{1}{3}}{\exp(-a) - \exp(-b)} = \frac{\exp(a)}{3 (1 - \exp(a-b))}.$$

Since $exp(-a/2) - \frac{1}{d} \geq \exp(-a/2) - \exp(-a) \geq \frac{1}{2} \exp(-a/2)$ by assumption that $a \geq 5,$ then Markov's inequality implies that

$$\frac{\exp(a)}{3 (1 - \exp(a-b))} \leq \frac{2 (\deg q_2)^2}{\exp(-a/2) - \frac{1}{d}} \leq 4 \exp(a/2) (\deg q_2)^2.$$

Hence, in this case, $\deg q_2 \geq \Omega(\frac{\exp(a/4)}{\sqrt{1 - \exp(a-b)}}) \geq \Omega(\exp(a/4)).$

Otherwise, there exists a point $x \in [\frac{1}{d}, \exp(-a/2)]$ where $q_2(x) = k \geq 2$. In this case, there is a point $\frac{1}{r_1} < y < \frac{1}{r_1 + 1}$ where the derivative satisfies:

$$|q_2'(y)| \geq \frac{k-1}{\frac{1}{r_1} - \frac{1}{r_1 + 1}} = \frac{k-1}{\frac{1}{r_1 (r_1 + 1)}} \geq \frac{k}{2} r_1^2 \geq \frac{k}{2} \exp(a).$$

Hence in this case, Markov's inequality implies that 

$$\frac{k}{2} \exp(a) \leq \frac{2 k (\deg q_2)^2}{\exp(-a/2) - \frac{1}{d}} \leq 4 k \exp(a/2) (\deg q_2)^2,$$

which implies that, $\deg q_2 \geq \Omega(\exp(a/4)).$ Combining the two cases yields the claimed lower bound. 
\end{proof}

Observe the previous bound applies to any local unitarily invariant measure of entanglement entropy with the property that if the reduced density $\rho$ of $\ket{\psi}$ is maximally mixed on a $r$-dimensional subspace, then $H_2(\ket{\psi}) = \log r.$ Hence, the query lower bound applies to the von Neumann entropy as well as the Renyi $\alpha$-entropy for any $\alpha \neq 1,$ as these properties are satisfied by these measures of entanglement entropy as well.  

Furthermore, the lower bound also extends to the $\QMA$ setting much like for the recurrence problem by using Aaronson's guessing lemma. 

\entanglemententropyboundqma*

\begin{proof}
By the proof of \Cref{thm:entanglement_entropy_bound} and the guessing lemma (\Cref{lem:guessing_lemma}), if there was a $T$-query algorithm using an $m$-qubit witness that solves the entanglement entropy problem, there exists a polynomial $q$ of degree $O(mT)$ with the property that:

\begin{itemize}
    \item $q(i) \leq 2^{-10m}$ for all $i = \frac{1}{n}$ and integers $n \leq \exp(a).$
    \item $q(i) \geq 2^{-m}$ for all $i = \frac{1}{n}$ and integers $\exp(b) \leq n \leq d.$
\end{itemize}

In the first case, assume that the maximum of $q$ in the range $[\exp(-a), \exp(-a/2)]$ satisfies $q(x) = k \geq 2^{-9m}$ so that $2^{-10m} \leq \frac{k}{2^{-m}} \leq \frac{k}{2}$. Then, since for that point satisfies $\frac{1}{r + 1} \leq k \leq \frac{1}{r}$ for some $r \geq \exp(a/2),$ the derivative of $q$ in that interval satisfies

$$|q'(y)| \geq \frac{k - 2^{-10m}}{\frac{1}{r+1} - \frac{1}{r}} \geq \frac{k}{2} r^2 \geq \frac{k}{2} \exp(a),$$

for some point $y$ in that interval. Hence, by Markov's inequality the degree of $q$ satisfies

$$\frac{k}{2} \exp(a) \leq \frac{2 k (\deg q)^2}{\exp(-a/2) - \exp(-a)}  \leq 4 k \exp(a/2) (\deg q)^2,$$

since $\exp(-a/2) - \exp(-a) \geq \frac{1}{2} \exp(-a/2)$ by assumption. Therefore, $\deg q \geq \Omega(\exp(a/4))$ in this case. 

Otherwise, we have that $q$ is bounded by $2^{-9m}$ in the range $[\exp(-a), \exp(-a/2)].$ Let $y_0 = \exp(-a/2)$ and $y_1 = \exp(-a)$, by rescaling $q$, using

$$r(x) = 2^{9m} q \left( \frac{y_1 - y_0}{2} (x - 1) + y_1 \right),$$

$r$ satisfies $|r(x)| \leq 1$ for all $|x| \leq 1$ from the low entropy case. If $y_2 = \exp(-b),$ when $x =
1 + \frac{2 (y_2 - y_1)}{y_1 - y_0} = 1 + \frac{2 (\exp(-a) - \exp(-b))}{\exp(-a/2) - \exp(-a)}$ we have reached the high entropy case, we have $r(x) \geq 2^{8m}.$ Therefore, by Paturi's inequality with $\mu = \frac{2 (\exp(-a) - \exp(-b))}{\exp(-a/2) - \exp(-a)} \leq 4 \exp(-a/2),$ we obtain

$$2^{8m} \leq \exp(4 (\deg r) \sqrt{4 \exp(-a/2)}),$$

and hence $\deg q$ satisfies $\deg q \geq \deg r \geq \Omega(m \exp(a/4)).$

Hence recalling by the guessing lemma that $\deg q = O(mT)$ where $m$ is the witness size and $T$ is the query complexity, we have $mT \geq \Omega(\exp(a/4))$ as claimed.
\end{proof}

We now briefly sketch an upper bound in the setting where our property tester has access to proof states and our entanglement entropy measure is the Renyi 2-entropy and in the regime where $a \geq 5$, and $b \geq 2a$. Given two copies of the state $\ket{\psi}$ with reduced density matrix $\rho$, a swap test can be used to produce a Bernoulli random variable $X$ with mean $\mu$ equal to $\frac{1}{2} - \frac{1}{2} \Tr(\rho^2).$ Furthermore, by using the quantum amplitude estimation algorithm \cite{brassard2002quantum, hamoudi2021quantum, kothari2022mean}, one can produce an estimation of the mean $\mu$ of $X$ to additive error $\epsilon$ by an efficient quantum algorithm given $O(\frac{1}{\epsilon})$ samples from $X$. In particular, $O(\frac{1}{\beta-\alpha})$ samples can be used to distinguish between a Bernoulli random variable with mean at least $\beta$ or at most $\alpha$. 

Applying these results to our setting, we can solve the entanglement entropy problem if we can distinguish between states $\ket{\psi}$  whose purity satisfies $\Tr(\rho^2) \geq \exp(-b)$ or $\Tr(\rho^2) \leq \exp(-a).$ Hence with $O(\exp(a))$ copies of the state $\ket{\psi},$ a series of swap tests produces $O(\exp(a))$ samples from a Bernoulli random variable with mean equal to $\frac{1}{2} - \frac{1}{2} \Tr(\rho^2).$ As the gap between the means in the yes and no cases satisfies $\beta - \alpha = \exp(-a) - \exp(-b) \geq \frac{1}{2} \exp(-a)$ by assumption, this sample complexity is sufficient to distinguish between the yes and no cases. Overall, this yields an algorithm using $O(\exp(a))$ queries and a proof state with $O(\exp(a))$ qubits.

\section{The Entangled Subspace Problem and $\QMA$ versus $\QMAtwo$}
\label{sec:qma2}

We now turn towards studying an $\LU$-invariant unitary property testing problem that corresponds to a candidate oracle separation between $\QMA$ and $\QMAtwo$. Recall the definition of completely entangled subspaces and the Entangled Subspace problem from the introduction: an \emph{$\epsilon$-completely entangled subspace} $S \subseteq \C^d \otimes \C^d$ is such that all states $\ket{\theta} \in S$ are $\epsilon$-far in trace distance from any product state $\ket{\psi} \otimes \ket{\phi}$.

\entangledsubspace*
As mentioned earlier, the Entangled Subspace property is $\LU$-invariant: applying local unitaries $g \otimes h$ to a subspace $S$ preserves whether it is a \emph{yes} instance or a \emph{no} instance of the problem.

\subsection{$\QMAtwo$ Upper Bound}
\label{subsection:entangled_subspace}

We first give a $\QMAtwo$ upper bound to the Entangled Subspace problem:

\entangledsubspaceqmatwo*

\begin{figure}
    \centering
    \mbox{
    \Qcircuit @C=1.4em @R=1.4em{  
       \lstick{\ket{\psi}} & \qw & \multigate{1}{U} & \qw & \qw & \qw & \qw \\ 
       \lstick{\ket{\phi}}  & \qw & \ghost{U} & \qw & \qw & \qw & \qw \\ 
       \lstick{\ket{+}} & \qw & \ctrl{-1} & \qw & \gate{H} & \qw & \meter \\ 
     }
     }
    \caption{The verifier $V'$ in the proof of \Cref{thm:one_dimensional_property_qma}.}
    \label{fig:qma2_verifier}
\end{figure}
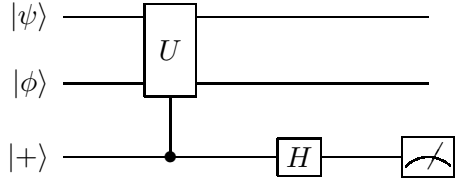

\begin{proof}
Consider the verifier illustrated in \Cref{fig:qma2_verifier} where the provided proof state is the product state $\ket{\psi} \otimes \ket{\phi}.$
The controlled-$U$ operation is essentially performing a subspace membership test. The state after the controlled-$U$ operation and the Hadamard on the ancilla qubit can be written as
\[
    \frac{I + U}{2} \Big( \ket{\psi} \otimes \ket{\phi}  \Big) \otimes \ket{0} +\frac{I - U}{2} \Big ( \ket{\psi} \otimes \ket{\phi} \Big ) \otimes \ket{1}.
\]
Since $(I - U)/2 = \Pi$, the acceptance probability of the subspace membership test is
\[
    \Big \| \Pi \Big ( \ket{\psi} \otimes \ket{\phi} \Big ) \Big \|^2 = |\braket{\xi \vert \psi \otimes \phi}|^2 = 1 - d(\ket{\xi}, \ket{\psi} \otimes \ket{\phi})^2 
\]
where $\ket{\xi}$ is the projection of $\ket{\psi} \otimes \ket{\phi}$ on $S$ and $d$ is the trace distance. 

In the \text{yes} case, there exists a product state that is $a$-close to a state in $S$, and hence providing that state as a certificate makes the verifier accept with probability at least $1 - a^2$. Otherwise in the \text{no} case, all states $\ket{\xi} \in S$ are $b$-far from product, and hence the verifier accepts with probability no more than $1 - b^2.$ %
\end{proof}

We note that the verifier analyzed in the proof of \Cref{prop:entangledsubspaceqmatwo} has the property that in the \text{yes} case, proof state may not be a symmetric product state (i.e. a state of the form $\ket{\psi}^{\otimes 2}$). We now present a $\QMAtwo$ verifier, which we call the \textit{product test verifier}, for the Entangled Subspace Problem, with the additional property that in the \text{yes} case there exists a valid proof state that is symmetric. The verifier relies on a procedure known as the \emph{product test}, which was analyzed by Harrow and Montanaro \cite{harrow2013testing} and also later in \cite{soleimanifar2022testing}. We state the main results about the product test here, specialized to the case of bipartite states.

\begin{definition}[Product test]
Let $\ket{\psi}$ be a state in $\C^d \otimes \C^d$. Consider two copies of the $\ket{\psi}^{\otimes 2}$, where the first copy is on registers $A_1 B_1$ and the second copy is on registers $A_2 B_2$. The product test applies the swap test on registers $A_1 A_2$, and another swap test on $B_1 B_2$. The product test accepts iff both swap tests accept.
\end{definition}
Observe that if $\ket{\psi} = \ket{\varphi} \otimes \ket{\xi}$, then the product test accepts with probability $1$. On the other hand, we have the following bound for the probability an entangled $\ket{\psi}$ will pass the product test. %

\begin{theorem}[{{\cite[Theorem 8]{soleimanifar2022testing}}}]
Given a state $\ket{\psi} \in \C^d \otimes \C^d$, let $$\omega_{\ket{\psi}} = \max \{ |\braket{\psi \vert \phi_1 \otimes \phi_2}|^2, \ket{\phi_1} \in \mathbb{C}^d, \ket{\phi_2} \in \mathbb{C}^d\}$$
denote the overlap of $\ket{\psi}$ with the closest product state. Then the probability $\alpha$ that the product test  passes satisfies 
\[
\frac{1}{2} (1 + \omega_{\ket{\psi}}^2) \leq \alpha \leq \frac{1}{3} \omega_{\ket{\psi}}^2 + \frac{2}{3}~.
\]
\label{thm:product_test}
\end{theorem}

While \Cref{thm:product_test} assumes that the input to the product test is symmetric, we note that the product test is also sound against non-symmetric witnesses.

\begin{proposition}[{{\cite[Appendix E]{harrow2013testing}}}]
    Let $P(\ket{\Phi})$ be the probability that the product test passes when given a state $\ket{\Phi}$ as input. Then, for any $\ket{\psi_1} \in \C^d \otimes \C^d$ and $\ket{\psi_2} \in \C^d \otimes \C^d,$ we have

    $$P(\ket{\psi_1} \otimes \ket{\psi_2}) \leq \frac{P(\ket{\psi_1}^{\otimes 2}) + P(\ket{\psi_2}^{\otimes 2})}{2}.$$
\label{prop:product_test_non_symmetric}
\end{proposition}

By combining \Cref{thm:product_test} and \Cref{prop:product_test_non_symmetric} we obtain the following result.

\begin{proposition}
Suppose $0 \leq a < b \leq 1$ are constants satisfying satisfy $$\frac{1}{2} (1 + (1-a^2)^2) > \frac{1}{3} (1 - b^2)^2 + \frac{2}{3}.$$ Then there exists a $\QMAtwo$ verifier for the $(a,b)$-Entangled Subspace problem with the property that on \text{yes} instances with oracle $O = I - 2P_S$ where $P_S$ is a projector onto subspace $S$, a valid proof state is $\ket{\psi}^{\otimes 2}$ where $\ket{\psi}$ is any state in $S$ that is $a$-close to product. \label{prop:qma2_product_test_verifier}
\end{proposition}

\begin{figure}[ht]
         \centering
         \mbox{
         \Qcircuit @C=1.4em @R=1.4em {
                 & \lstick{\ket{\psi_1}_{A_1 B_1}}  & \gate{U} & \qw & \multigate{1}{\mathrm{Swap}_{A_1 A_2}} & \multigate{1}{\mathrm{Swap}_{B_1 B_2}} & \qw &  \qw  \\
                 & \lstick{\ket{\psi_2}_{A_2 B_2}}  & \qw & \gate{U} & \ghost{\mathrm{Swap}_{A_1 A_2}} & \ghost{\mathrm{Swap}_{B_1 B_2}} & \qw & \qw  \\
                 & \lstick{\ket{+}_1} & \ctrl{-2} &  \qw & \qw   & \qw & \gate{H} & \meter  \\
                 & \lstick{\ket{+}_2} & \qw & \ctrl{-2} & \qw   & \qw & \gate{H} & \meter \\
                 & \lstick{\ket{+}_3} & \qw & \qw & \ctrl{-3}  & \qw & \gate{H} & \meter \\
                 & \lstick{\ket{+}_4} & \qw & \qw & \qw  & \ctrl{-4} & \gate{H} &  \meter 
         }    }
         \caption{$\QMA(2)$ product test verifier}
         \label{figure:product_test_verifier}
\end{figure}
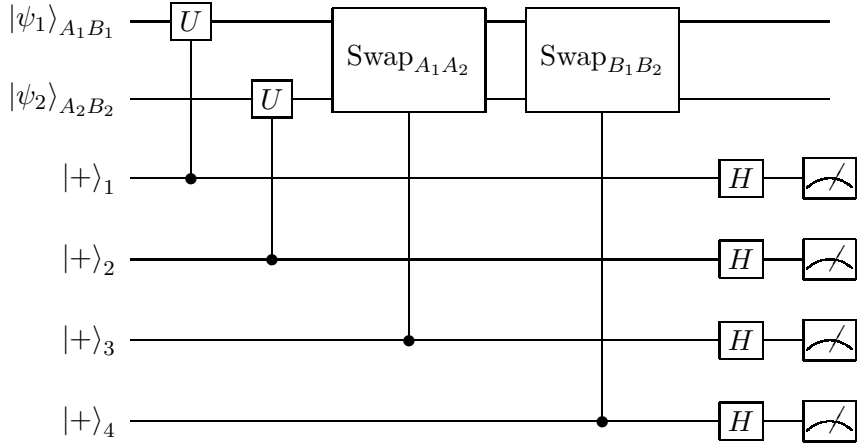

\begin{proof}

Suppose $\ket{\psi_1} \otimes \ket{\psi_2}$ is given as input to the verifier. From the proof of \Cref{prop:entangledsubspaceqmatwo}, the probability that the first two ancillas accept is $\|P_S \ket{\psi_1}\| \|P_S \ket{\psi_2}\|.$ Conditioned on the first two ancillas accepting, the probability the third and fourth ancillas accept is the probability that the product test passes when provided the state $\frac{P_S \ket{\psi_1}}{\| P_S \ket{\psi_1} \|} \otimes \frac{P_S \ket{\psi_2}}{\| P_S \ket{\psi_2} \|}$ as input. Hence, the verifier maximizes its success probability when $\ket{\psi_1} \in S$ and $\ket{\psi_2} \in S$. Furthermore, by \Cref{prop:product_test_non_symmetric}, we can assume that $\ket{\psi_1} = \ket{\psi_2}$ to maximize the verifier's success probability.

By \Cref{thm:product_test}, in the \text{yes} case, there exists a state in $\ket{\psi} \in S$ that is $a$-close to product, and hence the product test passes with probability at least $\frac{1}{2} (1 + (1-a^2)^2)$ when given $\ket{\psi}^{\otimes 2}$ as input. Otherwise, in the \text{no} case, all states in $S$ are $b$-far from product, and hence the product test passes with probability at most $\frac{1}{3} (1-b^2)^2 + \frac{2}{3}$ in this case on any input $\ket{\psi}^{\otimes 2}$ for $\ket{\psi} \in S$. Therefore, as long as $\frac{1}{2} (1 + (1-a^2)^2) > \frac{1}{3} (1-b^2)^2 + \frac{2}{3},$ there is a bounded gap in the success probability between the two cases.
\end{proof}

To extend the result to an arbitrary gap between $a$ and $b$, we note that the product test can be further generalized to the situation when input consists of $k \geq 2$ copies of a given state $\ket{\psi}.$

\begin{definition}[$k$-copy product test]
Let $\ket{\psi}$ be a state in $\C^d \otimes \C^d$. Consider $k \geq 2$ copies of the $\ket{\psi}^{\otimes 2}$ where copy $i$ is on registers $A_i B_i$. The product test is a circuit that performs a projective measurement $\{P = \Pi_A \otimes \Pi_B, I - P\}$ where $\Pi_A$ is the projector on the symmetric subspace on registers $A_1, \dots, A_k$ and $\Pi_B$ is the projector on the symmetric subspace on registers $B_1, \dots, B_k.$ The success probability of the product test is $\|P \ket{\psi}^{\otimes k}\|^2.$
\label{defn:k-copy-product-test}
\end{definition}

By the results of \cite{harrow2013testing} and \cite{barenco1997stabilization}, there is an efficient quantum circuit which implements the product test for all constant $k$. Similarly, we can bound the success probability of the $k$-copy product test in terms of the overlap with the closest product state, using the proof techniques presented in \cite{soleimanifar2022testing}.

\begin{restatable}{theorem}{generalqmakptbound}
Let $\omega_{\ket{\psi}}$ be the overlap of $\ket{\psi}$ with the closest product state, as defined in \Cref{thm:product_test}. For all constant $k \geq 2$, the probability $\alpha$ that the product test passes when given $\ket{\psi}^{\otimes k}$ as input satisfies 

$$\alpha \leq \frac{k-1}{k+1} \omega_{\ket{\psi}}^k + \frac{2}{k+1}.$$
\label{thm:general_qma_k_bound}
\end{restatable}

We defer the proof of \Cref{thm:general_qma_k_bound} to  \Cref{sec;product_test_analysis}. We apply \Cref{thm:general_qma_k_bound} to obtain the following result, whose proof is deferred to \Cref{sec:symqma_entangledsubspace}.

\begin{restatable}{theorem}{generalentangledverifier}
Let $0 \leq a < b < 1$ be constants. Then there exists a constant $k \geq 2$ sufficiently large such that there is a $\mathsf{SymQMA(k+1)}$ verifier for the $(a,b)$-Entangled Subspace problem.
\label{thm:general_qma_k_verifier}
\end{restatable}

We recall that $\mathsf{SymQMA(k)}$ is the variant of $\mathsf{QMA(k)}$ where the witness is promised to be a symmetric product state $\ket{\psi}^{\otimes k}$. Since for any constant $k \geq 2,$ $\textsf{SymQMA(k)} = \textsf{QMA(k)} = \QMAtwo$ result of \cite[Lemma 38]{aaronson2008power}, this construction provides another proof that the Entangled Subspace Problem is in $\QMAtwo.$

\subsection{$\QMA$ versus $\QMAtwo$ for State Property Testing}

Next, we turn to constraints on using the Entangled Subspace problem to obtain an oracle separation between $\QMA$ and $\QMAtwo$. One might hope that, given the characterization of $\LU$-invariant polynomials for unitaries that encode a one-dimensional subspace (\Cref{lem:lu-one-dim}), we may be able to obtain a $\QMA$ versus $\QMAtwo$ separation by proving a lower bound on the Entangled Subspace problem by focusing on one-dimensional subspaces only. 

However we show that generally property testing questions concerning states (equivalently, one-dimensional subspaces) are not sufficient to resolve the $\QMA$ versus $\QMAtwo$ problem. Therefore proving \Cref{conj:qmatwo} necessarily requires studying problems about the entanglement of higher dimensional subspaces.

\onedimensionalqma*

\begin{proof}
Let $V$ denote the $\QMAtwo$ verifier that decides $\cal{P}$. We construct a $\QMA$ verifier $V'$ (depicted in \Cref{fig:new_verifier}) that can receive an entangled proof state $\ket{\Phi}$. Label the registers of $\ket{\Phi}$ by $A_1 B_1 A_2 B_2$. The verifier $V'$ performs the subspace membership test on registers $A_1 B_1$ and $A_2 B_2$ separately by calling $U$ controlled on two ancilla qubits initialized in the $\ket{+}$ state. 

The verifier $V'$ then applies Hadamard gates to the ancilla and measures. It takes the post-measurement state of registers $A_1 B_1 A_2 B_2$ and runs the original verifier $V$ on them. The new verifier $V'$ accepts if and only if the two ancilla bits accepted and the original verifier $V$ accepted.

If $U$ encodes a pure state $\ket{\psi}$ and is a \emph{yes} instance, then by assumption on the original verifier $V$, we can run $V'$ with proof state $\ket{\Phi} = \ket{\psi}^{\otimes 2}$ and it will accept with the same probability as $V$.

On the other hand, assume that $U$ is a \emph{no} instance, and let $\ket{\Phi}$ be the (possibly entangled) proof state provided to verifier $V'$.
Decompose $\ket{\Phi} = c_0 \ket{\psi}^{\otimes 2} + c_1 \ket{\xi}$ for some state $\ket{\xi}$ orthogonal to $\ket{\psi}^{\otimes 2}$. Since $V'$ accepts $\ket{\Phi}$ only when the subspace membership tests accept, then $V'$ accepts $\ket{\Phi}$ with probability $|c_0|^2 s$, where $s$ is the probability that $V$ accepts $\ket{\psi}^{\otimes 2}.$ Since $|c_0|^2 \leq 1,$ then the acceptance probability of $V$ at most $s$.

Therefore, $V'$ has the same soundness and completeness as $V$, even allowing for entangled states as input.

\begin{figure}
    \centering
    \mbox{
    \Qcircuit @C=1.4em @R=1.4em{  
       \lstick{\ket{\Phi}} & \qw & \gate{U} & \qw & \qw & \multigate{2}{V} & \qw \\ 
       & \qw & \qw & \gate{U} & \qw & \ghost{V} & \qw \\ 
       \lstick{\ket{0}} & \qw & \qw & \qw & \qw &  \ghost{V} & \qw \\
       \lstick{\ket{+}} & \qw & \ctrl{-3} & \qw & \gate{H} & \qw & \meter \\ 
       \lstick{\ket{+}} & \qw & \qw & \ctrl{-3} & \gate{H} & \qw & \meter \\ 
     }
     }
    \caption{The verifier $V'$ in the proof of \Cref{thm:one_dimensional_property_qma}.}
    \label{fig:new_verifier}
\end{figure}
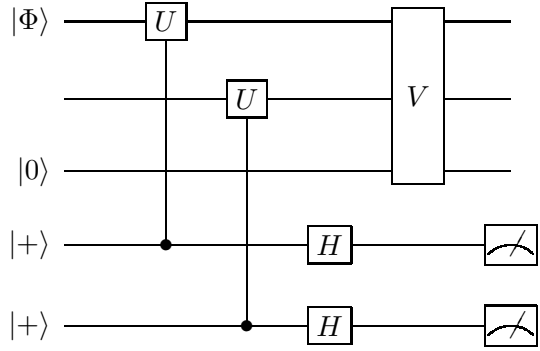
\end{proof}

Hence, combining \Cref{thm:one_dimensional_property_qma} and \Cref{prop:qma2_product_test_verifier}, we obtain that there exists a parameter range $(a, b)$ for which there exists a $\QMA$ verifier for solving the one-dimensional $(a,b)$-Entangled Subspace problem.

\subsection{The $\QMA$ (Un)soundness of the Product Test Verifier}
\label{subsec:qma_example}

We now show that when there is no unentanglement guarantee for the proof state, the product test verifier fails to be sound. What this means is that there is a \emph{no} instance $U$ of the Entangled Subspace problem (with different parameters) but also a proof state $\ket{\theta}$ that may be completely entangled across the four registers $A_1 B_1 A_2 B_2$, such that the product test verifier will accept with probability $1$ when making queries to $U$. In other words, the product test verifier can be \emph{fooled} by an entangled proof in the $\QMA$ setting.

\begin{proposition}
Let $d \geq 4$. Let $u,v,w,x \in [d]$ be distinct. Let $S \subseteq \C^d \otimes \C^d$ denote the six-dimensional subspace spanned by $\frac{\ket{uv} + \ket{vu}}{\sqrt{2}}, \frac{\ket{uw} + \ket{wu}}{\sqrt{2}}, \frac{\ket{ux} + \ket{xu}}{\sqrt{2}}, \frac{\ket{vw} + \ket{wv}}{\sqrt{2}}, \frac{\ket{vx} + \ket{xv}}{\sqrt{2}}, \frac{\ket{wx} + \ket{xw}}{\sqrt{2}}.$ Let $\Pi$ be the projector onto $S$ and let $\ket{\psi_{uvwx}}$ be the state

\begin{equation*} 
\begin{split} 
\ket{\psi_{uvwx}} &= \frac{1}{\sqrt{24}} \Big [(\ket{uv} + \ket{vu}) \otimes (\ket{wx} + \ket{xw}) +  (\ket{wx} + \ket{xw}) \otimes (\ket{uv} + \ket{vu})  \\ 
&+ (\ket{uw} + \ket{wu}) \otimes (\ket{vx} + \ket{xv}) + (\ket{vx} + \ket{xv}) \otimes (\ket{uw} + \ket{wu}) \\
&+  (\ket{ux} + \ket{xu}) \otimes (\ket{vw} + \ket{wv}) + (\ket{vw} + \ket{wv}) \otimes (\ket{ux} + \ket{xu}) \Big ]~.
\end{split}
\end{equation*}

Then:

\begin{enumerate}
    \item $S$ is a $1/4$-completely entangled subspace, with every state $\ket{\varphi} \in S$ having overlap at most $\frac{3}{4}$ with a product state.
    \item The product test verifier making queries to $U = I - 2\Pi$, accepts the entangled proof state $\ket{\psi_{uvwx}}$ with probability $1$ for any $\ket{\phi} \in S$. %
\end{enumerate}
\label{thm:counterexample}
\end{proposition}

\begin{proof}
Assume without loss of generality that $d = 4$. Otherwise apply an isometry $W \otimes W$ to $S$ where $W: \text{Span}(\ket{u}, \ket{v}, \ket{w}, \ket{x}) \rightarrow \mathbb{C}^d$ is an isometry, which does not change the magnitude of the closest product state by construction.  

Let $\ket{\varphi} \in S$. Observe that $S$ is contained in the symmetric subspace of $\C^d \otimes \C^d$, so by \cite[Lemma 1]{hayashi2008entanglement}, the closest product state to $\ket{\varphi}$ can be chosen to be a symmetric state $\ket{\phi}^{\otimes 2}.$ Write $\ket{\phi} = \sum_{i=1}^4 \beta_i \ket{i}.$ Write $\ket{\psi_{ij}} = \frac{\ket{ij} + \ket{ji}}{\sqrt{2}}$ and $\ket{\psi} = \sum_{1 \leq i < j \leq 4} \alpha_{ij} \ket{\psi_{ij}}.$ Then by the Cauchy-Schwartz inequality and normalization: 

$$| \braket{\phi^{\otimes 2} \vert \psi} |^2 = \left| \sum_{1 \leq i < j \leq 4} \alpha_{ij} \frac{\beta_i \beta_j + \beta_j \beta_i}{\sqrt{2}} \right |^2 \leq 2 \sum_{1 \leq i < j \leq 4} |\alpha_{ij}|^2  \sum_{1 \leq i < j \leq 4} |\beta_i|^2 |\beta_j|^2   = 2 \sum_{1 \leq i < j \leq 4} |\beta_i|^2 |\beta_j|^2 $$

Since $2 \sum_{i < j} |\beta_i|^2 |\beta_j|^2 + \sum_{i=1}^4 |\beta_i|^4 = 1$ and $\sum_{i=1}^4 |\beta_i|^2 = 1$ since $\ket{\phi}$ and $\ket{\phi}^{\otimes 2}$ are normalized, we can conclude that $$\sum_{i=1}^4 |\beta_i|^4  = \sum_{i=1}^4 |\beta_i|^2 |\beta_i|^2 \geq \frac{1}{4} \sum_{i=1}^4 |\beta_i|^2 = \frac{1}{4}.$$

Therefore, $2 \sum_{i < j} |\beta_i|^2 |\beta_j|^2 \leq \frac{3}{4}$ is an upper bound for the overlap with a product state for any state $\ket{\varphi} \in S$. This establishes the first item of the proposition statement.

For the second item, let $\ket{\psi_{uvwx}}$ be the proof state given to the product test verifier in Figure \ref{figure:product_test}. By construction, the membership queries pass with probability one. Furthermore, observe that $\ket{\psi_{uvwx}}$ is symmetric under all permutations of the registers, so $\ket{\psi_{uvwx}}$ passes the 2-copy product test with probability $1$. Hence the verifier accepts the state $\ket{\psi_{uvwx}}$ and oracle $U$ with probability one. Hence the verifier is not sound against entangled proofs since $\ket{\psi_{uvwx}}$ was entangled, and $S$ is a $\frac{1}{4}$-completely entangled subspace. 
\end{proof}

In \Cref{sec:fool} we give a complete characterization of the proof states that fool the product test verifier; this relies on the representation theory of the symmetric group. We believe that these examples of subspaces and proof states that fool purported $\QMA$ verifiers for the Entangled Subspace problem should yield insight into proving general $\QMA$ lower bounds for the Entangled Subspace problem. 

\subsection{Average Case Versions of the Entangled Subspace Problem}
\label{subsec:average_case}

In this section, we discuss the average case variants of the Entangled Subspace problem, restated from the introduction.

\plantedproductstate*

\restricteddimcount*

Combined with the following result about Haar-random subspaces, we obtain our two property testing problems are in fact average case versions of the Entangled Subspace problem as claimed.  Informally, the lemma asserts that if $S \subseteq \mathbb{C}^d \otimes \mathbb{C}^d$ was a Haar-random subspace of small dimension $s$ compared to $d^2,$ then every state in $S$ is entangled with high probability. Otherwise, for sufficiently large $s$, there exists a state in $S$ that is close to a product state. The proof is based on the techniques of \cite{Hayden_2006} that use L\'evy lemma for the Haar measure and some additional observations about the closest product state.

\begin{lemma}[Levy's Lemma {{\cite{milman1986asymptotic, ledoux2001concentration}}}]
\label{lem:levy}
Let $f: \mathbb{S}^k \to \mathbb{R}$ be a function with Lipschitz constant $\eta$ (with respect to the Euclidean norm) and let $\ket{\psi} \in \mathrm{S}^k$ be chosen uniformly at random from the Haar measure. Then
\[
\Pr \Big( \Big| f(\ket{\psi}) - \E f \Big| > \alpha \Big) \leq 2 \exp \Big( - C (k+1)\alpha^2 /\eta^2 \Big)
\]
where $C = (9\pi^3 \ln 2)^{-1}$ and $\E f$ denotes the average of $f$ over $\mathbb{S}^k$.
\end{lemma}

\begin{lemma}
Let $\omega$ be the overlap of $\ket{\psi}$ with the closest product state as defined in \Cref{thm:product_test}. If $\rho = \Tr_1(\ket{\psi} \bra{\psi})$ was the reduced density matrix of $\ket{\psi}$, then 

$$\omega^2 \leq \Tr(\rho^2) \leq \omega.$$ 
\label{lem:entanglement_schmidt}
\end{lemma}

\begin{proof}
Let $\lambda_1 \geq \lambda_2 \geq \dots \geq \lambda_d$ be the eigenvalues of $\rho$. Then \cite[Lemma 2]{harrow2013testing} shows that $\omega = \lambda_1.$ Therefore, since $\sum_{i=1}^d \lambda_i = 1$ and each $\lambda_i$ is non-negative, we have $$\lambda_1^2 \leq \Tr(\rho^2) = \sum_{i=1}^d \lambda_i^2 \leq \sum_{i=1}^d \lambda_1 \lambda_i = \lambda_1.$$
\end{proof}

\begin{theorem}
Let $S \subseteq \mathbb{C}^d \otimes \mathbb{C}^d$ be a Haar-random subspace of dimension $s$, and let $\ket{\psi} \in S$ and let $C$ be the constant from Levy's Lemma. Then 

\begin{enumerate}
    \item For all constant $\delta > 0$, if $d \geq \frac{4}{\delta}$ and $s \leq \frac{C \delta^2}{1024 \log(\frac{40}{\delta})} d^2,$ then all states in $S$ have purity at most $\delta$ with high probability. In particular,
    
    $$\Pr_{\ket{\psi} \in S}[\omega_{\ket{\psi}} \geq \sqrt{\delta}] \leq \exp(-O(d^2)).$$
    \item For all constant $\delta > 0,$ if $s \geq 2 \sqrt{\delta}  d^2$, then $S$ contains a state with overlap with the closest product state $\delta$ with high probability. In particular,  $$\Pr_{\ket{\psi} \in S}[\sup_{\ket{\psi} \in S} \omega_{\ket{\psi}} \leq \delta] \leq  \exp(-O(d^2)).$$ 
\end{enumerate}
\label{lem:random_entanglement}
\end{theorem}

\begin{proof}
To show the first part, let $f(\ket{\psi}) = \Tr(\rho^2)$ denote the purity of the reduced density matrix of $\ket{\psi}$ on the first $\mathbb{C}^d$ factor. The Lipschitz constant of $f$ is $4$ by \cite[Lemma III.8]{Hayden_2006}. 

Let $P$ be a fixed projector onto the first $s$ basis states (according to some canonical ordering) of $\mathbb{C}^d \otimes \mathbb{C}^d$, and $U$ is a Haar-random unitary on $\mathbb{C}^d \otimes \mathbb{C}^d$. Let $\cal{N}$ denote an $\epsilon$-net for the image of the projector $P$ of size $(5/\epsilon)^{2s}$, which exists by \cite[Lemma III.6]{Hayden_2006}. Note that $U \cal{N}$ is an $\epsilon$-net for the image of the projector $P_S = U P U^*$, which is a uniformly random subspace of dimension $s$ when $U$ is choosen to be Haar-random.

We want to bound the probability that there exists a state $\ket{\psi} \in S$ whose purity is large. 
\begin{align*}
    \Pr \Big( \exists \, \ket{\psi} \in S \, \text{ such that } f(\ket{\psi}) \geq \delta \Big) &\leq \Pr \Big( \exists \, \ket{\varphi} \in U \cal{N} \, \text{ such that } f(\ket{\varphi}) \geq \delta - 4\epsilon\Big) \\
    &\leq \sum_{\ket{\phi} \in \cal{N}} \Pr \Big( f(U \ket{\phi}) \geq \delta - 4\epsilon \Big) \\
    &= \Big( \frac{5}{\epsilon} \Big)^{2s} \cdot \Pr_{\ket{\phi} \sim \text{Haar}(d^2)} \Big( f(\ket{\phi}) \geq \delta - 4\epsilon \Big).
\end{align*}

On average, the purity of a Haar-random state in $\mathbb{C}^d \otimes \mathbb{C}^d$ is $\beta := 2d/(d^2 + 1)$ by  \cite[Proposition 4.14]{collins2016random}.

Thus by Levy's Lemma we have that
\[
\Pr_{\ket{\phi} \sim \text{Haar}(d^2)} \Big( f(\ket{\phi}) \geq \delta - 4\epsilon\Big) \leq 2 \exp \Big( -\frac{C}{16} (d^2 + 1) (\delta - 4\epsilon - \beta)^2 \Big)
\]
where $C$ was the constant from  \Cref{lem:levy}.  Now choosing $\epsilon = \frac{\delta}{8}$ and $d$ sufficiently large so that $\frac{1}{d} \leq \frac{2d}{d^2 + 1} \leq \frac{\delta}{4},$ we have

\begin{align*} 
 \Pr \Big( \exists \, \ket{\psi} \in S \, \text{ such that } f(\ket{\psi}) \geq \delta \Big) &\leq 2 \Big(\frac{40}{\delta} \Big)^{2s} \exp \Big (-\frac{C}{16}(d^2 + 1) (\frac{\delta}{4})^2 \Big) \\ &= 2  \exp \Big (-\frac{C}{256} (d^2 + 1) \delta^2 + 2s \log(\frac{40}{\delta}) \Big)
\end{align*}

by combining the above bounds. Hence the claimed choice of $s$ makes this probability exponentially small in $d^2$. Furthermore, by  \Cref{lem:entanglement_schmidt}, we have $\Pr[\Tr(\rho^2) \geq \delta] \geq \Pr[\omega_{\ket{\psi}} \geq \sqrt{\delta}]$, and hence the probability that the overlap with a product state is at least $\sqrt{\delta}$ is exponentially small.

To show the second part, fix a product state $\ket{v} = \ket{a} \otimes \ket{b} \in \mathbb{C}^d \otimes \mathbb{C}^d$. The overlap between $\ket{v}$ and the subspace $S$ is captured by the quantity
\[
     f(U) = \bra{v} P_S \ket{v} = \bra{v} U P U^* \ket{v}~.
\]

Equivalently, the overlap is \[
    f(\ket{\psi}) = \bra{\psi} P \ket{\psi} 
\]
where $\ket{\psi}$ is a Haar-random state since $\ket{v}$ was a fixed vector. The average of $f(\ket{\psi})$ can be computed as
\[
    \int f(\ket{\psi} \, \mathrm{d}\ket{\psi} = \Tr \Big( P \, \int \ket{\psi}\bra{\psi} \, \mathrm{d}\ket{\psi} \Big) = \frac{1}{d^2} \Tr(P) = \frac{s}{d^2}~.
\]

We compute the Lipschitz constant of $f$ since $P$ is a projector:
\begin{align*}
    \sup_{\ket{\psi},\ket{\varphi}} \frac{ | f(\ket{\psi}) - f(\ket{\varphi})|}{ \| \ket{\psi} - \ket{\varphi} \|} &= \sup_{\ket{\psi},\ket{\varphi}} \frac{ \Big| \| P \ket{\psi} \|^2 - \| P \ket{\varphi} \|^2 \Big| }{\| \ket{\psi} - \ket{\varphi} \|} \\
    &= \sup_{\ket{\psi},\ket{\varphi}} \frac{ \Big| \| P \ket{\psi} \| + \| P \ket{\varphi} \| \Big | \cdot \Big| \| P \ket{\psi} \| - \| P \ket{\varphi} \| \Big | }{\| \ket{\psi} - \ket{\varphi} \|} \\
    &\leq \sup_{\ket{\psi},\ket{\varphi}} \frac{ 2 \Big| \| P \ket{\psi} \| - \| P \ket{\varphi} \| \Big | }{\| \ket{\psi} - \ket{\varphi} \|} \\
    &\leq \sup_{\ket{\psi},\ket{\varphi}} \frac{ 2  \| P (\ket{\psi} - \ket{\varphi}) \| }{\| \ket{\psi} - \ket{\varphi} \|} \\
    &\leq 2~.
\end{align*}

We now apply Levy's Lemma to conclude that
\[
    \Pr \Big( f(\ket{\psi}) < \frac{s}{d^2} - \delta \Big) \leq 2 \exp \Big ( -\frac{C}{4} (d^2 + 1) \delta^2 \Big)
\]

Hence as long as $s \geq 2\sqrt{\delta} d^2,$ we have for $\ket{w} = P_S \ket{v},$

\begin{align*}
     \Pr[|\langle v | w \rangle|^2 \geq \delta] &=  \Pr[|\langle v | w \rangle| \geq \sqrt{\delta}] = 1 -  \Pr[|\langle v | w \rangle| \leq  \sqrt{\delta}] \\ 
     &\geq  1 -  \Pr[|\langle v | w \rangle| \leq \frac{s}{d^2} -  \sqrt{\delta}] \\
     &\geq 1 - 2 \exp \Big( -\frac{C \delta}{4} (d^2 + 1) \Big)
\end{align*}

and therefore $S$ contains a state with overlap at least $\delta$ with probability 1 - $\exp(-O(d^2)).$ 
\end{proof}

We note that the above bounds in  \Cref{lem:random_entanglement} are likely not tight, and finding tight bounds would be an interesting open problem. However, we are also now able to show \Cref{prop:plantedyesno} and \Cref{prop:restrictedyesno} using this result.

\plantedyesno*

\begin{proof}
Clearly in the \text{yes} case, the subspace $S$ contains a product state. Otherwise, given $\epsilon > 0$, choosing $\delta = (1 - \epsilon^2)^2$ in \Cref{lem:random_entanglement} (1) implies that there is some constant $C$ such that a Haar-random subspace of dimension $s \leq Cd^2$ is $\epsilon$-completely entangled. Hence setting choosing any $s$ in this range gives a \textit{no} instance with probability at least $1 - \exp(-O(d^2)).$ 
\end{proof}

\restrictedyesno*

\begin{proof}
Let $a < b$ be the two parameters in the Entangled Subspace problem where \textit{yes} instances have a state that is $a$-close to product and otherwise \textit{no} instances are $b$-completely entangled. Choose $C_1$ using $\delta = (1 - b^2)^2$ from \Cref{lem:random_entanglement} (1), and $C_2$ using $\delta = 1 -a^2$ from \Cref{lem:random_entanglement} (2). Then with probability at least $1  - \exp(-O(d^2)),$ a Haar-random subspace of dimension $\leq C_1 t^2$ is $b$-completely entangled, and a subspace of dimension $\geq C_2 t^2$ contains a state that is $a$-close to a product state. 
\end{proof}

Hence, having observed that our average-case problems can be reduced to the Entangled Subspace problem with overwhelming probability, we use our results from \Cref{subsection:entangled_subspace} to show that they can be solved by a $\QMAtwo$ tester with high probability. Furthermore, we conjecture that a lower bound for a $\QMA$ tester for the Entangled Subspace problem extends to this average case setting.

\subsection{Connections to Invariant Theory}
\label{sec:invariant_theory_qma}

Observe that all of our candidate problems, being special cases of the Entangled Subspace problem, have local unitary symmetries. This follows from product states being preserved under local unitary transformations, and the unitary invariance of the trace distance. 

This opens up the possibly of using the generalized polynomial method to prove  a  $\QMA$ lower bound for our candidate problems. While these problems are similar in spirit to the entanglement entropy problem that also has a local unitary symmetry introduced in \Cref{sec:local_unitary}, the main barrier to applying the polynomial method in this case is that we do not appear to have a good characterization of the invariant polynomials in  \Cref{thm:local_invariant_polys} in the case where $P$ is a projector onto a high-dimensional subspace. While  \Cref{thm:one_dim_local_unitaries} characterizes these polynomials in the case where $P$ is a one-dimensional projector, we have seen in the previous section that one-dimensional properties cannot be used to separate $\QMA$ and $\QMAtwo.$ We are not aware of a good characterization of these invariants even in the case where $P$ is a projector onto a two-dimensional subspace. A deeper understanding of these invariants appears necessary to make further progress on these questions.

\subsection{$\QCMA$ Lower Bound for the Entangled Subspace Problem}

As described in the previous section, we are not currently able to prove a strong $\QMA$ lower bound on the query complexity of the entangled subspace problem. However, using a similar proof strategy as Aaronson and Kuperberg in \cite{aaronson2007quantum}, we show a lower bound against $\textsf{QCMA}$, which is the subclass of $\QMA$ of problems verifiable by a polynomial time quantum verifier with a classical proof string. 

To present this lower bound, we first recall the definition of a $p$-uniform measure over quantum states from \cite{aaronson2007quantum}.

\begin{definition}
Let $\mu$ be the Haar measure over $n$-dimensional sphere $\mathbb{S}^n$. A measure $\sigma$ is $p$-uniform if it can be obtained from $\mu$ by conditioning on an event $A$ with measure $\mu(A) \geq p.$
\end{definition}

Using L\'evy's lemma, we can observe the following property of $p$-uniform measures.

\begin{lemma}
Let $f(\ket{\psi}): \mathbb{S}^d\rightarrow \R$ be a non-negative, Lipschitz function on the sphere bounded by 1. Let $\E_\mu[f]$ be its expectation over the Haar measure. Then if $\sigma$ is a $p$-uniform measure, then 

$$\E_\sigma[f] \leq \E_\mu[f] + O\left(\sqrt{\frac{\log \frac{1}{p} + \log d}{d}}\right).$$ 
\label{lem:p-uniform-lemma}
\end{lemma}

\begin{proof}
Let $\bar{f} = \mathbb{E}_\mu[f]$, $X = f(\ket{\psi}),$ and $a > 0.$ By the definition of $p$-uniform measure, we have

$$\Pr_\sigma[|X - \bar{f}| \geq a] \leq \frac{1}{p} \Pr_\mu[|X - \bar{f}| \geq a],$$

and by L\'evy's lemma there exists a constant $C$ such that, 

$$\Pr_\mu[|X - \bar{f}| \geq a] \leq 2 \exp(- C d a^2).$$  

Hence, for every $a > 0$, we get

$$\E_\sigma[f] \leq (\bar{f} + a) \Pr_\sigma[|X  - \bar{f}| \leq a] + \Pr_\sigma[|X - \bar{f}| \geq a] \leq \bar{f} + a+ \frac{2}{p} \exp(- C da^2).$$

To minimize the expectation, we choose $a = \sqrt{\frac{\log \frac{2}{p} + \log d}{Cd}} = O\left(\sqrt{\frac{\log \frac{1}{p} + \log d}{d}}\right) .$ This choice of $a$ ensures that $\frac{2}{p} \exp(-Cd a^2) = \frac{2}{d} \leq \sqrt{\frac{\log d}{d}} \leq a$ for sufficiently large $d$. Hence,

$$\mathbb{E}_\sigma[f] \leq \bar{f} + O(a) = \mathbb{E}_\mu[f] + O\left(\sqrt{\frac{\log \frac{1}{p} + \log d}{d}}\right).$$ 
\end{proof}

We also require the following observation about the Haar measure.

\begin{lemma}
Let $\ket{\theta}$ be a Haar random state in $\C^d$. For every density matrix $\rho$ on $\C^d \otimes \C^d$, we have
$$\E_{\ket{\theta}}[\Tr(\ketbra{\theta}{\theta}^{\otimes 2} \rho)] \leq \frac{2}{d(d+1)}.$$
\label{lem:haar_symmetric_qcma}.
\end{lemma}

\begin{proof}
Since $\E[\ketbra{\theta}{\theta}^{\otimes 2}] = \frac{2 \Pi}{d(d+1)}$ where $\Pi$ is the projector onto the symmetric subspace of $\C^d \otimes \C^d$, then for any density matrix $\rho:$

$$\E_{\ket{\theta}}[\Tr(\ketbra{\theta}{\theta}^{\otimes 2} \rho)] = \frac{2}{d(d+1)} \Tr(\Pi \rho) \leq  \frac{2}{d(d+1)}.$$
\end{proof}

We are now ready to prove the lower bound. In fact, we will prove the lower bound on the average case version of the Entangled Subspace problem, which is the Planted Product State problem, introduced in the previous sections. In this section, we modify the definition of the problem to ensure that the planted product state is always a symmetric state $\ket{\theta}^{\otimes 2}$. However, the modified problem is clearly also a special case of the Entangled Subspace problem.

\begin{theorem}
Any quantum algorithm solving the Planted Product State problem using $T$ queries and an $m$-bit classical witness must use $$T \geq \Omega\left(\sqrt[4]{\frac{d}{m + \log d}}\right)$$ queries to the oracle.
\label{thm:qcma_lower_bound}
\end{theorem}

\begin{proof}
We apply the hybrid method variant introduced in \cite{aaronson2007quantum}. Let $O_1$ be the entangled oracle and $O_2 = O_1 - 2 \ketbra{\theta}{\theta}^{\otimes 2}$ be the oracle with a hidden product state $\ket{\theta}^{\otimes 2}$, given $\ket{\theta} \in \mathbb{C}^d$.  

Suppose we have a quantum algorithm $A$ that solves the Planted Product State problem with $T$ queries with the help of an $m$-bit classical witness. For each $\ket{\theta},$ fix the string $w$ that maximizes the probability that algorithm accepts $O_2$. Let $S(w) \subseteq S^d$ be the set of states associated with witness string $w$. Since $S(w)$ form a partition, then there must be one set $S(w^*)$ with measure at least $\frac{1}{2^m}$. 

Let $\sigma$ be the uniform measure over $S(w^*)$. Hence, fix $w^*$ as the witness in the algorithm, and choose $O_2$ where $\ket{\theta}$ is selected from $\sigma$. We claim that the algorithm still requires a large number of queries $T$ to distinguish between oracles $O_1$ and $O_2$ in this case. 
To establish the lower bound, let $\ket{\psi_t}$ be the result of the algorithm $A$ with oracle $O_2$ applied $t$ times followed by oracle $O_1$ applied $T - t$ times. By \cite{aaronson2007quantum}, We can bound the difference in Euclidean norm between successive hybrids by:

$$\|\ket{\psi_{t+1}} - \ket{\psi_t}\|_2 \leq \sqrt{\Tr((O_1 - O_2)^* (O_1 - O_2) \rho_t)} = 2 \sqrt{\Tr(\ketbra{\theta}{\theta}^{\otimes 2} \rho_t)},$$

where $\rho_t$ is the marginal state of the query register before the $t^{th}$ query since $O_1 - O_2 = 2 \ketbra{\theta}{\theta}^{\otimes 2}$. Hence, the Cauchy-Schwartz inequality implies that over a randomly selected $\ket{\theta}$ from $\sigma$:

$$\E_{\sigma}[\|\ket{\psi_{t+1}} - \ket{\psi_t}\|_2] \leq 2 \sqrt{\E_{\sigma}[\Tr(\ketbra{\theta}{\theta}^{\otimes 2}  \rho_t)]}.$$

Since $\sigma$ is $2^{-m}$-uniform, and the function $f(\ket{\psi}) = \Tr(\ketbra{\psi}{\psi}^{\otimes 2} \rho)$ is a non-negative, bounded, Lipschitz function, then we can bound by \Cref{lem:p-uniform-lemma} and \Cref{lem:haar_symmetric_qcma} that:

\begin{equation*}
\begin{split} 
\E_{\sigma}[\Tr(\ketbra{\theta}{\theta}^{\otimes 2} \rho_t)] &\leq \E_{\mu}[\Tr(\ketbra{\theta}{\theta}^{\otimes 2} \rho_t)] +  O\left(\sqrt{\frac{m + \log d}{d}}\right)  \\  & \leq \frac{2}{d(d+1)} + O\left(\sqrt{\frac{m + \log d}{d}}\right) \\ &\leq O\left(\sqrt{\frac{m + \log d}{d}}\right),
\end{split} 
\end{equation*}

since $\frac{2}{d^2} \leq \sqrt{\frac{\log d}{d}}$ for sufficiently large $d$. Hence,

$$\E_{\sigma}[\|\ket{\psi_{t+1}} - \ket{\psi_t}\|_2] \leq 2 \sqrt{\E_{\sigma}[\Tr(\ketbra{\theta}{\theta}^{\otimes 2} \rho_t)]} \leq O\left(\sqrt[4]{\frac{m + \log d}{d}}\right).$$

Hence, if $\ket{\psi_0}$ was the final state where all oracle calls were to $O_1,$ and $\ket{\psi_T}$ was the final state where all oracle calls were to $O_2,$ then the triangle inequality implies that

$$\E_{\sigma}[\|\ket{\psi_T} - \ket{\psi_0}\|_2] \leq O\left(T \sqrt[4]{\frac{m + \log d}{d}}\right).$$

If the algorithm $A$ correctly distinguishes between the two cases, then $\E_{\sigma}[\|\ket{\psi_T} - \ket{\psi_0}\|_2]  = \Omega(1).$ Hence the number of queries $T$ satisfies $$T \geq \Omega\left(\sqrt[4]{\frac{d}{m + \log d}}\right)$$.
\end{proof}

In particular, this bound shows that a polynomial sized classical witness is not sufficient to help a quantum verifier solve the Entangled Subspace Problem efficiently, since any quantum verifier that solves the Entangled Subspace problem can also be used to solve the Planted Product State problem.

\section{Open Questions}

We end by describing some open problems and future directions. 

\paragraph{Strong $\QMA$ Lower Bounds for the Entangled Subspace Problem.} Can one show that any $\QMA$ tester for the Entangled Subspace problem requires either a superpolynomial number of queries, or a superpolynomial sized witness? This would yield a (quantum) oracle separation between $\QMA$ and $\QMAtwo$, and in particular would rule out the existence of so-called ``disentanglers''~\cite{aaronson2008power}.

\paragraph{Better Query Upper Bounds.} Are the bounds proven using the generalized polynomial method tight? In particular, the following gaps remain:
    
    \begin{itemize}

        \item We have shown that there is a $O(\frac{t \sqrt{d}}{\epsilon})$ upper bound and a $\Omega(\max(\frac{t}{\epsilon}, \sqrt{d}))$ lower bound in the $\BQP$ setting for the recurrence problem and used this bound to prove a similar lower bound in the $\QMA$ setting. Is there a better lower or upper bound in either the $\BQP$ or $\QMA$ settings? However, a more sophisticated symmetrization technique may be required to improve the lower bound. 
        \item We expect the $\BQP$ lower bound in Theorem \ref{thm:entanglement_entropy_bound} for the entanglement entropy can be improved by using a more creative application of the polynomial method. %
    \end{itemize}

\paragraph{Improving \Cref{thm:counterexample}.}
Is the counterexample of  \Cref{thm:counterexample} tight, in the sense that there are no examples that fool the verifier in dimensions 2, 3, 4, or 5? Otherwise, if there was an example that fools the verifier in dimension 2, this would give additional evidence that the Entangled Subspace problem in low dimensions is already hard for $\QMA.$ %

\paragraph{Other Applications of the Generalized Polynomial Method.} What are other applications of the generalized polynomial method? For instance, Procesi \cite{PROCESI1976306} has characterized the invariants of matrix tuples under conjugation by the general linear, unitary, orthogonal, and symplectic groups. Are there natural problems in quantum query complexity that display other, non-unitary symmetries? 

\paragraph{Applications of Weingarten Calculus.} Can the Weingarten calculus techniques introduced in Section \ref{sec:genpoly} be used to prove lower bounds on unitary property testing problems? In particular, could it be used to prove degree lower bounds for acceptance probability polynomials? 

\paragraph{A Generalized Dual Polynomial Method?} A line of works established tight quantum query lower bounds on classical problems by employing a method of \emph{dual polynomials}~\cite{sherstov2013intersection,spalek2008dual,bun2018polynomial}. The goal of this method is to prove degree lower bounds of acceptance probability polynomials, but instead of symmetrizing the polynomials to obtain a polynomial of one or two variables, one instead takes advantage of \emph{linear programming duality} to prove the degree lower bounds; this involves constructing objects known as dual polynomials. A natural question would be to investigate whether the method of dual polynomials can be extended to prove query lower bounds for unitary property testing.

\appendix

\bibliography{references}
\bibliographystyle{alpha}

\section{A Generalized Product Test Analysis}
\label{sec:product-test}
\label{sec;product_test_analysis}

Our goal in this section is to prove \Cref{thm:general_qma_k_bound}. Let $\ket{\psi} \in \C^d \otimes \C^d$ be a bipartite state. Recall that the $k$-copy product test uses as input $k$-copies of the state $\ket{\psi}^{\otimes k}$, where each copy is on registers $A_i$ and $B_i$, and then performs the measurement $\{P = \Pi_A \otimes \Pi_B, I - P\},$ where $\Pi_A$ is the projection onto the symmetric subspace among registers $A_1, \dots, A_k$ and $\Pi_B$ is the projection onto the symmetric subspace among registers $B_1, \dots, B_k.$ Recall from \Cref{thm:product_test} that we defined

$$\omega_{\ket{\psi}} = \max_{\ket{\phi_1}, \ket{\phi_2}} \{\lvert \braket{\psi \vert \phi_1 \otimes \phi_2} \rvert^2, \ket{\phi_1}, \ket{\phi_2} \in \C^d\}$$

be the overlap with the closest product state. In particular, from \Cref{lem:entanglement_schmidt}, we have that $\omega_{\ket{\psi}} = \lambda_1$ where $\lambda_1$ is the largest eigenvalue of the reduced density matrix $\rho$ of $\ket{\psi} \bra{\psi}.$

We establish the following bound on the performance of the product test for any constant $k \geq 2$, using the techniques of \cite{soleimanifar2022testing}.

\generalqmakptbound*

Before proceeding with the proof, we fix some notation we will use throughout the rest of this section. Let $[n] = \{1, \dots, n\}$, $\lambda = (\lambda_1, \dots, \lambda_m)$ be a list of real numbers and $\alpha = (\alpha_1, \dots, \alpha_l)$ be a list of non-negative integers. Whenever well-defined, let $\lambda_\alpha = \prod_{i=1}^l \lambda_{\alpha_i}$ and $\lambda^\alpha = \prod_{i=1}^m \lambda_i^{\alpha_i}.$ Similarly if $a = \{a_1, \dots, a_n\}$ is a set of vectors, then $\ket{a}_\alpha = \ket{a_{\alpha_1}} \otimes \dots \otimes \ket{a_{\alpha_l}}.$ Next, for the list $\alpha$, let $n(\alpha)$ be the number of non-zero entries in $\alpha$ and $\text{type}(\alpha)$ be the list $(\beta_1, \dots, \beta_j)$ where $\beta_j$ is the number of times integer $j$ appears in the list $\alpha$. We will also write $\alpha \vdash k$ if $\alpha = (\alpha_1, \dots, \alpha_l)$ is a list with $\sum_{i=1}^l \alpha_i = k$, and $\binom{k}{\alpha}$ for the multinomial coefficient $\frac{k!}{\prod_{i=1}^l \alpha_i!}.$ 

We first establish the \textit{exact} probability the product test passes given the state $\ket{\psi}^{\otimes k}$ as input.

\begin{lemma}
Let $h_k(x_1, \dots, x_d) = \sum_{\substack{\beta_1 + \dots + \beta_d = k \\ \beta_i \geq 0}} \prod_{i=1}^d x_i^{\beta_i}$ be the homogenous symmetric polynomial of degree $k$ in $d$ variables. Then the probability that the $k$-copy product state passes when run on state $\ket{\psi}^{\otimes k}$ is equal to $h_k(\lambda_1, \dots, \lambda_d),$ where $\lambda_1, \dots, \lambda_d$ are the eigenvalues of the reduced density matrix $\rho$ of $\ket{\psi}\bra{\psi}.$
\label{lem:exact_product_test}
\end{lemma}

\begin{proof}
Let $k \geq 2$ be given. Consider the Schmidt decomposition of $\ket{\phi} = \sum_{i=1}^d \sqrt{\lambda_i} \ket{a_i} \ket{b_i}$ across the two subsystems, ordered so that $\lambda_j \geq \lambda_{j+1}$ for every index $1 \leq j < d.$ Letting $\Lambda = (\sqrt{\lambda_1}, \dots, \sqrt{\lambda_n})$ be the list of Schmidt coefficients, then

\begin{equation}
    \ket{\phi}^{\otimes k} = \sum_{\substack{\alpha = (\alpha_1, \dots, \alpha_k) \\ \alpha_i \in [d]}} \Lambda_\alpha \ket{a}_\alpha \ket{b}_\alpha.
    \label{eqn:expansion_1}
\end{equation}

Now, note that if $\alpha, \alpha'$ are two sequences with the same type $\beta$, then $\Lambda_\alpha = \Lambda_{\alpha'} = \Lambda^\beta$. So rewrite Equation \ref{eqn:expansion_1}, combining sequences of the same type.

\begin{equation}
    \ket{\phi}^{\otimes k} = \sum_{\substack{\beta = (\beta_1, \dots, \beta_d) \\ \beta_1 + \dots + \beta_d = k}} \Lambda^\beta \sum_{\substack{\alpha = (\alpha_1, \dots, \alpha_k) \\  \text{type}(\alpha) = \beta}}  \ket{a}_\alpha \ket{b}_\alpha.
    \label{eqn:expansion_2}
\end{equation}

We now consider the action of the projection $\Pi_A$ on the state $\ket{\phi}^{\otimes k}.$ Write the projection as $\Pi_A = \frac{1}{k!} \sum_{\sigma \in S_k} P_\sigma$ where $P_\sigma$ permutes vectors in $(\mathbb{C}^d)^k.$ Note that if $\alpha$ and $\alpha'$ have the same type $\beta = (\beta_1, \dots, \beta_d)$, there is some permutation $P_{\sigma}$ for which $P_\sigma \ket{a}_{\alpha} = \ket{a}_{\alpha'}$ and that there is a subgroup $S_\beta = S_{\beta_1} \times \dots \times S_{\beta_d} \leq S_k$ of permutations fixing $\ket{a}_\alpha$ of size $\prod_{i=1}^d \beta_i !.$ Hence, for every $\ket{a}_\alpha,$

$$\Pi_A \ket{a}_\alpha = \frac{1}{k!} \sum_{\text{cosets } \sigma S_\beta} |S_\beta| P_\sigma\ket{a}_\alpha = \frac{1}{\binom{k}{\beta}} \sum_{\substack{\text{all } \alpha' = (\alpha_1, \dots, \alpha_k) \\ \text{type}(\alpha) = \alpha'}} \ket{a}_{\alpha'}.$$

Therefore,

\begin{equation}
\begin{split} 
    (\Pi_A \otimes I) \ket{\phi}^{\otimes k} &=\sum_{\substack{\beta = (\beta_1, \dots, \beta_d) \\ \beta_1 + \dots + \beta_d = k}} \frac{\Lambda^\beta}{\binom{k}{\beta}} \sum_{\substack{\alpha = (\alpha_1, \dots, \alpha_k) \\  \text{type}(\alpha) = \beta}} \sum_{\substack{\alpha' = (\alpha_1, \dots, \alpha_k) \\  \text{type}(\alpha') = \beta}}  \ket{a}_{\alpha'} \ket{b}_\alpha. \\ 
    &= \sum_{\substack{\beta = (\beta_1, \dots, \beta_d) \\ \beta_1 + \dots + \beta_d = k}} \Lambda^\beta \left( \frac{\displaystyle \sum_{\substack{\alpha = (\alpha_1, \dots, \alpha_k) \\  \text{type}(\alpha) = \beta}}  \ket{a}_{\alpha}}{\sqrt{\binom{k}{\beta}}} \right) \left( \frac{\displaystyle \sum_{\substack{\alpha = (\alpha_1, \dots, \alpha_k) \\  \text{type}(\alpha) = \beta}}  \ket{b}_{\alpha}}{\sqrt{\binom{k}{\beta}}} \right) 
\end{split}
\end{equation}

We will write $\ket{a}^\beta = \left( \frac{\displaystyle \sum_{\substack{\alpha = (\alpha_1, \dots, \alpha_k) \\  \text{type}(\alpha) = \beta}}  \ket{a}_{\alpha}}{\sqrt{\binom{k}{\beta}}} \right)$ and note that the set of all $\{\ket{a}^\beta\}$ over all $\beta \vdash k$ is an orthonormal basis for the symmetric subspace $(\text{Sym}^k \mathbb{C}^d).$ Hence the  probability that the  product test applied to the first subsystem passes is $\mu = \sum_{\beta \vdash k} (\Lambda^\beta)^2,$ and conditioned on this, the post-measurement mixed state is $\ket{a}^\beta \ket{b}^\beta$ with probability $\frac{(\Lambda^\beta)^2}{\mu}$. At this point, the second subsystem becomes a symmetric state, and hence the projection onto the $B$ registers passes with probability one. Hence, we obtain that the probability that the product test passes is

\begin{equation} 
\mu = \sum_{\beta \vdash k} (\Lambda^\beta)^2 = \sum_{\substack{\beta_1 + \dots + \beta_d = k \\ \beta_i \geq 0}} \prod_{i=1}^d \lambda_i^{\beta_i} = h_k(\lambda_1, \dots, \lambda_d).
\label{eqn:prod_test_equation}
\end{equation}
\end{proof}

\begin{proof}[Proof of \Cref{thm:general_qma_k_bound}]
Firstly, we claim that 

\begin{equation} 
\frac{2}{k+1} \lambda_1^k + \sum_{\beta \vdash k, \beta_1 < k} (\Lambda^\beta)^2 \leq \frac{2}{k+1}.
\label{eqn:bound_1}
\end{equation}

 From \Cref{eqn:expansion_2}, we observe that

\begin{equation}
    1 = \sum_{\beta \vdash k} \binom{k}{\beta} (\Lambda^\beta)^2 = \lambda_1^k +  \sum_{\beta \vdash k. \beta_1 < k} \binom{k}{\beta} (\Lambda^\beta)^2 . 
\end{equation}

Therefore, we can divide the sum into three cases depending on the length $n(\beta)$ and the the value of $\beta_1,$

\begin{equation}
\begin{split} 
    &\frac{2}{k+1} \lambda_1^k + \sum_{\beta \vdash k, \beta_1 < k} (\Lambda^\beta)^2 = \frac{2}{k+1} \left[1 - \sum_{\beta \vdash k, \beta_1 < k} \binom{k}{\beta} (\Lambda^\beta)^2 \right] + \sum_{\beta \vdash k, \beta_1 < k} (\Lambda^\beta)^2 \\  &=  \frac{2}{k+1} + \sum_{\beta \vdash k, \beta_1 < k} \left[1 - \frac{2}{k+1} \binom{k}{\beta}\right] (\Lambda^\beta)^2  \\ 
    &= \frac{2}{k+1} + \sum_{n(\beta) = 1, \beta_1 < k}  \left[1 - \frac{2}{k+1} \binom{k}{\beta}\right] (\Lambda^\beta)^2 + \sum_{n(\beta) = 2, \beta_1 = k - 1}  \left[1 - \frac{2}{k+1} \binom{k}{\beta}\right] (\Lambda^\beta)^2 \\ &+ \sum_{\beta \vdash k, \beta_1 < k-1, n(\beta) \geq 2} \left[1 - \frac{2}{k+1} \binom{k}{\beta}\right ] (\Lambda^\beta)^2.
    \label{eqn:want_to_be_negative}
\end{split} 
\end{equation}

Observe firstly that if $n(\beta) \geq 2$, $\binom{k}{\beta} \geq \binom{k}{\beta_i} \geq k$  since $\beta_i \leq k-1$ is the largest entry in the list $\beta.$ Hence in all these cases $\left[1 - \frac{2}{k+1} \binom{k}{\beta}\right ] (\Lambda^\beta)^2 \leq 0$ since $1 - \frac{2}{k+1} \binom{k}{\beta} < 0.$ This shows that the third case is always negative.

Next, to bound the first and second cases:

\begin{equation}
\begin{split}
& \sum_{n(\beta) = 1, \beta_1 < k}  \left[1 - \frac{2}{k+1} \binom{k}{\beta}\right] (\Lambda^\beta)^2 + \sum_{n(\beta) = 2, \beta_1 = k - 1}  \left[1 - \frac{2}{k+1} \binom{k}{\beta}\right] (\Lambda^\beta)^2 \\ 
&= (1 - \frac{2}{k+1}) \sum_{i=2}^d \lambda_i^k + [1 - \frac{2k}{k+1}] \sum_{i=2}^d \lambda_1^{k-1} \lambda_i  \\
&= \frac{k-1}{k+1} \sum_{i=2}^d \lambda_i^k - \frac{k-1}{k+1} \sum_{i=2}^d \lambda_1^{k-1} \lambda_i 
= \frac{k-1}{k+1} [\sum_{i=2}^d (\lambda_i (\lambda_i^{k-1} - \lambda_1^{k-1})] \leq 0
\end{split}
\end{equation}

since each $\lambda_i$ is positive and $\lambda_1$ is the greatest of all of the Schmidt coefficients. Therefore, we have established \Cref{eqn:bound_1} since we have shown all of the sums in \Cref{eqn:want_to_be_negative} are negative. 

\Cref{eqn:bound_1} and \Cref{lem:exact_product_test} implies our bound since the probability that the product test passes is

\begin{equation}
\begin{split} 
    h_k(\lambda_1, \dots, \lambda_d) &= \sum_{\beta \vdash k} (\Lambda^\beta)^2 = \lambda_1^k + \sum_{\beta \vdash k, \beta_1 < k} (\Lambda^\beta)^2 \\ &\leq \lambda_1^k + \frac{2}{k+1} - \frac{2}{k+1} \lambda_1^k = \frac{k-1}{k+1} \lambda_1^k + \frac{2}{k+1} = \frac{k-1}{k+1} \omega_{\ket{\psi}}^k + \frac{2}{k+1}.
\end{split}
\end{equation}

where we have used the fact that for bipartite states, $\omega_{\ket{\psi}} = \lambda_1$ from \Cref{lem:entanglement_schmidt}. 
\end{proof} 
\section{A $\mathsf{SymQMA}$ Verifier for the Entangled Subspace Problem}
\label{sec:symqma_entangledsubspace}
We now ready to apply \Cref{thm:general_qma_k_bound} to prove \Cref{thm:general_qma_k_verifier}.

\generalentangledverifier*

The verifier $V_{k+1}$ has $k + 1$ proof registers $\ket{\psi}^{\otimes k+1}$. The $k$-copy product test is performed on registers 1 to $k$. Finally, the circuit performs a controlled $U$ operation on the $k+1$(st) proof state. The verifier accepts if and only if the product test passes and the ancilla qubit for the controlled $U$ measures to be $1$. The verifier $V_3$ is depicted in \Cref{figure:product_test}.

\begin{figure}[ht]
         \centering
         \mbox{
         \Qcircuit @C=1.4em @R=1.4em {
                 & \lstick{\ket{\psi_1}_{A_1 B_1}}  & \multigate{1}{\mathrm{Swap}_{A_1 A_2}} & \multigate{1}{\mathrm{Swap}_{B_1 B_2}} & \qw &  \qw & \qw \\
                 & \lstick{\ket{\psi_2}_{A_2 B_2}}  & \ghost{\mathrm{Swap}_{A_1 A_2}} & \ghost{\mathrm{Swap}_{B_1 B_2}} & \qw & \qw & \qw \\
                 & \lstick{\ket{\psi_3}_{A_3 B_3}} & \qw & \qw & \gate{U} & \qw & \qw \\ 
                 & \lstick{\ket{+}_1} & \ctrl{-2} & \qw  & \qw & \gate{H} & \meter  \\
                 & \lstick{\ket{+}_2} & \qw & \ctrl{-3} & \qw  & \gate{H} & \meter \\
                 & \lstick{\ket{+}_3} & \qw & \qw & \ctrl{-3}  & \gate{H} & \meter 
         }    }
         \caption{$\mathsf{SymQMA(3)}$ verifier}
         \label{figure:product_test}
\end{figure}
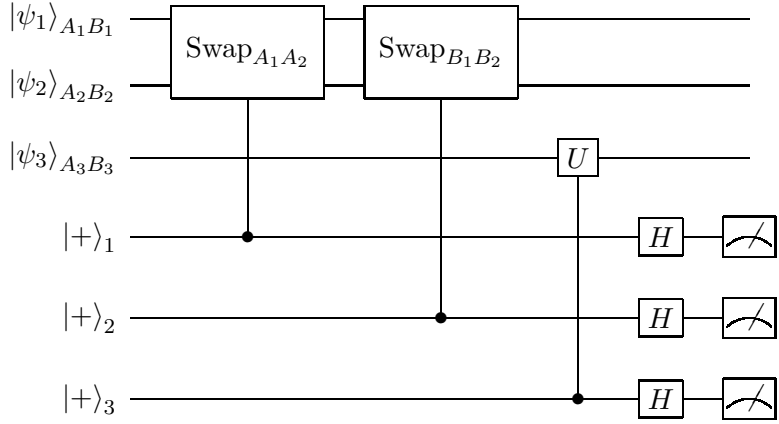

\begin{proof}
By assumption there exists a constant $\epsilon$ such that $b^2 - a^2 = \epsilon > 0$. Choose $k$ sufficiently large so that $\frac{k-1}{k+1} (1 - \frac{\epsilon^2}{4})^k + \frac{2}{k+1} \leq 1 - b^2.$ Such $k$ exists if $\epsilon > 0$ and $0 \leq b < 1$. We claim that the verifier $V_{k+1}$ suffices to distinguishes between the two cases.

Suppose we are given a \emph{yes} instance $U$ of the Entangled Subspace problem. This means that $U$ encodes a subspace $S$ containing a state $\ket{\psi} \in \C^d \otimes \C^d$ that is $a$-close to a product state $\ket{\varphi} \otimes \ket{\xi}$ in trace distance. Suppose we run the verifier on the following proof state: $\ket{\theta}^{\otimes (k+1)}$ where $\ket{\theta} = \ket{\varphi} \otimes \ket{\xi}$. Clearly, this proof will pass the product test with probability $1$ for any $k \geq 2$. Furthermore, %
the acceptance probability of the subspace membership test is $\Big \| \Pi \Big ( \ket{\varphi} \otimes \ket{\xi} \Big ) \Big \|^2 = |\braket{\psi \vert \varphi \otimes \xi}|^2 \geq 1 - a^2$, by assumption that $S$ contains a state that is $a$-close to product. %

Now suppose $U$ is a \emph{no} instance, and suppose for contradiction there exists a proof $\ket{\psi}^{\otimes (k + 1)}$ that is accepted by the verifier with probability greater than $1 - \frac{\epsilon}{2} - a^2$. Then in particular this means that the product test with $\ket{\psi}^{\otimes k}$ accepts with probability at least $1 - \frac{\epsilon}{2} - a^2$. Hence, by \Cref{thm:general_qma_k_bound}, if the witness causes the verifier to accept with probability at least $1 - \frac{\epsilon}{2} - a^2,$ then by the choice of $k$, there exists a product state $\ket{\theta} = \ket{\varphi} \otimes \ket{\xi}$ such that $|\braket{\phi \vert \theta}|^2 \geq 1 - \frac{\epsilon^2}{4}$ since $b^2 - a^2 = \epsilon$ implies that $1 - \frac{\epsilon}{2} - a^2 \geq 1 - b^2.$

Let $P_S$ be the projector onto the hidden subspace $S$. Since the definition of trace distance implies that:

$$|\Tr(P_S (\ket{\psi} \bra{\psi} - \ket{\theta} \bra{\theta}))| \leq \|\ket{\phi} \bra{\phi} - \ket{\theta} \bra{\theta}\| \leq \sqrt{1 - (1 - \frac{\epsilon^2}{4})} = \frac{\epsilon}{2}.$$

Then since $\Tr(P_S \ket{\psi} \bra{\psi}) > 1 - \frac{\epsilon}{2} - a^2,$ by assumption we have

$$\Tr(P_S \ket{\theta} \bra{\theta}) > 1 - 
\epsilon - a^2 = 1 - (\epsilon + a^2) = 1 - b^2.$$

which is a contradiction since $S$ was $b$-completely entangled. Therefore, the verifier must accept with probability at most $1 - \frac{\epsilon}{2} - a^2$ in the no case. Hence, there is a gap of at least $1 - a^2 - (1 - \frac{\epsilon}{2} - a^2) = \frac{\epsilon}{2}$ in distinguishing between the yes and no cases. Thus, there is a constant gap between the acceptance probabilities of the \emph{yes} and \emph{no} instances, showing that the $(a,b)$-Entangled Subspace problem is in $\textsf{SymQMA(k+1)}$. 
\end{proof}

\section{Characterization of Product Test Witnesses}
\label{sec:fool}
\label{sec:product_test_witnesses}

In this section, we characterize all witness states that ``fool'' the two-copy product test. This characterization enables us to construct additional counterexamples that fool the two-copy product test verifier, in addition to the one example presented in \Cref{subsec:qma_example}. To do this, we first translate our problem in the language of representation theory. 

\subsection{Background} 

To construct the necessary states and subspaces, we use the theory of Schur-Weyl duality, which states that there is an isomorphism

$$(\mathbb{C}^d)^{\otimes k} \cong \bigoplus_{\lambda \vdash k} (V_\lambda \otimes W_\lambda),$$

where $\lambda$ is a partition of $k$,  $V_\lambda$ is a Specht module (an irreducible representation of the symmetric group $S_k$) and $W_\lambda$ is a Weyl module (an irreducible representation of $GL(d).$) We focus on the action of $S_k$ on $(\mathbb{C}^d)^{\otimes k},$ which acts by permuting registers, so we only need to consider the structure of $V_\lambda$. 

There is a basis of $V_\lambda$ given explicitly by tableau states. Recall that a standard Young tableau (SYT) of shape $\lambda$ is a filling of a diagram of shape $\lambda$ in such a way that rows and columns increase strictly. A semistandard Young tableau (SSYT) is a filling such that rows weakly increase and columns strictly increase. An SYT and SSYT of shape $(3,2)$ is shown in the figure below.

\begin{figure} 
\center 
\begin{ytableau}
1 & 2 & 5 \cr 
3 & 4
\end{ytableau}
\hspace{1cm}
\begin{ytableau} 
1 & 2 & 2 \cr 
3 & 5
\end{ytableau} 
\caption{SYT and SSYT of shape $(3,2)$}
\end{figure} 

Given an SYT $T$ with $n$ cells, let $P_T$ and $Q_T$ be the row and column preserving subgroups of $T$. Let $a_T = \sum_{\sigma \in P_T} \Pi_{\sigma}$ and $b_T = \sum_{\sigma \in Q_T} \text{sgn}(\sigma) \Pi_{\sigma}$ be the row symmeterization operator and column anti-symmeterization operators associated tableau $T$. The Young symmeterizer is defined to be the product $c_T = b_T a_T.$

The following theorem expresses Schur-Weyl theorem combinatorially.

\begin{theorem}[\cite{berget2009symmetries}, Appendix]
Let $\lambda$ be a partition of $k$ and let $(P, Q)$ be tableau of shape $\lambda$. Let $e_{P, Q}$ be the tensor product of the basis elements in $Q$ in the order given by $P$. Then the set of tensors $e_{P, Q} c_P$ where $Q$ is fixed and $P$ varies across all possible SYTs is a basis for $V_\lambda.$
\end{theorem}

We call each basis vector $e_{P,Q} c_P$ a tableau state.

Let $\Pi_{ij}$ be a swap operator among registers $i$ and $j$. We observe that a state $\ket{\psi} \in (\mathbb{C}^d)^{\otimes 4}$ passes the product test with probability one if and only if $\Pi_{13} \ket{\psi} = \ket{\psi}$ and $\Pi_{24} \ket{\psi} = \ket{\psi}.$ Let $G = \{e, (13), (24), (13)(24)\}$ be the subgroup of $S_4$ generated by these swap operators. Therefore, the condition that $\ket{\psi}$ passes the product test with probability one is equivalent to the condition that the one-dimensional subspace spanned by $\ket{\psi}$ is a trivial representation of the group $G$.

Therefore, the Schur-Weyl isomorphism gives us a way to classify all states that pass the product test with probability one. First we identify which representations $V_\lambda$ of $S_4$ restrict to the trivial representation of the group $G$. Then, we construct a basis of $G$-invariant states by using the tableau state construction. This allows us to identify the set of all states that fool the product test. 

\subsection{Calculations}

\begin{lemma}
Let $V_\lambda$ be an irreducible $S_4$-representation. Then $V_\lambda$ contains a trivial representation of $G$ if and only if $\lambda = (4), (3,1)$ or $(2,2).$ Furthermore, the restricted representation is a one-dimensional subspace of $V_\lambda.$
\end{lemma}

\begin{proof}
Let $\chi_\lambda$ be the character of the irreducible representation $V_\lambda$. Using the character table of $S_4$, computed on \cite[Page 13]{westrich2011youngs}, the sum $$\frac{1}{4} (\chi_\lambda(e) + \chi_\lambda((13)) + \chi_\lambda((24)) + \chi_\lambda((13)(24)))$$ 

is non-zero (and equal to one) if and only if $\lambda = (4), (3,1)$ or $(2,2)$.
\end{proof}

\begin{figure} 
\center 
\begin{ytableau}
1 & 2 & 3 \cr 
4 
\end{ytableau}
\hspace{0.5cm}
\begin{ytableau} 
1 & 2 & 4 \cr 
3 
\end{ytableau} 
\hspace{0.5cm}
\begin{ytableau} 
1 & 3 & 4  \cr 
2
\end{ytableau} 
\hspace{0.5cm}
\begin{ytableau} 
a & b & c \cr 
d
\end{ytableau} 
\caption{SYT and SSYT of shape $(3,1)$}
\label{fig:tableau_1}
\end{figure} 

\begin{lemma}
Each copy of $V_4$ in $(\mathbb{C}^d)^{\otimes 4}$ is one-dimensional and symmetric under all permutations, given by the span of

\begin{equation} 
\begin{split} 
\ket{\psi_{abcd}^4} &= (\ket{ab} + \ket{ba})(\ket{cd} + \ket{dc}) +  (\ket{cd} + \ket{dc})(\ket{ab} + \ket{ba})  \\ 
&+ (\ket{ac} + \ket{ca})(\ket{bd} + \ket{db}) + (\ket{bd} + \ket{db})(\ket{ac} + \ket{ca}) \\
&+  (\ket{ad} + \ket{da})(\ket{bc} + \ket{cb}) + (\ket{bc} + \ket{cb})(\ket{ad} + \ket{da}). 
\end{split}
\end{equation}

where $a \leq b \leq c \leq d.$
\end{lemma}

\begin{proof}
Follows from the definition of Young symmetrizer. 
\end{proof}

\begin{lemma}
The unique $G$-invariant subspace in each copy of $V_{3,1}$ in $(\mathbb{C}^d)^{\otimes 4}$ is spanned by the vector

\begin{equation} 
\begin{split} 
\ket{\psi_{abcd}^{(3,1)}} &= (\ket{ab} + \ket{ba})(\ket{cd} - \ket{dc}) +  (\ket{ac} + \ket{ca})(\ket{bd} - \ket{db}) + (\ket{bc} + \ket{cb})(\ket{ad} - \ket{da})  \\ 
&+ (\ket{cd} - \ket{dc})(\ket{ab} + \ket{ba}) +  (\ket{bd} - \ket{db})(\ket{ac} + \ket{ca}) + (\ket{ad} - \ket{da}) (\ket{bc} + \ket{cb})
\end{split}
\end{equation} 

for some SSYT of shape $(3,1)$ containing elements $abcd.$
\end{lemma}

\begin{proof}
Let $T$ be the leftmost tableau in Figure \ref{fig:tableau_1} and let $c_T = b_T a_T$ be the corresponding Young symmeterizer. Since the row-preserving subgroup consists of all permutations of $\{1,2,3\}$ and the column-preserving subgroup consists of all permutations of $\{1,4\},$ then

$$a_T \ket{abcd} = \ket{abcd} + \ket{bacd} + \ket{acbd} + \ket{cabd} + \ket{bcad} + \ket{cbad},$$

and therefore

\begin{equation*}
\begin{split} 
c_T \ket{abcd} &= b_T a_T \ket{abcd} \\ 
&= \ket{abcd} - \ket{dbca} + \ket{bacd} - \ket{dacb} + \ket{acbd} - \ket{dcba} \\ &+ \ket{cabd} - \ket{dabc} + \ket{bcad} - \ket{dcab} + \ket{cbad} - \ket{dbac}.
\end{split} 
\end{equation*} 

Call this vector $v_1.$ To obtain the other two vectors in $V_{3,1}$, we compute 
\begin{equation*}
\begin{split} 
v_2 = c_T \ket{abdc} &= \ket{abdc} - \ket{dbac} + \ket{acdb} - \ket{dcab} + \ket{badc} - \ket{dabc} \\ &+ \ket{bcda} - \ket{dcba} + \ket{cadb} - \ket{dacb} + \ket{cbda} - \ket{dbca}.
\end{split} 
\end{equation*} 

using the second tableau in Figure \ref{fig:tableau_1}, and finally

\begin{equation*}
\begin{split} 
v_3 = c_T \ket{adbc} &= \ket{adbc} - \ket{dabc} + \ket{adcb} - \ket{dacb} + \ket{bdac} - \ket{dbac} \\ &+ \ket{bdca} - \ket{dbca} + \ket{cdab} - \ket{dcab} + \ket{cdba} - \ket{dcba}.
\end{split} 
\end{equation*} 

using the third and final tableau.

In the basis $\{v_1, v_2, v_3\}$ of $V_{3,1},$ the representation matrices of permutations $(13)$ and $(24)$ are

$$\begin{bmatrix} 1 & 0 & 0 \\ -1 & -1 & -1 \\ 0 & 0 & 1 \end{bmatrix} \quad  \begin{bmatrix} 0 & 0 & 1 \\ 0 & 1 & 0 \\ 1 & 0 & 0 \end{bmatrix}$$

respectively. One can verify this directly or using the computations in \cite{westrich2011youngs}. Therefore, a vector that is invariant under simultaneous $(13)$ and $(24)$ permutations is a common eigenvector of both matrices of eigenvalue 1, which is $\begin{bmatrix} 1 \\ -1 \\ 1 \end{bmatrix}.$ Hence setting $\ket{\psi_{abcd}^{(3,1)}} = v_1 - v_2 + v_3$ completes the proof. 
\end{proof}

\begin{figure} 
\center 
\begin{ytableau}
1 & 2 \cr 
3 & 4 
\end{ytableau}
\hspace{0.5cm}
\begin{ytableau} 
1 & 3  \cr 
2 & 4 
\end{ytableau} 
\hspace{0.5cm}
\begin{ytableau} 
a & b   \cr 
c & d
\end{ytableau} 
\caption{SYT and SSYT of shape $(2,2)$}
\label{fig:tableau_2}
\end{figure} 

\begin{lemma}
The unique $G$-invariant subspace in each copy of $V_{2,2}$ is spanned by the vector

\begin{equation} 
\begin{split} 
\ket{\psi_{abcd}^{(2,2)}} &= (\ket{ab} + \ket{ba})(\ket{cd} + \ket{dc}) + (\ket{cd} + \ket{dc})(\ket{ab} + \ket{ba}) \\ 
&+ \ket{acdb} + \ket{cabd} - 2\ket{acbd} - 2\ket{cadb}  \\ 
&+ \ket{bcda} + \ket{cbad} - 2\ket{bcad} - 2\ket{cbda} \\ 
&+ \ket{adcb} + \ket{dabc} - 2\ket{adbc} - 2\ket{dacb} \\ 
&+ \ket{bdca} + \ket{dbac} - 2\ket{bdac} - 2\ket{dbca}. 
\end{split}
\end{equation} 

for some SSYT of shape $(2,2)$ containing elements $abcd.$
\end{lemma}

\begin{proof}
The SYT of shape $(2,2)$ are listed in Figure \ref{fig:tableau_2}. Using the first tableau with row-preserving subgroup is $\{e, (12), (34), (12)(34)\}$ and the column-preserving subgroup is $\{e, (13), (24), (13)(24)\}$, we have
\begin{equation*}
\begin{split} 
    v_1 = c_T \ket{abcd} &= \ket{abcd} - \ket{cbad} - \ket{adcb} + \ket{cdab} \\ 
    &+ \ket{bacd} - \ket{cabd} - \ket{bdca} + \ket{cdba} \\ 
    &+ \ket{abdc} - \ket{dbac} - \ket{acdb} + \ket{dcab} \\ 
    &+ \ket{badc} - \ket{dabc} - \ket{bcda} + \ket{dcba}. 
\end{split} 
\end{equation*}

Similarly, using the second tableau, we get that

\begin{equation*}
\begin{split} 
    v_2 = c_T \ket{acbd} &= \ket{acbd} - \ket{cabd} - \ket{acdb} + \ket{cadb} \\ 
    &+ \ket{bcad} - \ket{cbad} - \ket{bcda} + \ket{cbda} \\ 
    &+ \ket{adbc} - \ket{dabc} - \ket{adcb} + \ket{dacb} \\ 
    &+ \ket{bdac} - \ket{dbac} - \ket{bdca} + \ket{dbca}. 
\end{split} 
\end{equation*}

In the basis $\{v_1, v_2\}$ of $V_{2,2},$ the matrices for $(13)$ and $(24)$ are the same, and equal to $\begin{bmatrix} -1 & -1 \\ 0 & 1 \end{bmatrix}.$ An eigenvector of eigenvalue one is $\begin{bmatrix} 1 \\ -2 \end{bmatrix}. $ Setting $\ket{\psi_{abcd}^{(2,2)}} = v_1 - 2v_2$ completes the proof. 
\end{proof}

\begin{theorem}
There exists a basis for the $G$-invariant subspace $V$ of $(\mathbb{C}^d)^{\otimes 4}$ consisting where each $\ket{\psi} \in V$ is in the span of states $\ket{\psi_1} \otimes \ket{\psi_2},$ where $\ket{\psi_1}, \ket{\psi_2} \in (\mathbb{C}^d)^{\otimes 2}.$
\label{theorem:basis_lemma}
\end{theorem}

\begin{proof}
We have already computed that $\{\ket{\psi_{abcd}^4}, \ket{\psi_{abcd}^{(3,1)}}, \ket{\psi_{abcd}^{(2,2)}}\}$ is a basis for $V$ ranging over all tableau of the appropriate shape with fillings $abcd.$ Observe that states of the form $\ket{\psi_{abcd}^4}$ and  $\ket{\psi_{abcd}^{(3,1)}}$ already have the desired property. Now defining $$\ket{\psi_{abcd}^{(2,2)}}' = \frac{1}{3} \ket{\psi_{abcd}^4} + \frac{2}{3} \ket{\psi_{abcd}^{(2,2)}},$$ we can compute that

\begin{equation*} 
\begin{split} 
\ket{\psi_{abcd}^{(2,2)}}' &= (\ket{ab} + \ket{ba})(\ket{cd} + \ket{dc}) + (\ket{cd} + \ket{dc})(\ket{ab} + \ket{ba}) \\ 
&+ \ket{acdb} + \ket{cabd} - \ket{acbd} - \ket{cadb}  \\ 
&+ \ket{bcda} + \ket{cbad} - \ket{bcad} - \ket{cbda} \\ 
&+ \ket{adcb} + \ket{dabc} - \ket{adbc} - \ket{dacb} \\ 
&+ \ket{bdca} + \ket{dbac} - \ket{bdac} - \ket{dbca} \\ 
&= (\ket{ab} + \ket{ba})(\ket{cd} + \ket{dc}) + (\ket{cd} + \ket{dc})(\ket{ab} + \ket{ba}) \\
&+ (\ket{ac} - \ket{ca})(\ket{db} - \ket{bd}) + (\ket{bd} - \ket{db})(\ket{ca} - \ket{ac}) \\ 
&+ (\ket{bc} - \ket{cb})(\ket{da} - \ket{ad}) + (\ket{ad} - \ket{da})(\ket{cb} - \ket{bc})
\end{split}
\end{equation*} 

has the desired property, and $\{\ket{\psi_{abcd}^4}, \ket{\psi_{abcd}^{(3,1)}}, \ket{\psi_{abcd}^{(2,2)}}'\}$ ranging over all tableau remains a basis of $V$. 
\end{proof}

Therefore, \Cref{theorem:basis_lemma} enables us to construct subspaces $S$ where there exists an witness on which the product test part of the verifier passes with probability one. This can be done, for instance, by taking $\ket{\psi} \in V$ in \Cref{theorem:basis_lemma} and letting $S$ be supported on the partial trace of $\ket{\psi}$ on the first two registers. In particular, $S$ could be an completed entangled subspace, for instance by considering the symmetric subspaces as in \Cref{thm:counterexample}.

\end{document}